\documentclass{statsoc}
\usepackage[a4paper]{geometry}  

\usepackage{listings}
\usepackage{amsmath}
\usepackage{amsthm}
\usepackage{tikz}
\usepackage{array}
\usepackage{mdwmath}
\usepackage{multirow}
\usepackage{mdwtab}
\usepackage{eqparbox}
\usepackage{multicol}
\usepackage{amsfonts}
\usepackage{tikz}
\usepackage{multirow,bigstrut,threeparttable}
\usepackage{array}
\usepackage{bbm}
\usepackage{epstopdf} 
\usepackage{mdwmath}
\usepackage{mdwtab}
\usepackage{eqparbox}
\usepackage{tikz}
\usepackage{latexsym}
\usepackage{amssymb}
\usepackage{bm}
\usepackage{amssymb}
\usepackage{graphicx}
\usepackage{mathrsfs}
\usepackage{epsfig}
\usepackage{psfrag}
\usepackage{setspace}
\usepackage[
            CJKbookmarks=true,
            bookmarksnumbered=true,
            bookmarksopen=true,
            colorlinks=true,
            citecolor=red,
            linkcolor=blue,
            anchorcolor=red,
            urlcolor=blue
            ]{hyperref}
\usepackage[ruled]{algorithm2e}
\usepackage{algpseudocode}
\usepackage{stfloats}
\usepackage{natbib}

\input xy
\xyoption{all}

\theoremstyle{plain}

\newtheorem{lemma}{Lemma}
\newtheorem{assumption}{Assumption}

\newtheorem{claim}{Claim}

\newtheorem{proposition}{Proposition}

\theoremstyle{definition}

\def\EE{\mathbb{E}}

\def\PP{\mathbb{P}}

\def\calB{\mathcal{B}}

\def\calF{\mathcal{F}}

\def\calR{\mathcal{R}}

\def\1{\mathbbm{1}}
\def\var{\mathsf{Var}}
\def\cov{\mathsf{Cov}}

\newcommand\independent{\protect\mathpalette{\protect\independenT}{\perp}}
\def\independenT#1#2{\mathrel{\rlap{$#1#2$}\mkern2mu{#1#2}}}
\DeclareMathOperator*{\argmax}{arg\,max}

\usepackage{booktabs}
\usepackage{adjustbox}
\usepackage{multirow}

\def \var {\mathsf{Var}}

\usepackage{xspace}

\def\independenT#1#2{\mathrel{\rlap{$#1#2$}\mkern2mu{#1#2}}}

\definecolor{myblue}{rgb}{.8, .8, 1}
\definecolor{mathblue}{rgb}{0.2472, 0.24, 0.6} 
\definecolor{mathred}{rgb}{0.6, 0.24, 0.442893}
\definecolor{mathyellow}{rgb}{0.6, 0.547014, 0.24}

\title[DINA]{Estimating Heterogeneous Treatment Effects for General Responses}
\author[Z. Gao]{Zijun Gao}
\address{Department of Statistics, Stanford University, Stanford, USA}
\email{zijungao@stanford.edu}
\author[Zijun Gao and Trevor Hastie]{Trevor Hastie}
\address{Department of Statistics and Department of Biomedical Data Science, Stanford University, Stanford, USA}

\begin{document}

\begin{abstract}
	Heterogeneous treatment effect models allow us to compare treatments at subgroup and individual levels, and are of increasing popularity in applications like personalized medicine, advertising, and education. In this talk, we first survey different causal estimands used in practice, which focus on estimating the difference in conditional means. We then propose DINA — the difference in natural parameters — to quantify heterogeneous treatment effect in exponential families and the Cox model. For binary outcomes and survival times, DINA is both convenient and more practical for modeling the influence of covariates on the treatment effect. Second, we introduce a meta-algorithm for DINA, which allows practitioners to use powerful off-the-shelf machine learning tools for the estimation of nuisance functions, and which is also statistically robust to errors in inaccurate nuisance function estimation. We demonstrate the efficacy of our method combined with various machine learning base-learners on simulated and real datasets.
\end{abstract}
\keywords{Heterogeneous treatment effect; Exponential family; Cox model}

\section{Introduction}
The potential outcome model \citep{rubin1974estimating} has received
wide attention \citep{rosenbaum2010design, imbens2015causal} in the
field of causal inference. Recent attention has focused on the
estimation of heterogeneous treatment effects (HTE), which allows the
treatment effect to depend on subject-specific features. 
In this paper, we investigate heterogeneous treatment effects instead of average treatment effects due to the following reasons.
\begin{enumerate}
	\item In applications like personalized medicine \citep{splawa1990application, low2016comparing, lesko2007personalized}, personalized education \citep{murphy2016handbook}, and personalized advertisements \citep{bennett2007netflix}, the target population is not the entire study population but the subset some patient belongs to, and thus the heterogeneous treatment effect is of more interest \citep{hernan2010causal}. 

	\item Marginal treatment effects can be obtained by marginalizing the heterogeneous counterparts \citep{daniel2021making,hu2021estimating}. The conditioning step can potentially adjust for observed confounders and yield less biased marginal estimators.
\end{enumerate}

For continuous responses, the difference in conditional means is
commonly used as an estimator of HTE \citep{powers2018some,
	wendling2018comparing}. However, for binary responses or count data,
no consensus of the estimand has been reached, and a variety of
objectives have been considered. For instance, for dichotomous responses,
conditional success probability differences, conditional success
probability ratios, and conditional odds ratios all have been
studied \citep{imbens2015causal,tian2012simple}. In this paper, we
propose to estimate a unified quantity --- the difference in natural
parameters (DINA) --- applicable for all types of responses from the
exponential family. 
The DINA estimand is appealing from several points of view compared to the difference in  conditional means.
\begin{enumerate}
	\item 
	Comparisons on the natural parameter scale are commonly adopted
	in practice. 
	DINA coincides with the conditional mean
	difference for continuous responses, the log of conditional odds
	ratio for binary responses, and the log of conditional mean ratio
	for count data. In the Cox model \citep{cox1972regression}, DINA
	corresponds to the log hazard ratio.
	In clinical trials with binary outcomes, the odds ratio is a popular indicator of diagnostic performance \citep{rothman2012epidemiology,glas2003diagnostic}.
	For rare diseases, the odds ratio is approximately equivalent to the relative risk, another commonly-used metric in epidemiology. 
	In addition, for survival outcomes, the hazard ratio is frequently reported which measures the chance of an event occurring in the treatment arm divided by that in the control arm.
	
	\item It is convenient to model the influence of covariates on the
	natural parameter scale. For binary responses or count data, the types of
	outcomes impose implicit constraints, such as zero-one or
	non-negative values, making the modeling on the original
	scale difficult. In contrast, the natural parameters are numbers on
	the real line, and easily accommodate various types of covariate dependence. 
	
	\item The difference in conditional means may exhibit uninteresting heterogeneity. Consider a vaccine example where before injections, the disease risk is $10\%$ among older people and $1\%$ among young people. It is not likely the absolute risk differences caused by some vaccine will be the same across age groups since the room for improvement is significantly different ($10\%$ versus $1\%$). Still, the relative risk may be constant, for example, both young and old are $80\%$ less likely to get infected.
\end{enumerate}

There are two critical challenges in the estimation of DINAs.
Like the difference in conditional means, the estimation of DINA could be biased due to confounders ---
covariates that influence both the potential outcomes and the
treatment assignment. 
Unlike the difference in conditional means, DINA estimators could potentially
suffer from the ``non-collapsibility'' issue
\citep{gail1984biased,greenland1999confounding,hernan2010causal} on the natural parameter scale.
As a result, even when there is no confounding, including a non-predictive covariate in the model will lead to a different estimator.
The distinction between confounding and non-collapsibility is discussed in detail by \citet{samuels1981matching,miettinen1981confounding,greenland1986identifiability}.

In this paper, we propose a DINA estimator that is robust to the
aforementioned confounding and non-collapsibility issues. The method
is motivated by Robinson's method \citep{robinson1988root} and R-learner \citep{nie2017quasi} proposed to deal with the conditional mean
difference. Like R-learner, our method consists of two steps:
\begin{enumerate}
	\item Estimation of nuisance functions using any flexible algorithm; 
	\item Estimation of the treatment effect with nuisance function estimators plugged in.
\end{enumerate}
The method is locally insensitive to the
misspecification of nuisance functions, and despite this
inaccuracy is still able to produce accurate DINA estimators.
By separating the estimation of nuisance functions from that of DINA, we can use  powerful machine-learning tools, such as random forests and neural networks, for the former task.

The organization of the paper is as follows. 
In Section~\ref{sec:background}, we formulate the problem, discuss the difficulties, and summarize related work. 
In Section~\ref{sec:Robinson}, we illustrate Robinson's method and R-learner that our proposal is built on.
In Section~\ref{sec:method}, we construct the DINA estimator for the exponential family and discuss its theoretical properties. 
In Section~\ref{sec:cox}, we extend the DINA estimator to the Cox model based on full likelihoods and partial likelihoods, respectively.
In Section~\ref{sec:simulation}, we assess the performance of our proposed DINA estimator on simulated datasets. 
In Section~\ref{sec:realData}, we apply the DINA estimator to the SPRINT dataset, evaluating its robustness to designs of experiments. 
In Section~\ref{sec:multiValued}, we briefly discuss the extension of our DINA estimator from single-level treatments to multi-level treatments. 
We conclude the paper with discussions in Section~\ref{sec:discussion}. 
All proofs are deferred to the appendix.

\section{Background}\label{sec:background}

\subsection{Problem formulation}\label{subsec:formulation}

We adopt the Neyman-Rubin potential outcome model. 
Each unit is associated with a covariate vector $X$, a treatment assignment indicator $W$, and two potential outcomes $Y(0)$, $Y(1)$. We observe the response $Y = Y(1)$ if the unit is under treatment, i.e., $W = 1$, and $Y = Y(0)$ if the unit is under control, i.e., $W = 0$. 
We make the standard assumptions in causal inference \citep{imbens2015causal}.

\begin{assumption}[Stable unit treatment value assumption]\label{assu:SUTVA}
	The potential outcomes for any unit do not depend on the treatments assigned to other units. 
\end{assumption}

\begin{assumption}[Unconfoundedness]\label{assu:unconfoundedness}
	The assignment mechanism does not depend on potential outcomes,
	\begin{align*}
		Y(1), Y(0) \independent W \mid X.
	\end{align*}
\end{assumption}

\begin{assumption}[Overlap]\label{assu:overlap}
	The probability of being treated, i.e., the propensity score $e(X) = \PP(W = 1 \mid X)$, takes value in $[\varepsilon, 1-\varepsilon]$ for some $\varepsilon > 0$.
\end{assumption}

To facilitate asymptotic analyses, we introduce the super-population model
\begin{align}
	\label{eq:model1}
	& \quad\quad~ X \stackrel{\text{i.i.d.}}{\sim} f_X, \\
	\label{eq:model2}
	&\quad W \mid X = x \stackrel{\text{ind.}}{\sim} \text{Ber}(e(x)), \\
	\label{eq:generalizedModel}
	\begin{split}
		\begin{cases}
			Y(1) \mid X = x \sim f_1(\cdot \mid x),  & \\
			Y(0) \mid X = x \sim f_0(\cdot \mid x),  &
		\end{cases}
	\end{split}
\end{align}
where $f_X$ denotes the distribution of covariate, $e(x)$ denotes the propensity score, and $f_1$ and $f_0$ denote the distributions of potential outcomes.

For survival analysis, we let $C(0)$, $C(1) \in [0,\infty]$ be the counterfactual censoring times and $C$ be the observed censoring time. Then $C = C(1)$ if $W=1$ and $C = C(0)$ if $W=0$. Let $\Delta := \1_{\{C \ge Y\}}$ be the observed censoring indicator. The observed responses are pairs $(Y^c := \min\{Y, C\}, \Delta)$. We make the following assumption on the counterfactual censoring times.

\begin{assumption}[Censoring mechanism]\label{assu:censor}
	The counterfactual censoring times are independent of the survival times given the covariates and the treatment assignment,
	\begin{align*}
		Y(0) &\independent C(0) \mid W, X,	\\
		Y(1) &\independent C(1) \mid W, X,
	\end{align*}
	and the counterfactual censoring times are unconfounded,
	\begin{align*}
		C(1), C(0) &\independent W \mid X.
	\end{align*}
\end{assumption}

Assumption~\ref{assu:censor} implies that the counterfactual censoring times do not directly depend on the treatment assignment or survival times.

Let $\lambda_0(y; x)$ be the conditional control hazard rate 
at time $y$, and similarly we define the conditional treatment hazard rate  $\lambda_1(y; x)$. We follow the Cox model \citep{cox1972regression} and make the  proportional hazards assumption.
\begin{assumption}[Proportional hazards]\label{assu:Cox}
	The hazard rate functions follow
	\begin{align*}
		\lambda_0(y; x) &= \lambda(y) e^{\eta_0(x)},\\
		\lambda_1(y; x) &= \lambda(y) e^{\eta_1(x)},
	\end{align*}
	where $\lambda(y)$ denotes the baseline hazard function and $e^{\eta_0(x)}$, $e^{\eta_1(x)}$ denote the exponential tilting functions.
\end{assumption}

In this paper, we aim to estimate the heterogeneous (conditional) treatment effects, denoted by $\tau(x)$. The exact form of $\tau(x)$ depends on the type of the response.
\begin{enumerate}
	\item For continuous data, we use the difference of the conditional means
	\begin{align}\label{eq:tauContinuous}
		\tau(x) := \EE[Y(1) \mid X = x] - \EE[Y(0) \mid X = x].
	\end{align}
	\item For binary responses, we use the difference of log conditional odds, i.e., the log of conditional odds ratio,
	\begin{align}\label{eq:tauBinary}
		\tau(x) := \log\left(\frac{\PP(Y(1) = 1 \mid X = x)}{\PP(Y(1) = 0 \mid X = x)}\right) - \log\left( \frac{\PP(Y(0) = 1 \mid X = x)}{\PP(Y(0) = 0\mid X = x)}\right).
	\end{align}
	\item For count data, we use the difference of the log of conditional means, i.e., the log of conditional mean ratio,
	\begin{align}\label{eq:tauCount}
		\tau(x) := \log\left(\EE[Y(1) \mid X = x]\right) - \log\left(\EE[Y(0) \mid X = x]\right).
	\end{align}
	\item For survival data, we use the difference of the log of conditional hazards, i.e., the log of hazards ratio,
	\begin{align}\label{eq:tauCox}
		\begin{split}
			\quad		\tau(y; x) 
			&:= \log\left(\lambda_1( y ; x)\right) - \log\left(\lambda_0(y; x)\right)\\
			&~= \log\left( \lim_{\delta \to 0^+} \frac{\PP(Y(1) \in [y,y+\delta]\mid X = x)}{\delta\PP(Y(1) \ge y\mid X = x)}\right) \\
			&\quad~- \log\left(\lim_{\delta \to 0^+} \frac{ \PP(Y(0) \in [y,y+\delta]\mid X = x)}{\delta\PP(Y(0)\ge y\mid X = x)}\right).
		\end{split}
	\end{align}
	Under Assumption~\ref{assu:Cox}, $\tau(y; x)$ does not depend on $y$. Without further specification, we will omit $y$ and use $\tau(x)$.
\end{enumerate}
If the responses indeed follow Gaussian, Bernoulli, Poisson distributions, or Cox model, then the above estimands are the differences in the treatment and control group natural parameter functions (DINA).
While the conditional means are usually supported on intervals for non-continuous responses, natural parameter functions often take values over the entire real axis and are more appropriate for modeling the dependence on covariates.

Under Assumptions~\ref{assu:SUTVA}, \ref{assu:unconfoundedness}, \ref{assu:overlap}, the above causal estimands are identifiable. In fact, for non-survival responses,
\begin{align*}
	\EE[Y(1) \mid X = x]
	= \EE[Y(1)  \mid X = x, W = 1]
	= \EE[Y \mid X = x, W = 1],
\end{align*}
where $\EE[Y \mid X = x, W = 1]$ is the conditional mean in the treatment group. Estimands~\eqref{eq:tauContinuous}, \eqref{eq:tauBinary}, and \eqref{eq:tauCount} are functions of $\EE[Y(1) \mid X = x]$ and $\EE[Y(0) \mid X = x]$, and thus estimable. 
For survival responses, since
\begin{align*}
	\PP(Y(1) \in  [y, y + \delta] \mid X = x)
	&= 	\PP(Y(1) \in  [y, y + \delta] \mid X = x, W = 1) \\
	&= 	\PP(Y \in  [y, y + \delta] \mid X = x, W = 1),
\end{align*}
the causal estimand~\eqref{eq:tauCox} can be simplified so as not to depend on the counterfactuals and is identifiable.

In this paper, we work under the partially linear assumption \citep{robinson1988root,chernozhukov2018double}. 
\begin{assumption}[Semi-parametric model]\label{assu:semi}
	Assume the heterogeneous treatment effect follows the linear model
	\begin{align}\label{eq:DINA:parametric}
		\tau(x) = x^\top \beta.
	\end{align}
\end{assumption}
The partially linear assumption allows non-parametric nuisance functions (the natural parameter functions and the propensity score) and assumes that the heterogeneous treatment effect follows the linear model.
Predictors $x$ in model~\eqref{eq:DINA:parametric}  can be replaced by any known functions of the covariates.
For instance, if the treatment effect is believed to be homogeneous, we will use $x = 1$;
if we are interested in the treatment effect in some sub-populations, we will design categorical predictors to specify the desired subgroups.
Our method with non-parametric $\tau(x)$ is discussed in Section~\ref{sec:discussion}.

The semi-parametric model encodes the common belief that the natural parameter functions are more complicated than the treatment effect \citep{hansen2008prognostic,Kunzel4156, gao2020minimax}. Consider a motivating example of hypertension: the blood pressure of a patient could be determined by multiple factors over a long period, such as the income level and living habits, while the effect of an anti-hypertensive drug is likely to interact with only a few covariates, such as age and gender, over a short period.

\subsection{Literature}

There is a rich literature on using flexible modeling techniques to estimate heterogeneous treatment effects. \citet{tian2012simple, imai2013estimating} formulate the estimation of heterogeneous treatment effects as a variable selection problem and consider a LASSO-type approach. \citet{athey2015machine,su2009subgroup,foster2011subgroup} design recursive partitioning methods for causal inference and \citep{hill2011bayesian} adapts the Bayesian Additive Regression Trees (BART). Ensemble learners, such as random forests \citep{wager2018estimation} and boosting \citep{powers2018some}, have been investigated under the counterfactual framework. In addition, neural network-based causal estimators have also been proposed \citep{kunzel2018transfer,shalit2017estimating}.

More recently, meta-learners for heterogeneous treatment effect estimation are of increasing popularity. Meta-learners decompose the estimation task into sub-problems that can be solved by off-the-shelf machine learning tools (base learners) \citep{Kunzel4156}. 
One common meta-algorithm, which we call separate estimation (SE) later, applies base learners to the treatment and control groups separately and then takes the difference \citep{foster2011subgroup, lu2018estimating, hu2021estimating,foster2011subgroup}. Another approach regards the treatment assignment as a covariate and uses base learners to learn the dependence on the enriched set of covariates.
\citet{Kunzel4156} propose X-learner that first estimates the control group mean function, subtracts the predicted control counterfactuals from the observed treated responses, and finally estimates treatment effects from the differences. X-learner is effective if the control group and the treatment group are unbalanced in sample size.
A different approach R-learner \citep{nie2017quasi}, motivated by Robinson's method \citep{robinson1988root}, estimates the propensity score and the marginal mean function (nuisance functions) using arbitrary machine learning algorithms and then minimizes a designed loss with the estimators of nuisance function plugged in to learn treatment effects.
R-learner is able to produce accurate estimators of treatment effects from less accurate estimators of nuisance functions (more details are discussed in Section \ref{sec:Robinson}).

In this paper, we aim to employ available predictive machine learning algorithms to quantify causal effects on the natural parameter scale. In particular, we extend the framework of R-learner which focuses on the conditional mean difference to estimate differences in natural parameters and hazard ratios. Inherited from R-learner, our proposed method uses black-box predictors and is robust to observed confounders. Beyond R-learner, our proposed method provides quantifications of causal effects that may be of more practical use and protects against non-collapsibility under non-linear link functions.

\section{Robinson's method and R-learner}\label{sec:Robinson}

In this section, we briefly summarize Robinson's method and R-learner. We highlight how Robinson's method and R-learner provide protection against confounding. We also discuss the difficulty, non-collapsibility in the natural parameter scale, in extending Robinson's method and R-learner.

We consider the additive error model.
Let $\eta_0(x)$, $\eta_1(x)$ be the conditional
control and treatment group mean functions, and assume the error term $\varepsilon$ satisfies $\EE[\varepsilon \mid X] = 0$. The additive error model takes the form
\begin{align}\label{eq:model:additive}
	\begin{split}
		\begin{cases}
			Y(1) =
			\eta_1(X) + \varepsilon,&\\
			Y(0) =
			\eta_0(X) + \varepsilon.&
		\end{cases}
	\end{split}
\end{align}	
Robinson's method \citep{robinson1988root} and R-learner \citep{nie2017quasi} aim to estimate the difference of the conditional means $\tau(x) = \eta_1(x) - \eta_0(x)$. 

Let $m(x) := \EE[Y \mid X = x] = \eta_0(x) + e(x)\tau(x)$ be the marginal mean function. Model~\eqref{eq:model:additive} can be reparametrized as 
\begin{align}\label{eq:model:additive4}
	Y =  \eta_0(X) + W \tau(X) + \varepsilon = m(X) + (W-e(X))\tau(X) + \varepsilon.
\end{align}
Robinson works under the semi-parametric assumption $\tau(x) = x^\top \beta$ and proposes to estimate $m(x)$, $e(x)$, and $\beta$ in two steps:
\begin{enumerate}
	\item {Estimation of nuisance functions}. Estimate the propensity score $e(x)$ and the marginal mean function $m(x)$.
	\item {Least squares}. Fit a linear regression model to response $Y$ with offset $\hat{m}(x)$ and predictors $(W - \hat{e}(X))X$.
\end{enumerate}
R-learner adopts the same reparametrization~\eqref{eq:model:additive4} and the two-step procedure, but estimates $\tau(x)$ non-parametrically. In the following, we will use Robinson's method and R-learner interchangeably.

\subsection{Confounding}\label{subsec:confounding}

Confounders are covariates that affect both the treatment assignment and potential outcomes. Confounders are common in observational studies. 
We show that R-learner is insensitive to confounding while separate estimation will produce biased estimators.

We consider an example of the additive error model~\eqref{eq:model:additive} with 
\begin{align*}
	\eta_0(x) = \eta_1(x) = x_1^2,
\end{align*}
and $\tau(x) = 0$. 
We adopt the propensity score $e(x) = e^{x_1}/(1+e^{x_1})$ and assume it is known. 
In the example, $x_1$ influences both the treatment assignment and potential outcomes, and is thus a confounder.

Separate estimation fits separate models in the control and treatment groups to get $\hat{\eta}_0(x)$ and $\hat{\eta}_1(x)$, and then estimates the HTE by $\hat{\tau}_{\text{SE}}(x) = 
\hat{\eta}_1(x) - \hat{\eta}_0(x)$.
For illustration, we estimate nuisance functions $\eta_0(x)$, $\eta_1(x)$ by linear regression.
By the design of propensity score, the distributions of $x_1$ are different in the control and treatment groups, and the best linear approximations $\hat{\eta}_0(x)$, $\hat{\eta}_1(x)$ to $x_1^2$ are also different. Therefore, the estimator $\hat{\eta}_1(x) - \hat{\eta}_0(x)$ is far from zero and inconsistent (Figure~\ref{fig:intuition} panel (a)).

In contrast, R-learner produces consistent estimators even if the estimator of the nuisance function $m(x)$ is incorrect.
Continuing from the example above, the true marginal mean function is $m(x) = \eta_0(x)  = x_1^2$. 
In the first step of R-learner, we also estimate the nuisance function $m(x)$ by linear regression to make the comparison fair.
In the second step, we regress the residuals $Y - \hat{m}(X)$ on $(W - e(X))X$ with the true propensity score $e(x)$ provided. 
The residuals $Y - \hat{m}(X)$ consist of three parts: the component of treatment effect $(W-e(X))\tau(X)$, the bias of the nuisance-function estimator $m(X)-\hat{m}(X)$, and an uncorrelated noise $\varepsilon$. Since there is no treatment effect in the example, the residuals are essentially realizations of a function of covariates plus a zero-mean noise,
\begin{align*}
	Y - \hat{m}(X) =  m(X)-\hat{m}(X) + \varepsilon.
\end{align*}
Notice that for any function of covariates $g(x)$,
\begin{align*}
	\cov\left((W - e(X))X, ~g(X)\right) = 0.
\end{align*}
Therefore, the residuals $Y - \hat{m}(X)$ are uncorrelated with the regressors $(W - e(X))X$ (Figure~\ref{fig:intuition} panel (b)), and the regression coefficients of R-learner are asymptotically zero and consistent. 

\begin{figure}
	\centering
	\begin{minipage}{7cm}
		\centering  
			\includegraphics[scale=0.4]{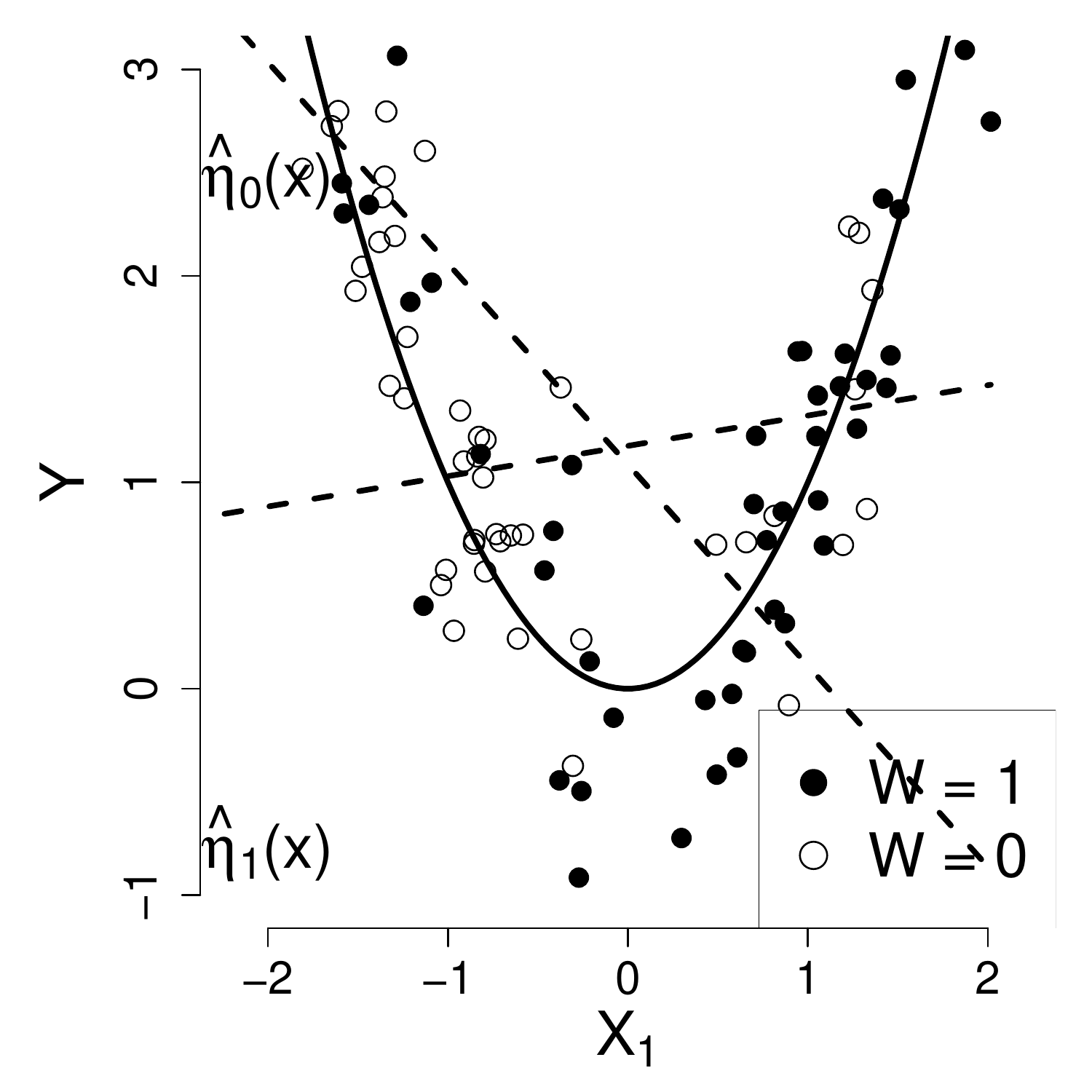}  \\
	{{(a)}}
	\end{minipage}
	\begin{minipage}{7cm}
		\centering  
		\includegraphics[scale=0.6]{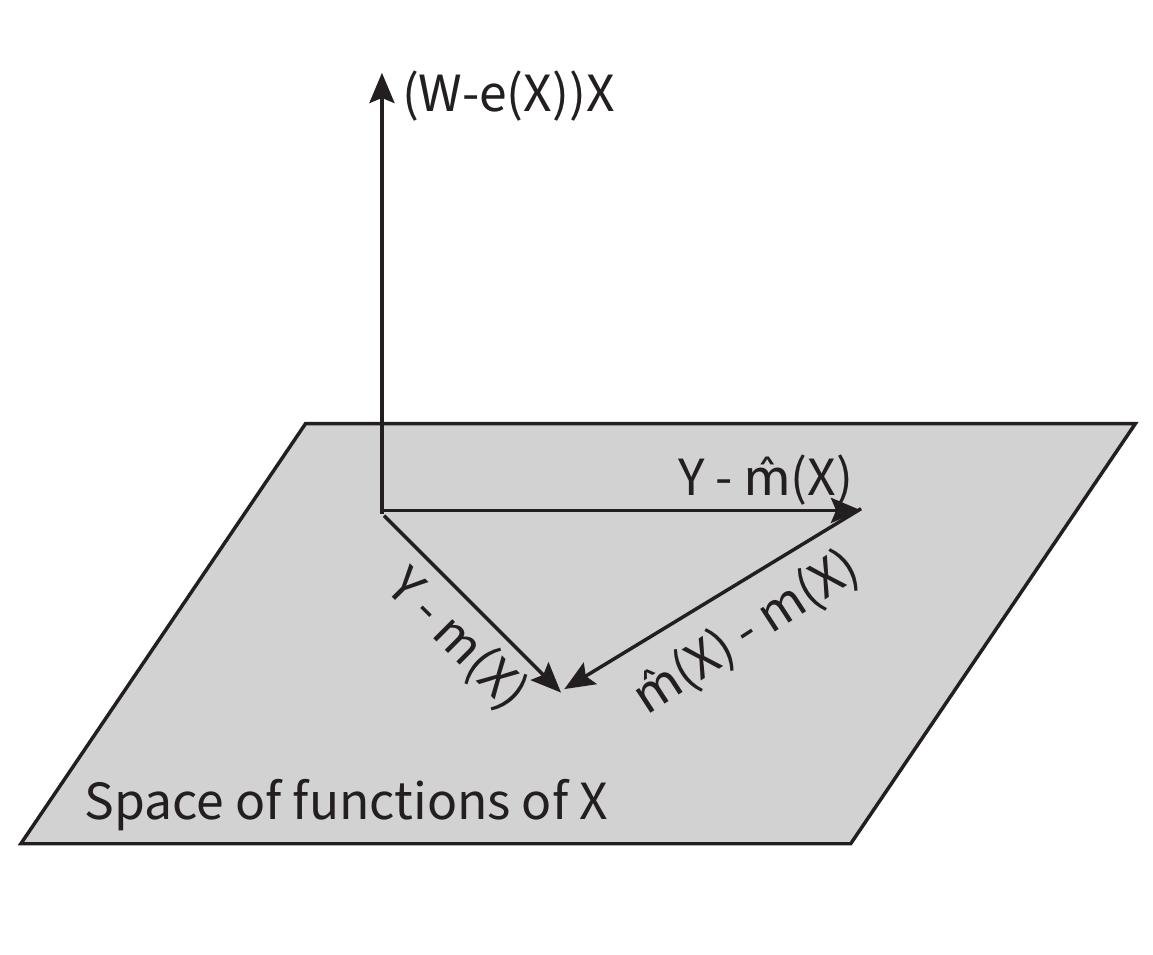} \\
		{{(b)}}
	\end{minipage}
	\caption{{In panel (a), solid points stand for treated units and hollow circles stand for control units. The solid parabola represents the true $\eta_0(x) = \eta_1(x) = x_1^2$. For separate estimation, the dashed curve with a negative slope is the nuisance-function estimator $\hat{\eta}_0(x)$ and that with a positive slope is $\hat{\eta}_1(x)$. The estimator of HTE $\hat{\tau}_{\text{SE}}(x) = \hat{\eta}_1(x) - \hat{\eta}_0(x)$ is far from zero and inconsistent. 
			In panel (b), the space of functions of $X$ are orthogonal to that of $(W-e(X))X$.
			Even if the nuisance-function estimator $\hat{m}(x) $ is inaccurate,
			the error $\hat{m}(X)  -  m(X)$ lies in the space of functions $X$ and will not change the projection of the residuals $Y - \hat{m}(X)$ to the space spanned by $(W-e(X))X$.}}
	\label{fig:intuition}
\end{figure}

\subsection{Non-collapsibility}\label{subsubsec:noncollapsibility}

Non-collapsibility refers to the phenomenon that the mean of conditional measures does not equal the marginal counterpart.
Non-collapsibility \citep{fienberg2007analysis, greenland1999confounding} and confounding are two distinct issues.
In other words, even if there is no confounder, the non-collapsibility issue may arise when we use non-collapsible association measures. 
For example, we consider randomized experiments ($e(x) = 0.5$) and the exponential family model~\eqref{eq:model:exponentialFamily} with the logistic link. We assume the true natural parameter functions are 
\begin{align}\label{eq:model:additive5}
	\eta_0(x,z) = z, \quad \eta_1(x,z) = \tau(x) + z,
\end{align}
where $Z$ is an unobserved binary random variable. 
For constant $\tau(x)$, the true treatment effect is the log odds ratio conditional on $Z = 0$ or $Z = 1$, but may not equal the log of the marginal odds ratio.  
More generally, for natural parameter functions \eqref{eq:model:additive5} and constant $\tau(x)$, the separate estimation method that estimates $\eta_0(x,z)$, $\eta_1(x,z)$ by generalized linear regression is only consistent for Gaussian and Poisson responses \citep{gail1984biased}.
For non-constant $\tau(x)$, the separate estimation method is only consistent for Gaussian (Claim~\ref{claim:conditionalMean} in the appendix).
Figure~\ref{fig:errorToy2} verifies the non-collapsibility issue numerically.

\begin{figure}
	\centering
	\begin{minipage}{4.6cm}
		\centering  
		\includegraphics[scale=0.3]{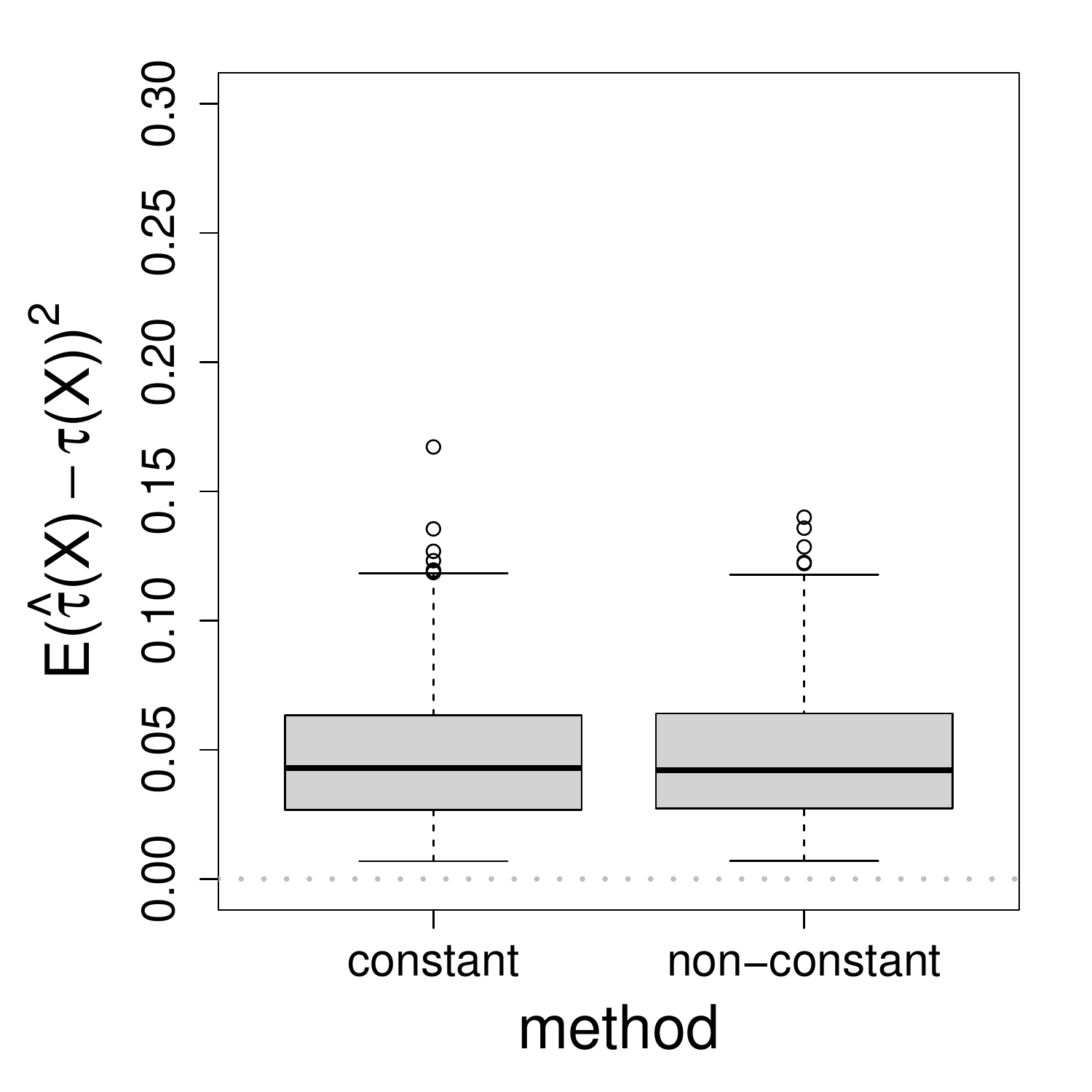}  
		{{(a) linear link}}
	\end{minipage}
	\begin{minipage}{4.6cm}
		\centering  
		\includegraphics[scale=0.3]{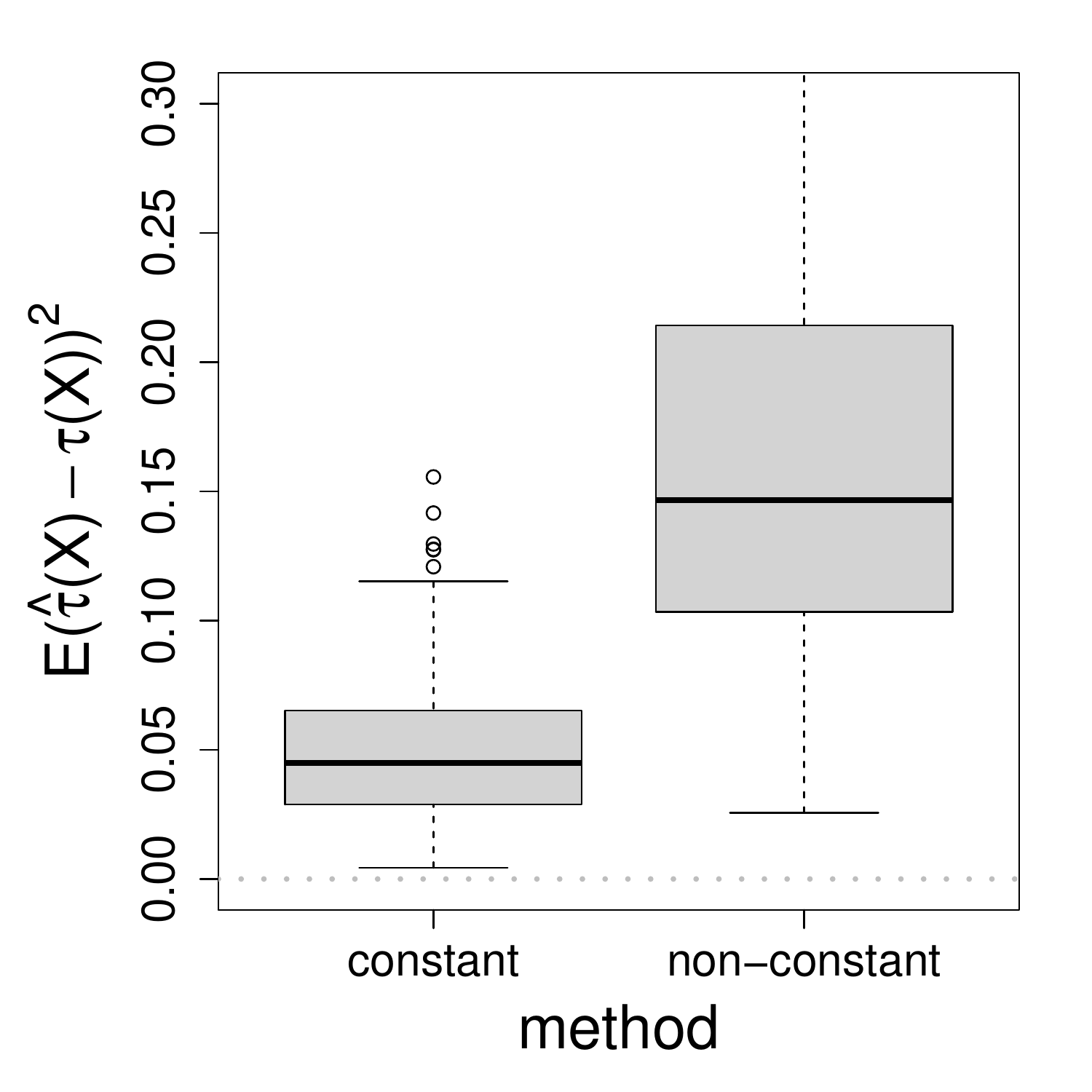}  
		{{(b) log-link}}
	\end{minipage}
	\begin{minipage}{4.6cm}
		\centering  
		\includegraphics[scale=0.3]{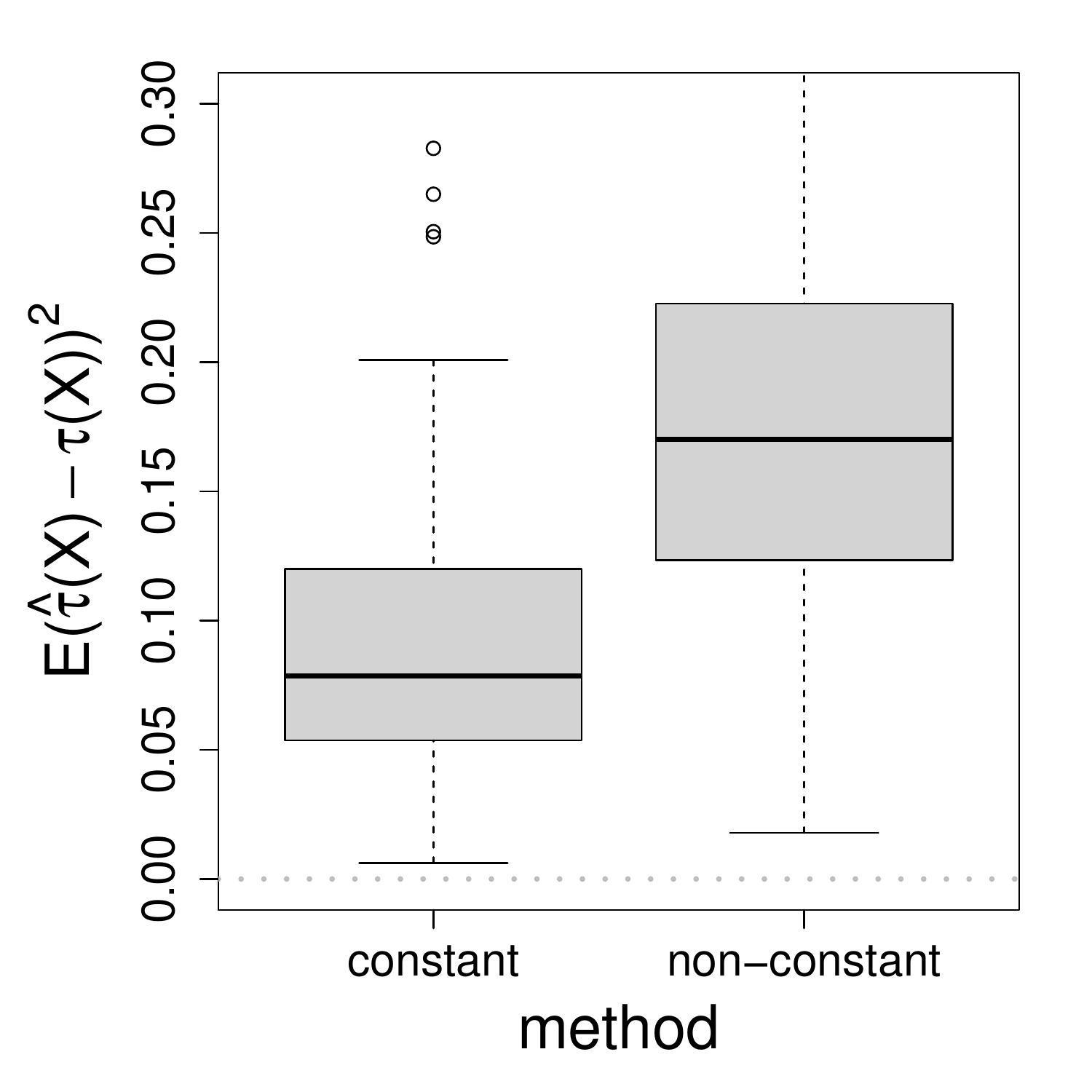}
		{{(c) logit-link}}
	\end{minipage} 
	\caption{{Estimation error box plots of separate
			estimation applied to the exponential
			family with natural parameter functions~\eqref{eq:model:additive5}. In the left, middle, and right panels, the
			responses are generated from Gaussian, Poisson, and
			Bernoulli distributions, corresponding to the identity,
			log, and logit canonical link functions, respectively. For
			each response distribution, we consider constant treatment
			effect and heterogeneous treatment effect. The signal-to-noise ratio is kept at the same level and below one. 
			The separate estimation method estimates the nuisance functions $\eta_0(x,z)$, $\eta_1(x,z)$ by fitting generalized linear regression to observed covariates $x$. The separate estimation method is always unbiased with the linear link, only unbiased for constant treatment effect with the log-link, and always biased with the logit-link. 
	}}
	\label{fig:errorToy2}
\end{figure}

R-learner is designed to estimate the difference in conditional means but not natural parameter functions. 
Direct extension of R-learner to the natural parameter scale is infeasible: the response distribution conditional on the covariates $X$ is often not a member of the exponential family, and thus the marginal natural parameter function analogous to the marginal mean function $m(x)$ is not well-defined. 
In Section~\ref{sec:method} and \ref{sec:cox}, we introduce our extension that inherits the robustness to confounders from R-learner and provides protection against non-collapsibility on the natural parameter scale.

\section{DINA for exponential family}\label{sec:method}

In this section, we introduce our DINA estimator in the exponential family inspired by R-learner.
We consider the following working model of the potential outcomes,
\begin{align}\label{eq:model:exponentialFamily}
	\begin{split}
		\begin{cases}
			f_1(y \mid X=x) = 
			\kappa(y) e^{y \eta_1(x) - \psi(\eta_1(x))},&\\
			f_0(y  \mid X=x) = 
			\kappa(y) e^{y \eta_0(x) - \psi(\eta_0(x))}.&
		\end{cases}
	\end{split}
\end{align}
Here $\eta_0(x)$, $\eta_1(x)$ denote the natural parameters of the control group and the treatment group respectively, $\psi(\eta)$ is the cumulant generating function, and $\kappa(y)$ is the carrier density.
The distribution family \eqref{eq:model:exponentialFamily}
includes the commonly used Bernoulli and Poisson
distributions, among others.

We start by rewriting the natural parameters in model \eqref{eq:model:exponentialFamily} with Assumption~\ref{assu:semi} as 
\begin{equation}
	\label{eq:expfampair}
	\begin{array}{rcl}
		\eta_0(x)&=&\eta_0(x),\\
		\eta_1(x)&=&\eta_0(x)+x^\top\beta.
	\end{array}
\end{equation}
We then construct a baseline $\nu(x)$ that is particular mixture of the natural parameters
\begin{equation}\label{eq:condition4}
	\begin{aligned}
		\nu(x) &=& a(x) \eta_1(x) + (1-a(x)) \eta_0(x),\\
		a(x) &=& \frac{e(x)V_1(x) }{e(x)V_1(x) + (1-e(x))V_0(x)}.
	\end{aligned}
\end{equation}
Here $a(x)$ is a type of modified propensity score, with
$V_1(x)=\frac{d\mu}{d\eta}(\eta_1(x))$ and
$V_0(x)=\frac{d\mu}{d\eta}(\eta_0(x))$ the variance functions for
the exponential family, 
where $\mu(\eta)$ denotes the mean function (inverse of the canonical link function). 
This allows us to reparametrize (\ref{eq:expfampair}) 
\begin{align}\label{eq:model:exponentialFamily2}
	\eta_w(x) = \nu(x) + (w - a(x)) x^\top \beta,\quad w\in\{0,1\}.
\end{align}
R-learner for the Gaussian distribution uses $a(x) = e(x)$,
$\nu(x) = e(x) \eta_1(x) + (1-e(x)) \eta_0(x) = \EE[Y \mid X = x]$ (since
$\eta_w(x)=\mu_w(x)$ in this case), and is
able to produce unbiased $\beta$ even if $\nu(x)$ is
misspecified. 
Similarly, as shown in Proposition \ref{prop:derivationExponential} in the appendix, we designed $a(x)$ and $\nu(x)$ so that even if we start with an inaccurate baseline $\nu^*(x) \neq \nu(x)$, we can still arrive at the true DINA parameter $\beta$.

\begin{claim}\label{claim:conditionalMean}
	Under the model \eqref{eq:model:exponentialFamily}, for arbitrary covariate value $x$, $\nu(x)$ in \eqref{eq:condition4} satisfies
	\begin{align*}
		\EE[Y \mid X = x]
		= \mu(\nu(x)) + O\left(\tau^2(x)\right).
	\end{align*}
\end{claim}
For linear canonical link function (R-learner), the equation in Claim~\ref{claim:conditionalMean} is exact with no remainder term.
For arbitrary canonical link functions, the marginal conditional mean $\EE[Y \mid X = x]$ approximately equals $\mu(\nu(x))$ if the treatment effect $\tau(x)$ is relatively small in scale compared to the marginal mean function $\EE[Y \mid X = x]$.

Based on \eqref{eq:condition4}, we propose the following two-step estimator (details in Algorithm~\ref{algo:exponentialFamily}):
\begin{enumerate}
	\item {Estimation of nuisance functions}. We estimate the
	functions $a(x)$ and $\nu(x)$ in \eqref{eq:condition4}, using
	estimators of the propensity score $e(x)$ and the natural
	parameter functions $\eta_0(x)$ and $\eta_1(x)$.
	\item {Maximum likelihood estimator (MLE)}. 
	Fit a generalized linear regression model to response $Y$ with offset $\hat{\nu}(x)$ and predictors $(W - \hat{a}(X))X$. 
\end{enumerate}
We remark that the construction \eqref{eq:condition4} can also be regarded as designing Neyman's orthogonal scores \citep{neyman1959optimal,newey1994asymptotic,van2000asymptotic} specialized to the potential outcome model with the exponential family (Claim \ref{prop:Neyman} in Appendix \ref{subsec:appendix:exponentialFamily}). 

We explain why the robustness to the nuisance-function estimators protects the proposed method from confounding and non-collapsibility. 
In fact, provided with the true nuisance
functions, the MLE of $\beta$ is well-behaved despite confounding and
non-collapsibility. Since the proposed method is insensitive to the
misspecification of nuisance functions,
the DINA estimator $\hat{\beta}$ will
remain close to that using the accurate nuisance functions. In
contrast, the separate estimation method relies more heavily on
precise nuisance-function estimation, and is prone to confounding and
non-collapsibility.

We demonstrate the improvement of our proposed method over separate estimation and the direct extension of R-learner with Poisson responses. 
Here the direct extension of R-learner refers to the configuration $a(x) = e(x)$, $\nu(x) = e(x) \eta_1(x) + (1-e(x))\eta_0(x)$ in \eqref{eq:model:exponentialFamily2}. We focus on the confounding and non-collapsibility issues, and design three scenarios: (a) confounding only; (b) non-collapsibility only; (c) confounding and non-collapsibility. In Figure~\ref{fig:errorToy3}, the proposed method is insensitive to both the confounding and the non-collapsibility issues. 

\begin{figure}
	\centering
	\begin{minipage}{4.6cm}
		\centering  
		\includegraphics[scale=0.3]{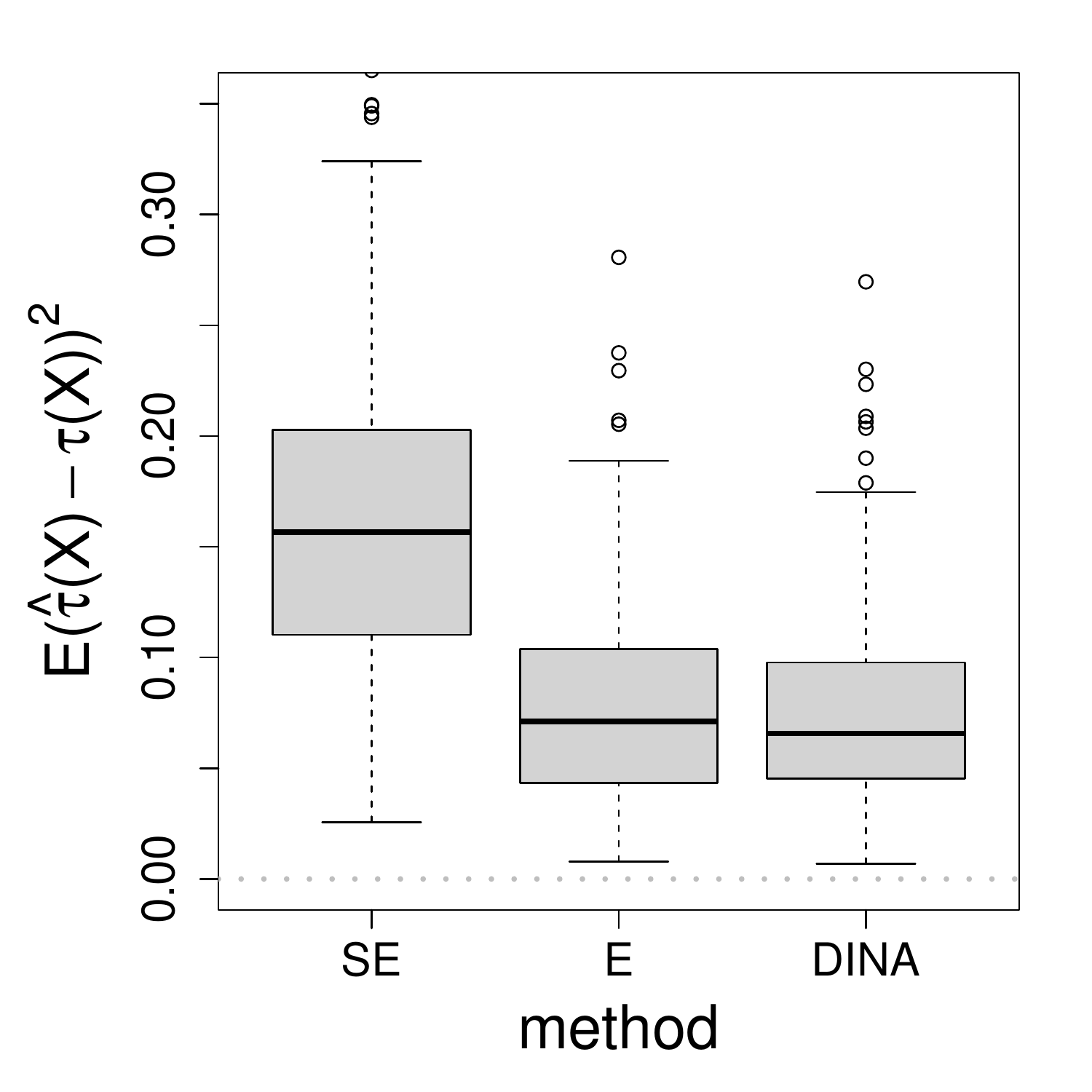}  \\
		{{(a) confounding only}}
	\end{minipage}
	\begin{minipage}{4.6cm}
		\centering  
		\includegraphics[scale=0.3]{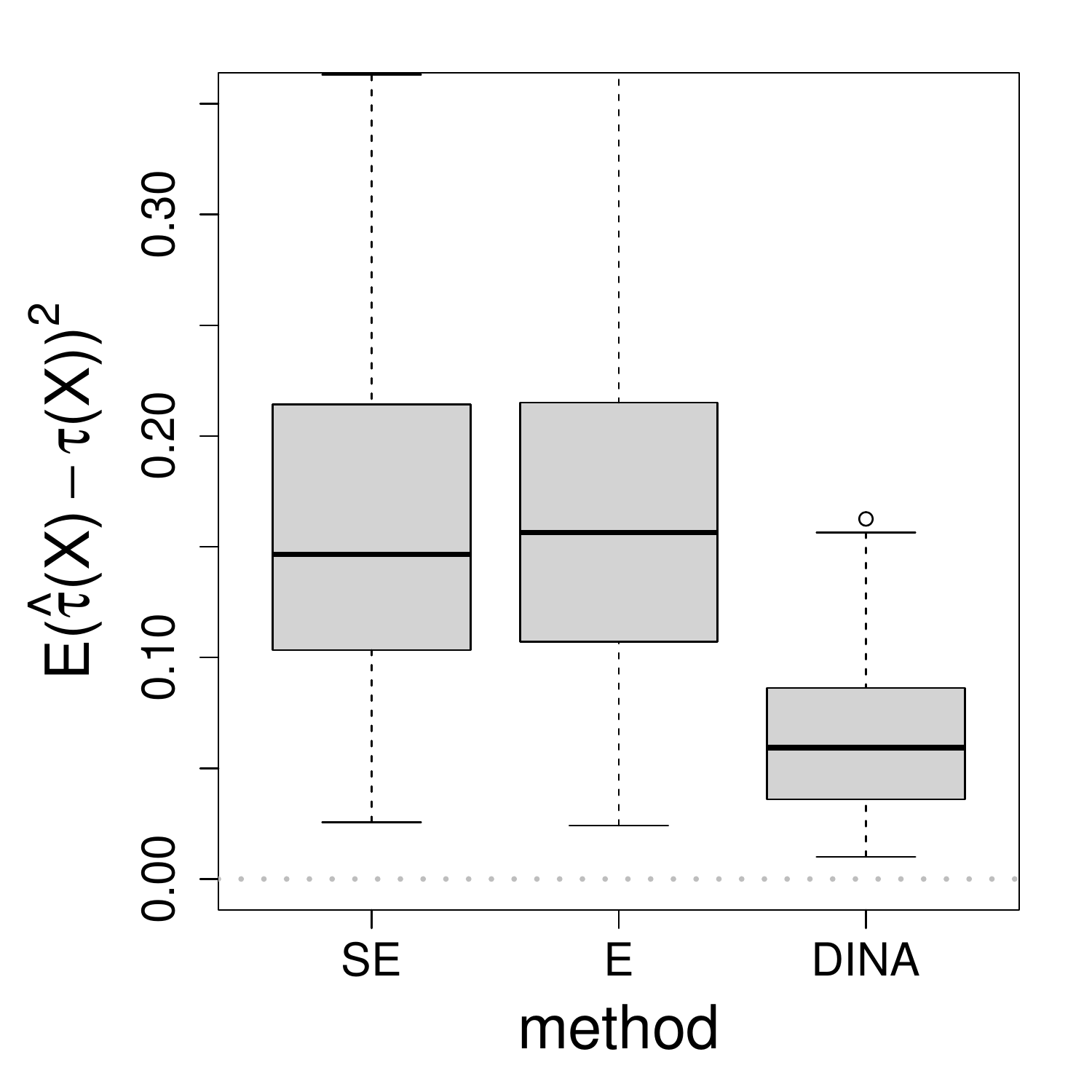}  \\
		{{(b) non-collapsibility only}}
	\end{minipage}
	\begin{minipage}{4.6cm}
		\centering  
		\includegraphics[scale=0.3]{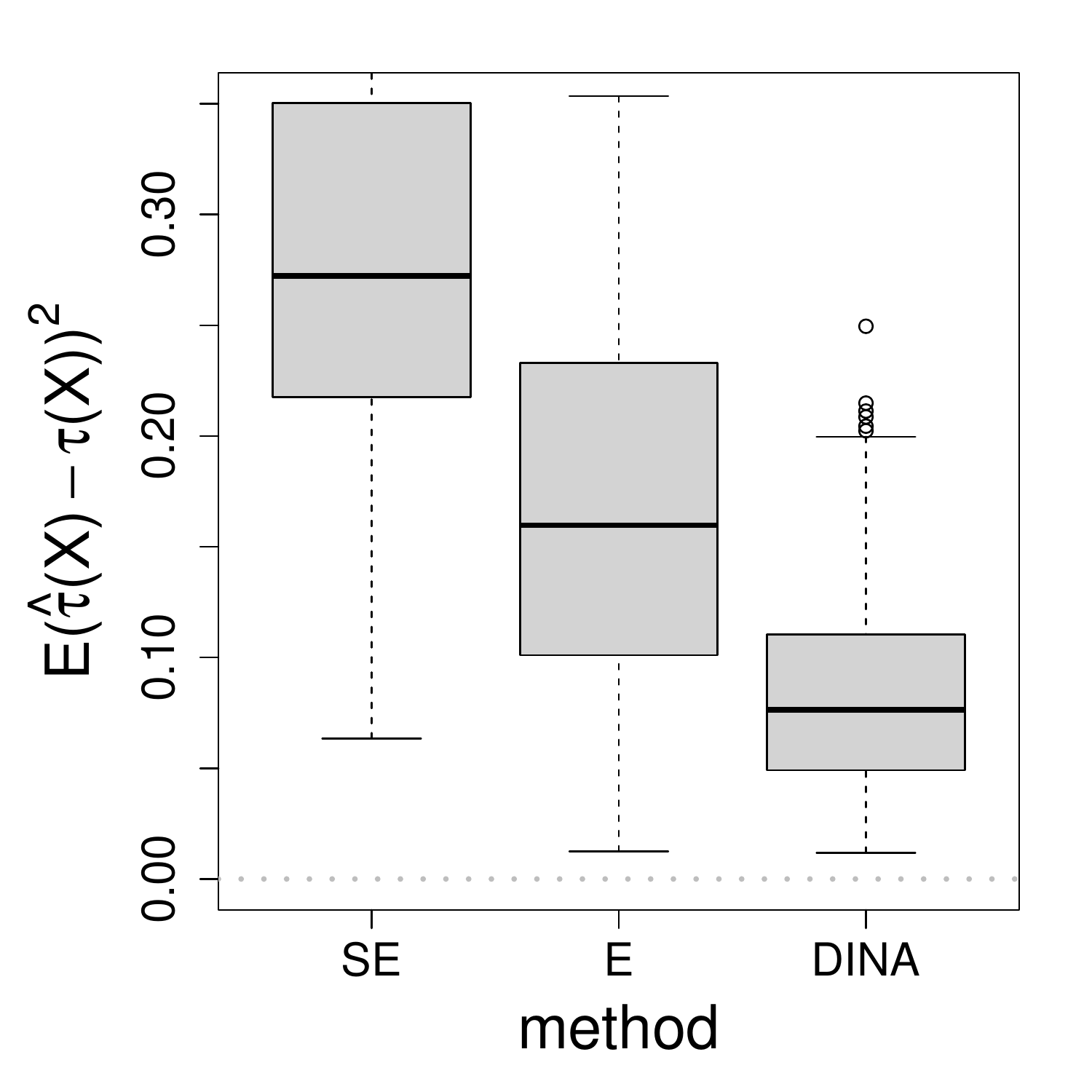}\\
		{{(c) confound.\&non-collaps.}}
	\end{minipage}
	\caption{{Boxplots of estimation error for  separate estimation (SE), the direct extension of R-learner (E), and the proposed method (DINA) with Poisson responses. In the left panel, the treatment assignment is confounded, but the treatment effect is constant and thus there is no non-collapsibility issue according to \citet{gail1984biased}; in the middle panel, the treatment is randomized, but the treatment effect is not constant; in the right panel, the treatment assignment is confounded and the treatment effect is not constant. The signal-to-noise ratio is kept at the same level and below one. Separate estimation, direct extension of R-learner, and Algorithm~\ref{algo:exponentialFamily} all estimate $\eta_0(x)$, $\eta_1(x)$ by generalized linear regression. Propensity score is estimated by logistic regression.}}
	\label{fig:errorToy3}
\end{figure}

Next we provide some intuition for the proposed method, details of the implementation, and some theoretical properties.

\subsection{Interpretation}\label{subsec:interpretation}
We provide insights into Algorithm~\ref{algo:exponentialFamily} by comparing it to R-learner.
As discussed above, Algorithm~\ref{algo:exponentialFamily} is an extension of R-learner, and shares its motivation as well as the method skeleton. 
The major difference lies in the nuisance functions \eqref{eq:condition4}.
In R-learner, $a(x)$ equals the propensity function $e(x)$; in
the proposed DINA estimator, $a(x)$ not only depends on $e(x)$ but
also the variance functions $V_0(x)$ and $V_1(x)$ for
the particular exponential family.
Since $a(x)$ lies between 0 and 1, it can be decomposed as
\begin{equation}
	\label{eq:logitax}
	\log\frac{a(x)}{1-a(x)}=\log\frac{e(x)}{1-e(x)}+\log\frac{V_1(x)}{V_0(x)}.
\end{equation}
$a(x)$ is large if a unit is likely to be treated or
their response under treatment has higher variance than under no
treatment.
As a result, the $a(x)$ not only
balances the treatment group and the control group, but also reduces
the impact of differential variance between the two groups.
For more comments on the adjustment based on $V_w(x)=\frac{d\mu}{d\eta}(\eta_w(x))$,  see the multi-valued treatment setting in Section~\ref{sec:multiValued}.

\subsection{Algorithm}\label{subsec:exponential:algorithm}

Similar to R-learner, we estimate the DINA in
\eqref{eq:model:exponentialFamily2} in two
steps. We estimate nuisance functions~\eqref{eq:condition4} and the DINA using independent subsamples so that the estimation of the
nuisance functions
does not interfere with that of DINA. We employ
cross-fitting \citep{chernozhukov2018double} to boost data
efficiency. We now give details.

In the first step, we estimate the nuisance functions $\nu(x)$ and
$a(x)$ in \eqref{eq:condition4}. These in turn depend on the 
propensity score $e(x)$, and separate estimators of $\eta_0(x)$ and
$\eta_1(x)$. Given these latter two, the variance functions $V_0(x)$
and $V_1(x)$ are immediately available from the corresponding
exponential family. For example for the binomial family, we have that
$V_0(x)=\mu_0(x)(1-\mu_0(x))$, where
$\mu_0(x)=e^{\eta_0(x)}/(1+e^{\eta_0(x)})$.
The propensity scores can be estimated by
any classification methods that provide probability estimates.
Likewise the functions $\eta_0(x)$ and $\eta_1(x)$ can be separately estimated by
any suitable method that operates within the particular exponential family.

Upon obtaining
$\hat{e}(x)$,  $\hat{\eta}_0(x)$, and
$\hat{\eta}_1(x)$, we plug them into (\ref{eq:condition4}) to get
$\hat{a}(x)$ and $\hat{\nu}(x)$. We remark that though we can directly use the
difference of the two estimators $\hat{\eta}_1(x)$,
$\hat{\eta}_0(x)$ as a valid DINA estimator ---  separate estimation
--- we show in Proposition~\ref{prop:main} that our method yields more accurate DINA estimators. 
Note that in the case of the Gaussian distribution,  $a(x) = e(x)$,
$\nu(x) = \EE[Y\mid X = x]$, and we can directly estimate $\nu(x)$ by
any conditional mean estimator of $Y$ given $X$ and avoid having to
separately estimate  $\hat{\eta}_0(x)$ and  $\hat{\eta}_1(x)$.

In the second step, we maximize the log-likelihood corresponding to
(\ref{eq:model:exponentialFamily2}) with the functions $\hat{\nu}(x)$
and $\hat{a}(x)$ fixed from the first step\footnote{If there are
	collinearity or sparsity patterns, penalties such as ridge and LASSO
	can be directly added to the loss function.}. The function $\nu(x)$
is regarded as an \emph{offset}, and the function $a(x)$ is used to
construct the predictors $(W-a(X))X$. The second step can be
implemented, for example,  using the \texttt{glm} function in {R}.

\begin{algorithm}
	\BlankLine
	\caption{Exponential family}\label{algo:exponentialFamily}
	We split the data randomly equally into two folds.
	
	1. On fold one, we obtain nuisance-function estimators $\hat{a}(x)$ and $\hat{\nu}(x)$. 
	\begin{enumerate}
		\item {Estimation of $e(x)$}. Estimate the
		propensity scores by any classification method that
		provides probability estimates, such as logistic
		regression, random forests, or boosting.
		\item {Estimation of $\eta_0(x)$,
			$\eta_1(x)$}. Estimate the natural parameter
		functions $\eta_0(x)$ and $\eta_1(x)$ based on the
		control group and the treatment group respectively,
		such as through fitting generalized linear models,
		random forests, or boosting.
		\item {Substitution}. Plug the estimators $\hat{e}(x)$, $\hat{\eta}_0(x)$, $\hat{\eta}_1(x)$ into \eqref{eq:condition4} and get $\hat{a}(x)$. Further let $\hat{\nu}(x) = \hat{a}(x) \hat{\eta}_1(x) + (1-\hat{a}(x)) \hat{\eta}_0(x)$.
	\end{enumerate}
	2. On fold two, we obtain $\hat{\beta}$ by maximizing the
	log-likelihood
	\begin{align}\label{eq:MLE}
		\begin{split}
		\max_{\beta'} \ell(\beta'; \hat{a}(x), \hat{\nu}(x)) :=& \frac{1}{n}\sum_{i=1}^n Y_i\left(\hat{\nu}(X_i) + (W_i-\hat{a}(X_i)) X_i^\top \beta'\right) \\
		&- \psi\left(\hat{\nu}(X_i) + (W_i-\hat{a}(X_i)) X_i^\top \beta'\right)
		\end{split}
	\end{align}
	with $\hat{a}(x)$ and $\hat{\nu}(x)$ plugged in. Denote the estimated coefficient by $\hat{\beta}^{(1)}$.\\
	
	3. We swap the two folds and obtain another estimate
	$\hat{\beta}^{(2)}$. We output the average $(\hat{\beta}^{(1)} +
	\hat{\beta}^{(2)})/2$.
\end{algorithm}

From the practical perspective, Algorithm~\ref{algo:exponentialFamily}
decouples the estimation of DINA from that of the  nuisance functions and allows for flexible estimation of nuisance functions. 
Though the nuisance functions $a(x)$ and $\nu(x)$ are complicated,
various methods are applicable since there are no missing data in the
nuisance-function estimation. By modularization, in step one, we are free to use any off-the-shelf methods, and in step two, we solve the specially designed MLE problem with the robustness protection.

\subsection{Theoretical properties}\label{subsec:exponential:property}
The motivation behind the method is to gain robustness to nuisance functions, and we make our idea rigorous in the following proposition.
\begin{proposition}
	\label{prop:main}
	Under the regularity conditions:
	\begin{enumerate}
		\item  Covariates $X$ are bounded, the true parameter $\beta$ is in a bounded region $\calB$, nuisance functions $a(x)$, $\nu(x)$ and nuisance-function estimators ${a}_n(x)$, ${\nu}_n(x)$ are uniformly bounded;
		\item The minimal eigenvalues of the score derivative $\nabla_{\beta} s(a(x), \nu(x),\beta)$ in $\calB$ are lower bounded by $C$, $C>0$.
	\end{enumerate}
	Assume that\footnote{The norm $\|f(x)\|_2$ is defined as $(\EE[f^2(X)])^{1/2}$.} $\|{a}_n(x) - a(x)\|_2$, $\|{\nu}_n(x) - \nu(x)\|_2 = O(c_n)$, $c_n \to 0$, then
	\begin{align*}
		\|{\beta}_n - \beta\|_2 = \tilde{O}\left(c_n^2 + n^{-1/2}\right).
	\end{align*}
\end{proposition}

Proposition~\ref{prop:main} states that for $\hat{\tau}(x) = x^\top \hat{\beta}$ to achieve a certain accuracy, the conditions on $\hat{a}(x)$ and $\hat{\nu}(x)$ are relatively weak. This implies the method is locally insensitive to the nuisance functions, and thus is robust to noisy plugged-in nuisance-function estimators. In other words, we can view Algorithm~\ref{algo:exponentialFamily} as an accelerator:
we input crude estimators $\hat{\eta}_1(x)$, $\hat{\eta}_0(x)$, and  $\hat{e}(x)$, and Algorithm~\ref{algo:exponentialFamily} outputs twice accurate DINA estimators compared to the simple difference $\hat{\eta}_1(x) - \hat{\eta}_0(x)$. 
As a corollary of Proposition~\ref{prop:main}, if both $\hat{a}(x)$ and $\hat{\nu}(x)$ can be estimated at the rate $n^{-1/4}$, then the DINA can be estimated at the parametric rate $n^{-1/2}$.

\section{DINA for Cox model}\label{sec:cox}

In this section, we first draw a connection between the Cox model and the exponential family using the full likelihood, and further generalize Algorithm~\ref{algo:exponentialFamily} to the Cox model. We then discuss how the generalized method fits in the partial-likelihood framework. 

Let $\Lambda(y)$ be the baseline cumulative hazard function, i.e., $\Lambda(y) =
\int_0^y \lambda(t) dt$.
We assume the potential outcomes follow
\begin{align}\label{eq:model:cox}
	\begin{split}
		\begin{cases}
			\PP(Y(1) \ge y \mid X = x) = e^{-\Lambda(y) e^{\eta_1(x)}},\\
			\PP(Y(0) \ge y \mid X = x) = e^{-\Lambda(y) e^{\eta_0(x)}}.
		\end{cases}
	\end{split}
\end{align}	
In the Cox model, the baseline hazard function $\Lambda(y)$ is shared among all subjects, and the hazard of a subject
is the baseline multiplied by a unique tilting function
independent of the survival time (the proportional
hazards assumption).

\subsection{Full likelihood}\label{subsec:full}
The following proposition illustrates a connection between the Cox model and the exponential distribution.
\begin{claim}\label{claim:coxExp}
	Assume the Cox model \eqref{eq:model:cox} and $\lambda(y) > 0$ for $y \ge 0$. Then $\Lambda(Y(0)) \mid X = x$ and $\Lambda(Y(1)) \mid X = x$ follow the exponential distribution with rate $e^{\eta_0(x)}$ and $e^{\eta_1(x)}$, respectively.
\end{claim}
If there is no censoring and the hazard function $\lambda(y)$ is known, Claim \ref{claim:coxExp} implies the Cox model estimation can be simplified to the case based on exponential distribution with responses $\Lambda(Y)$ and the log-link.

Now we discuss the Cox model with censoring under Assumption~\ref{assu:censor}. Similar to Section~\ref{sec:method}, in Proposition \ref{prop:derivationCox} in the appendix, we derived a result analogous to \eqref{eq:condition4}, 
\begin{equation}\label{eq:conditionCox3}
	\begin{aligned}
		a(x) 
		&= \frac{e(x)\PP(C \ge Y \mid W=1, X=x )}{e(x)\PP(C \ge Y \mid W=1, X=x ) + (1-e(x))\PP(C \ge Y \mid W=0, X = x)}, \\
		\nu(x) 
		&= a(x) \eta_1(x)+ (1-a(x))\eta_0(x),
	\end{aligned}
\end{equation}
under the reparametrization \eqref{eq:model:exponentialFamily2} to provide protection to misspecified nuisance functions. 
The $a(x)$ depends on the propensity score as well as the probability of being censored.
Plugging the $a(x)$, $\nu(x)$ in \eqref{eq:conditionCox3} into Algorithm~\ref{algo:exponentialFamily} leads to the hazard ratio estimator of the Cox model (Algorithm~\ref{algo:coxFull}).
\begin{algorithm}
	\BlankLine
	\caption{Cox model with full likelihood}\label{algo:coxFull}
	\textbf{Input}: baseline hazard function $\lambda(y)$ (or the baseline cumulative hazard function $\Lambda(y)$).\\
	We split the data randomly equally into two folds.\\
	1. On fold one, 
	\begin{enumerate}
		\item {Estimation of $e(x)$}. The same as Algorithm~\ref{algo:exponentialFamily}.
		\item {Estimation of $\eta_0(x)$, $\eta_1(x)$}. We maximize the full likelihood on the control and treatment group respectively provided with the hazard function $\lambda(y)$, and obtain estimators $\hat{\eta}_0(x)$, $\hat{\eta}_1(x)$. 
		\item {Estimation of the probability of not being censored}. We estimate the not-censored probabilities $\widehat{\PP}(C \ge Y \mid X = x, W = w)$ by closed form solution, numerical integral, or classifiers with response $\Delta$ and predictors $(X,W)$.
		\item {Substitution}. We construct $\hat{a}(x)$ and $\hat{\nu}(x)$ according to \eqref{eq:conditionCox3} based on $\hat{\eta}_0(x)$, $\hat{\eta}_1(x)$, $\hat{e}(x)$, and $\widehat{\PP}(C \ge Y \mid X = x, W = w)$.
	\end{enumerate} 
	2. On fold two, we plug in $\hat{a}(x)$ and $\hat{\nu}(x)$ to
	estimate $\tau(x)$ by maximizing the full likelihood
	\begin{align}\label{eq:MLECoxFull}
		\begin{split}
		\max_{\beta'} \ell(\beta'; \hat{a}(x), \hat{\nu}(x)) 
		:=&  \frac{1}{n}\sum_{i=1}^n \Delta_i \left(\hat{\nu}(X_i) + (W_i-\hat{a}(X_i))X_i^\top \beta'\right) \\
		&- \Lambda(Y_i^c)e^{\hat{\nu}(X_i) + (W_i-\hat{a}(X_i)) X_i^\top \beta'}. 
		\end{split}
	\end{align}
	
	3. We swap the folds and obtain another estimate. We average the two estimates and output.
\end{algorithm}

In Figure~\ref{fig:errorToy4}, we illustrate the efficacy of the proposed method applied to the Cox model with known baseline hazard. We consider the baseline hazard function $\lambda(y) = y$, and uniform censoring with $5\%$, $50\%$, $95\%$ censored units. The tilting functions $\eta_0(x)$, $\eta_1(x)$ are non-linear, and we approximate them by linear functions. We observe that regardless of the magnitude of censoring, the proposed method makes significant improvements over the separate estimation method. As the censoring gets heavier, $a(x)$ in \eqref{eq:conditionCox3} differs more from $e(x)$, and thus the proposed method outperforms the direct extension of R-learner by larger margins. 

\begin{figure}
	\centering
	\begin{minipage}{4.6cm}
		\centering  
		\includegraphics[scale=0.3]{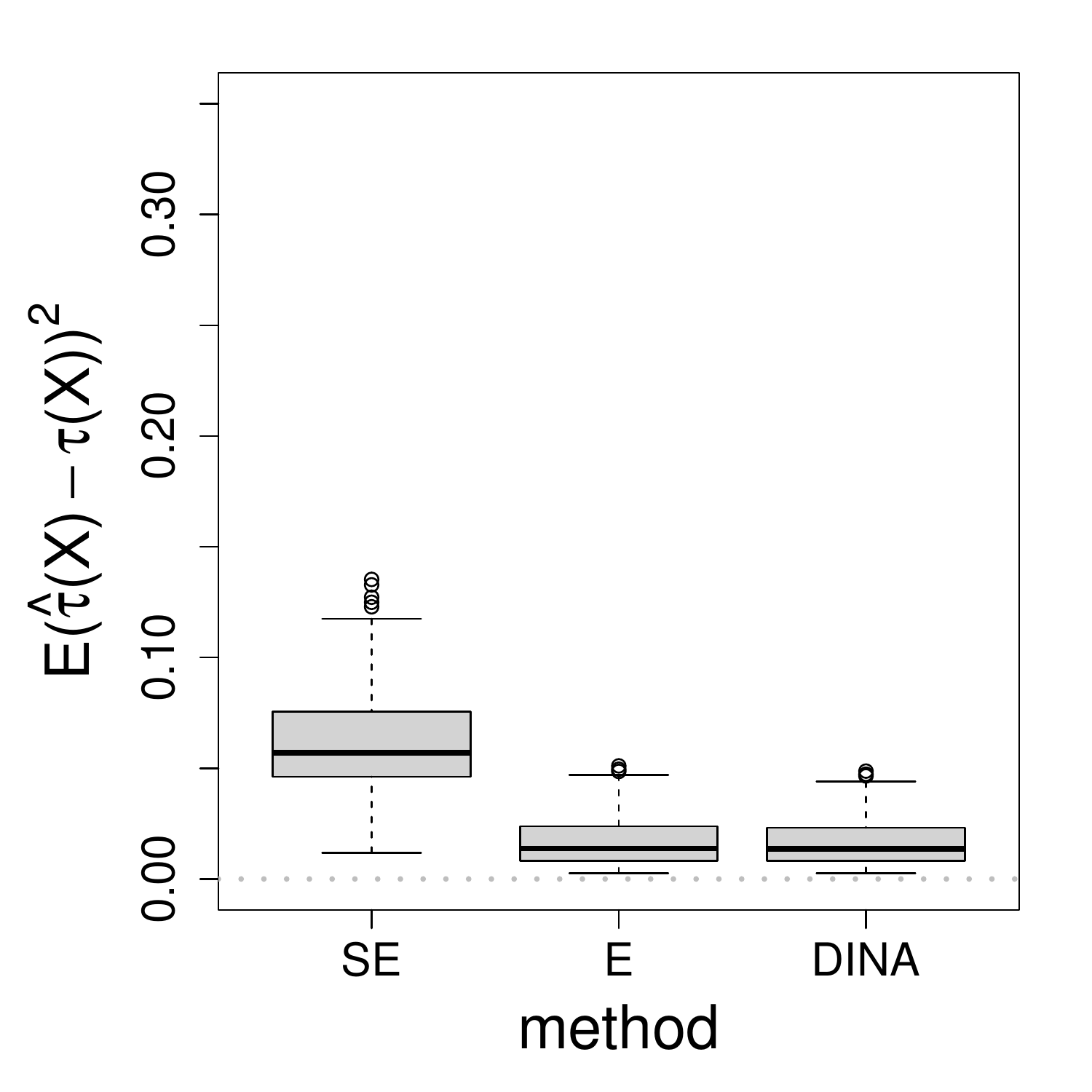}  
		{{\quad \quad (a) $5\%$ censored}}
	\end{minipage}
	\begin{minipage}{4.6cm}
		\centering  
		\includegraphics[scale=0.3]{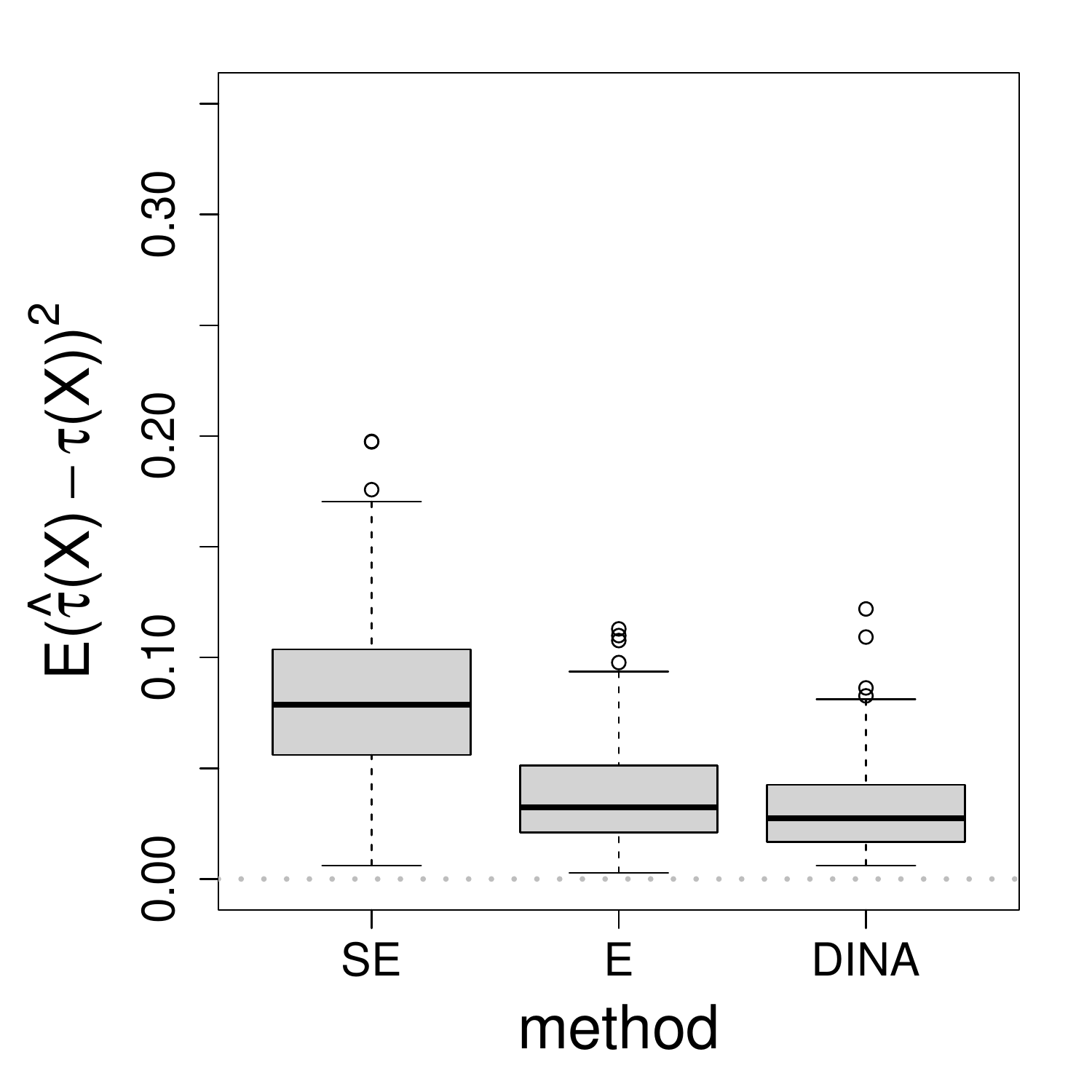}  
		{{\quad \quad (b) $50\%$ censored}}
	\end{minipage}
	\begin{minipage}{4.6cm}
		\centering  
		\includegraphics[scale=0.3]{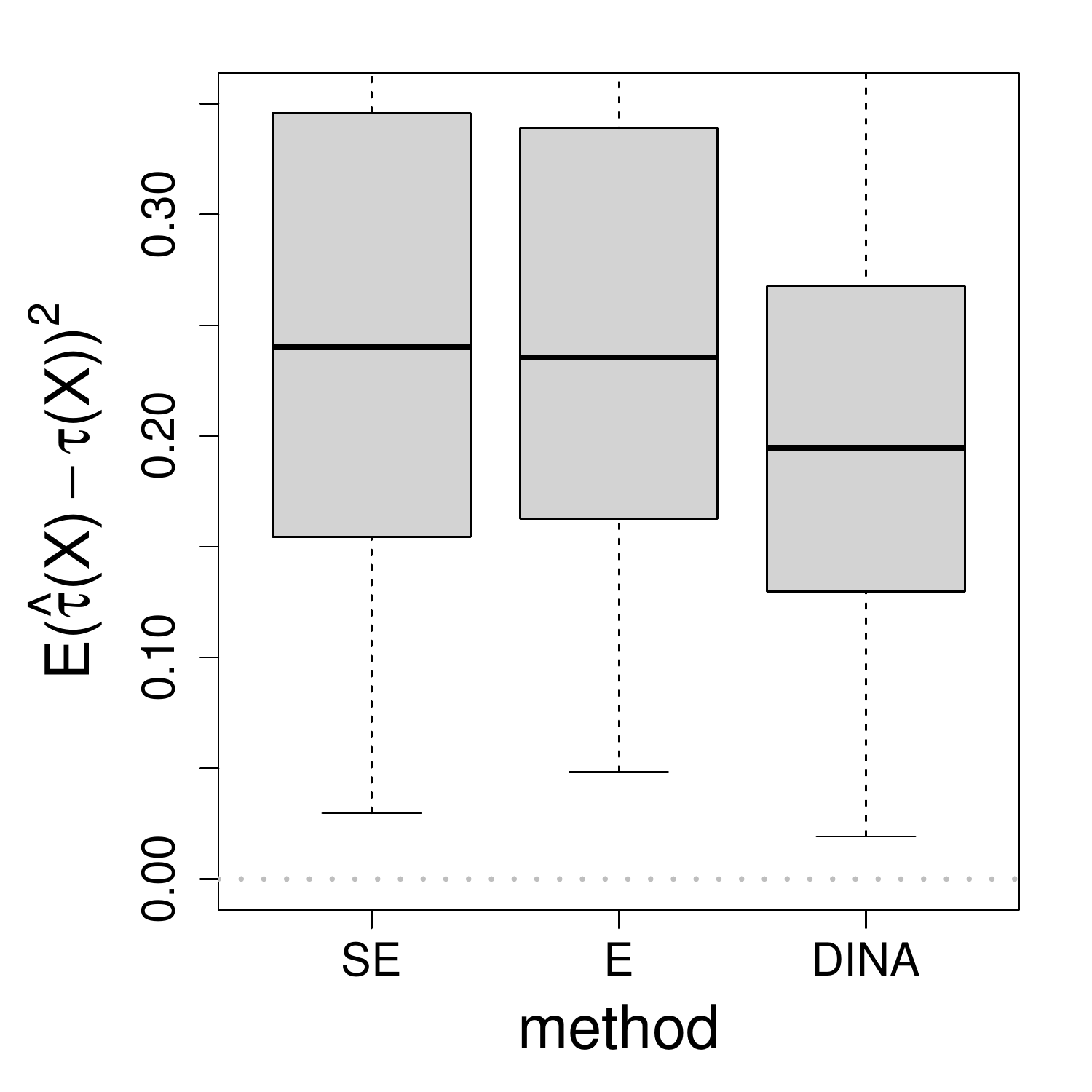}
		{{\quad (c) $95\%$ censored}}
	\end{minipage}
	\caption{{Estimation error box plots of separate estimation (SE), the direct extension of R-learner (E), and the proposed method (DINA) with the Cox model via full likelihood maximization (Algorithm~\ref{algo:coxFull}). From the left to right, $5\%$, $50\%$, $95\%$ units are censored. The baseline hazard is $\lambda(y) = y$ and is provided to the estimators. The tilting functions $\eta_0(x)$, $\eta_1(x)$ are non-linear, and we approximate them by linear functions. }}
	\label{fig:errorToy4}
\end{figure}

We discuss several special instances of \eqref{eq:conditionCox3}. In particular, we consider the subcases that the censoring times satisfy $C(0) \mid X \stackrel{d}{=} C(1) \mid X$, i.e., the conditional distributions of the censoring times are the same across the treatment and control arms. 

\begin{enumerate}
	\item {No treatment effect}. If there is no treatment effect, then the distributions of $Y$ and $C$ are conditionally independent of the treatment assignment indicator $W$. As a result, the probability ratio of not being censored is one and $a(x) = e(x)$ --- R-learner.
	\item {Light censoring}. If the proportion of censored units is small, the censoring time will be longer than the survival time despite the treatment effect, suggesting the ratio of not being censored to be close to one. Consequently, the case approximately reduces to R-learner.
	\item {Heavy censoring}. If the proportion of censored units is high, then as argued by \citet{lin2013bias}, almost all the information is contained in the indicator $\Delta$, and we can directly model the rare event --- ``not censored" by Poisson distribution. In this case, 
	\begin{align}\label{eq:conditionCox4}
		\begin{split}
			\frac{\PP(C \ge Y \mid W=1, X=x )}{\PP(C \ge Y \mid W=0, X=x )} \stackrel{\text{c.p.}\to 1\footnotemark}{\longrightarrow}  e^{\tau(x)},
		\end{split}
	\end{align}
	and plug \eqref{eq:conditionCox4} in \eqref{eq:conditionCox3}
	gives the $a(x)$ and  $\nu(x)$ corresponding to a Poisson distribution.
	\footnotetext{c.p. stands for censored probability.}
\end{enumerate}

To provide intuition for Algorithm~\ref{algo:coxFull}, we compare the nuisance-function construction \eqref{eq:conditionCox3} with that of the direct extension of R-learner and the DINA estimator for the exponential family (Algorithm \ref{algo:exponentialFamily}).
Different from R-learner,  the $a(x)$ in \eqref{eq:conditionCox3} also relies on the censoring probabilities under treatment and control. For a unit with covariate value $x$, the multiplier $(w-a(x))$ associated with the HTE in \eqref{eq:model:exponentialFamily2} is smaller for $w = 1$ compared to $w = 0$ if the unit tends to be censored under treatment. Consequently, the weight $a(x)$ emphasizes more the units not censored, which agrees with the common sense that censored units contain less information.
Compared with \eqref{eq:condition4}, \eqref{eq:conditionCox3} does not
include the derivatives $\frac{d\mu}{d\eta}(\eta_w(x))$ due to
different parameter scales (variances): in the Cox model, we focus on the tilting functions $\eta_w(x)$ instead of the natural parameters $e^{\eta_w(x)}$ in the exponential family.

We now probe into the implementation details in
Algorithm~\ref{algo:coxFull}. Algorithm~\ref{algo:coxFull} is
reminiscent of Algorithm~\ref{algo:exponentialFamily} except in the
estimation of $a(x)$, and in particular, the estimation of the censoring probabilities. We list several common random censoring mechanisms and the associated probabilities of ``not censored'' in Table \ref{tab:example2}. If the censoring mechanism is known, for example, all the units are enrolled simultaneously and operate for the same amount of time (singly-censored), then we can compute the censoring probabilities in closed form or via numerical integration given rough nuisance-function estimators $\hat{\eta}_0(x)$, $\hat{\eta}_1(x)$. If the censoring mechanism is unknown, we can directly estimate the censoring probabilities from the data: solving the classification problem with the response $\Delta$, predictors $(X,W)$ by any classifier that comes with predicted class probabilities. 
If the censoring is extremely low or high, we can turn to the aforementioned special instances.

\begin{table}
	\caption{\label{tab:example2}{Examples of censoring mechanisms and the associated probabilities of ``not censored".}}	
	\centering
	\scriptsize
	\hspace{-0cm}
	\fbox{%
	\begin{tabular}{c|c|c|c|c}
		\multicolumn{2}{c|}{Censoring}  & Parameters           & Density of censoring & Probability of ``not censored"\\
		\multicolumn{2}{c|}{mechanism}     &      &   times  $f_C(c)$    & $\PP(C \ge Y \mid W = w, X = x)$             \\ \hline		\multicolumn{2}{c|}{No censoring}        & -   &    $0$    &    $1$         \\ \hline\multicolumn{2}{c|}{Singly-censored}        & $T$    &    $\delta_{\{c = T\}}$     &    $1 - e^{-\Lambda(T) \eta_w(x)}$         \\ \hline
		\multirow{2}{*}{\shortstack{Multiply- \\ censored}} & Uniform   &$T$  &  $1/T$  &     $1 - \frac{1}{T}\int_0^T e^{-\Lambda(c)\eta_w(x)}  dc$   \\ \cline{2-5} 
		& Weibull   &$r$, $k$  & $k r (r c)^{k-1} e^{-(r c)^k} $      & $1 - k r\int_0^\infty e^{-\Lambda(c)\eta_w(x)}  (r c)^{k-1} e^{-(r c)^k} dc$           
	\end{tabular}}
\end{table}

In the following proposition, we extend the robustness property of Algorithm~\ref{algo:exponentialFamily} in Proposition~\ref{prop:main} to Algorithm~\ref{algo:coxFull}.
\begin{proposition}
	\label{prop:mainCoxFull}
	Under the regularity conditions of Proposition~\ref{prop:main}, and assume the baseline hazard function $\lambda(y)$ is known, $\|{a}_n(x) - a(x)\|_2$, $\|{\nu}_n(x) - \nu(x)\|_2 = O(c_n)$, $c_n \to 0$, then
	\begin{align*}
		\|{\beta}_n - \beta\|_2 = \tilde{O}\left(c_n^2 + n^{-1/2}\right)
	\end{align*}
	for Algorithm~\ref{algo:coxFull}.
\end{proposition}

To end the section, we discuss how to carry out the estimation when the baseline hazard $\Lambda(y)$ is inaccessible. One option is to start with estimating the baseline hazard, and plug it in for subsequent procedures. Another option is using the partial likelihood that cleverly avoids the baseline hazard (see Section~\ref{subsec:partial} for more details). We remark that the proposed method is not provably robust to the baseline hazard misspecification.

\subsection{Partial likelihood}\label{subsec:partial}

Instead of focusing on the full likelihood, \citet{cox1972regression} proposes to maximize the partial likelihood. The partial likelihood is prevalent in practice because it does not require the baseline hazard function and preserves promising statistical properties \citep{tsiatis1981large,andersen1982cox}. 

\citet{van2000asymptotic} shows that the partial likelihood $\text{pl}_n$ can be obtained from the full likelihood \eqref{eq:cox:fulllike} by profiling out the baseline hazard 
\begin{align*}
	\text{pl}_n(\eta_w(x)) = \sup_{\{\Lambda_i \}} \ell_n(\eta_w(x), \{\Lambda_i\}), \quad w \in \{0,1\},
\end{align*}
where $\Lambda_i$ denotes the cumulating hazard at $Y_i^c$ if the subject $i$ is not censored. Let $\hat{\Lambda}_i$ be the baseline hazard estimators associated with the partial likelihood maximization. As a corollary, the partial likelihood maximizer is equivalent to that of the full likelihood with $\hat{\Lambda}_i$. The connection 
motivates Algorithm~\ref{algo:coxPartial} --- an application of Algorithm~\ref{algo:coxFull} with the true hazard baseline function replaced by $\hat{\Lambda}_i$ from the partial likelihood maximization.

\begin{algorithm}
	\DontPrintSemicolon  
	\SetAlgoLined
	\BlankLine
	\caption{Cox model with partial likelihood}\label{algo:coxPartial}
	We split the data randomly equally into two folds.\\
	1. On fold one, 
	\begin{enumerate}
		\item {Estimation of $e(x)$}. The same as Algorithm~\ref{algo:exponentialFamily}.
		\item {Estimation of $\eta_0(x)$, $\eta_1(x)$}. We maximize the partial likelihoods on the control and treatment group respectively, and obtain estimators $\hat{\eta}_0(x)$, $\hat{\eta}_1(x)$.
		\item {Estimation of the probability of not  censored}. We estimate the not censored probabilities $\widehat{\PP}(C \ge Y \mid X = x, W = w)$ by closed form solution, numerical integral, or classifiers with response $\Delta$ and predictors $(X,W)$;
		\item {Substitution}. We construct $\hat{a}(x)$
		and $\hat{\nu}(x)$ according to
		\eqref{eq:conditionCox3}  based on $\hat{\eta}_0(x)$, $\hat{\eta}_1(x)$, $\hat{e}(x)$, and $\widehat{\PP}(C \ge Y \mid X = x, W = w)$.
	\end{enumerate} 
	2. On fold two, we plug in $\hat{a}(x)$ and $\hat{\nu}(x)$ to estimate $\tau(x)$ by maximizing the partial log-likelihood,
	\begin{align}\label{eq:MLECoxPartial}
		\begin{split}
			\min_{\beta'}\text{pl}_n(\beta'; \hat{a}(x), \hat{\nu}(x)) := \frac{1}{n}\sum_{\Delta_i = 1} \left(\hat{\nu}(X_i) + (W_i-\hat{a}(X_i))X_i^\top \beta' \right.\\
			\left.- \log\left(\sum_{\calR_i} e^{\hat{\nu}(X_i) + (W_i-\hat{a}(X_i))X_i^\top \beta'}\right) \right),
		\end{split}
	\end{align}	
	where $\calR_i = \{j: Y_j^c \ge Y_i^c\}$ denotes the risk set of subject $i$.
	
	3. We swap the folds and obtain another estimate. We average the two estimates and output.
\end{algorithm}

We first numerically compare the estimators based on the full likelihood (given the true baseline hazard function) and the partial likelihood. We consider three  cumulative baseline hazard functions $\Lambda(y) = y$, $y^2$, and $y^5$, corresponding to Weibull distributions with shape parameters $1$, $2$, and $5$. We adopt the uniform censoring and fix the proportion of censored units at $50\%$. The baseline parameters, the separate estimation method, and the proposed method based on the full likelihood are the same as Figure~\ref{fig:errorToy4}. 
In Figure~\ref{fig:errorToy5}, the partial likelihood based estimator performs comparatively to that based on the full likelihood despite the form of the baseline hazard function, and significantly improves the separate estimation method.

\begin{figure}
	\centering
	\begin{minipage}{4.6cm}
		\centering  
		\includegraphics[scale=0.3]{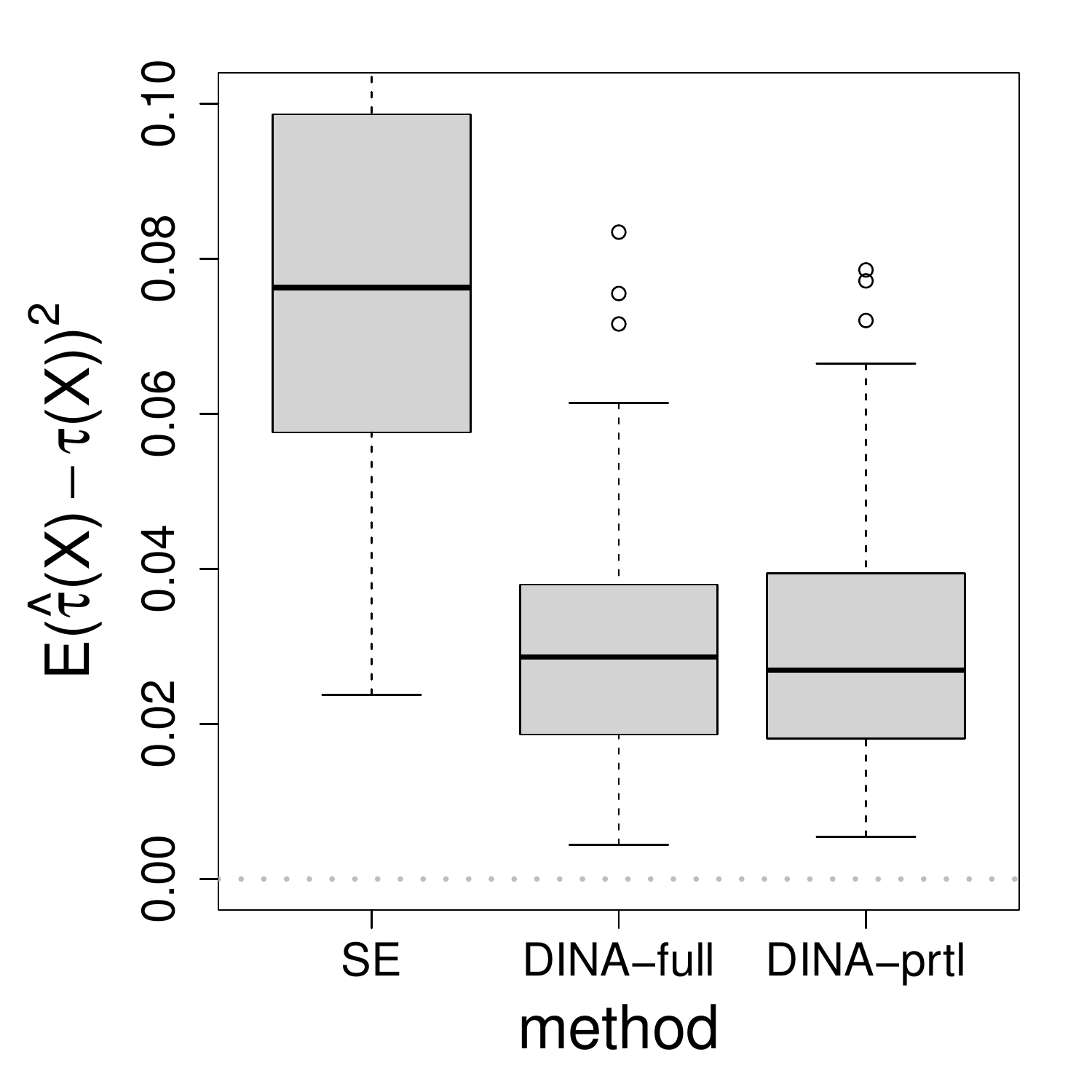}  
		{{\quad \quad (a) Weibull shape $k=1$}}
	\end{minipage}
	\begin{minipage}{4.6cm}
		\centering  
		\includegraphics[scale=0.3]{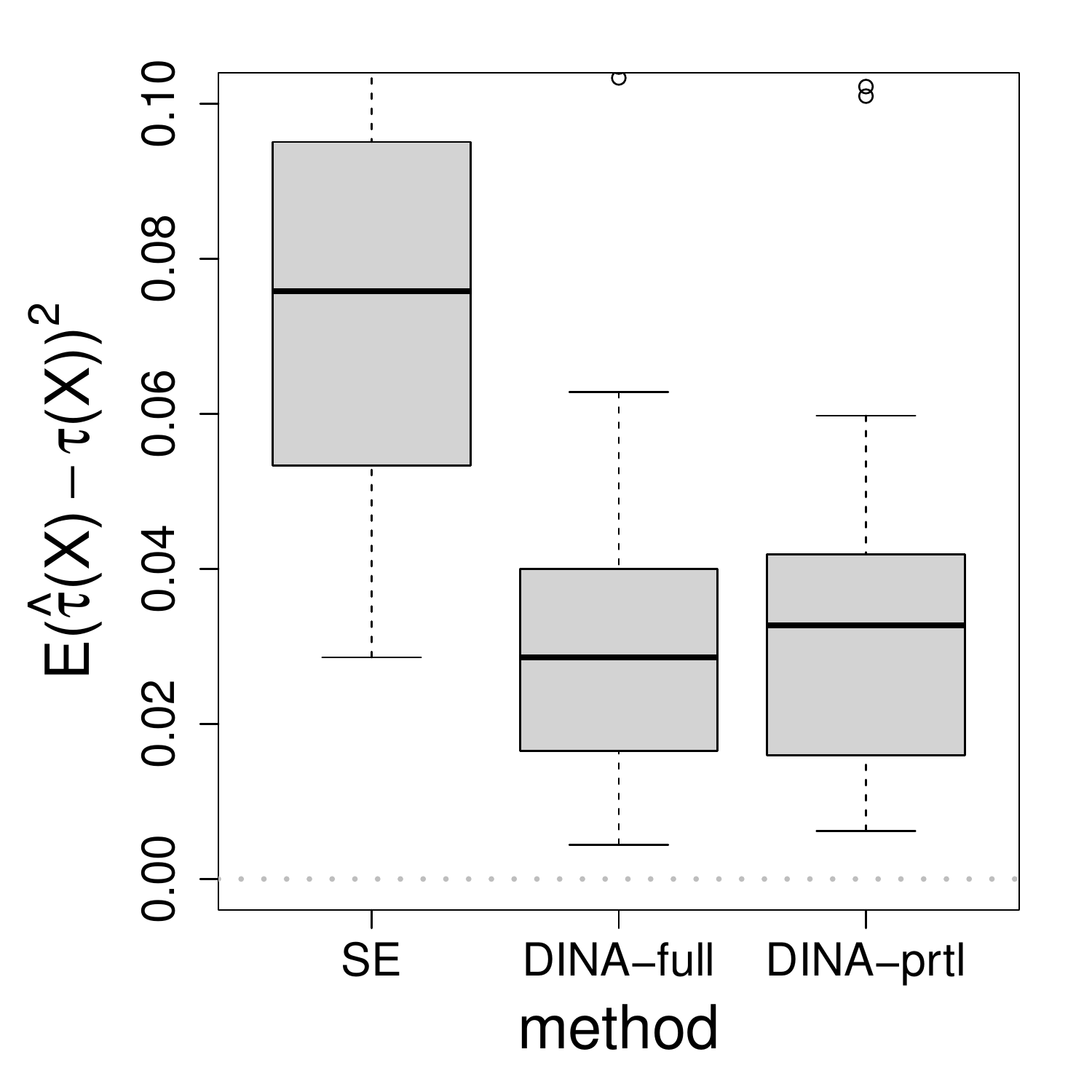}  
		{{\quad \quad (b) Weibull shape $k=2$}}
	\end{minipage}
	\begin{minipage}{4.6cm}
		\centering  
		\includegraphics[scale=0.3]{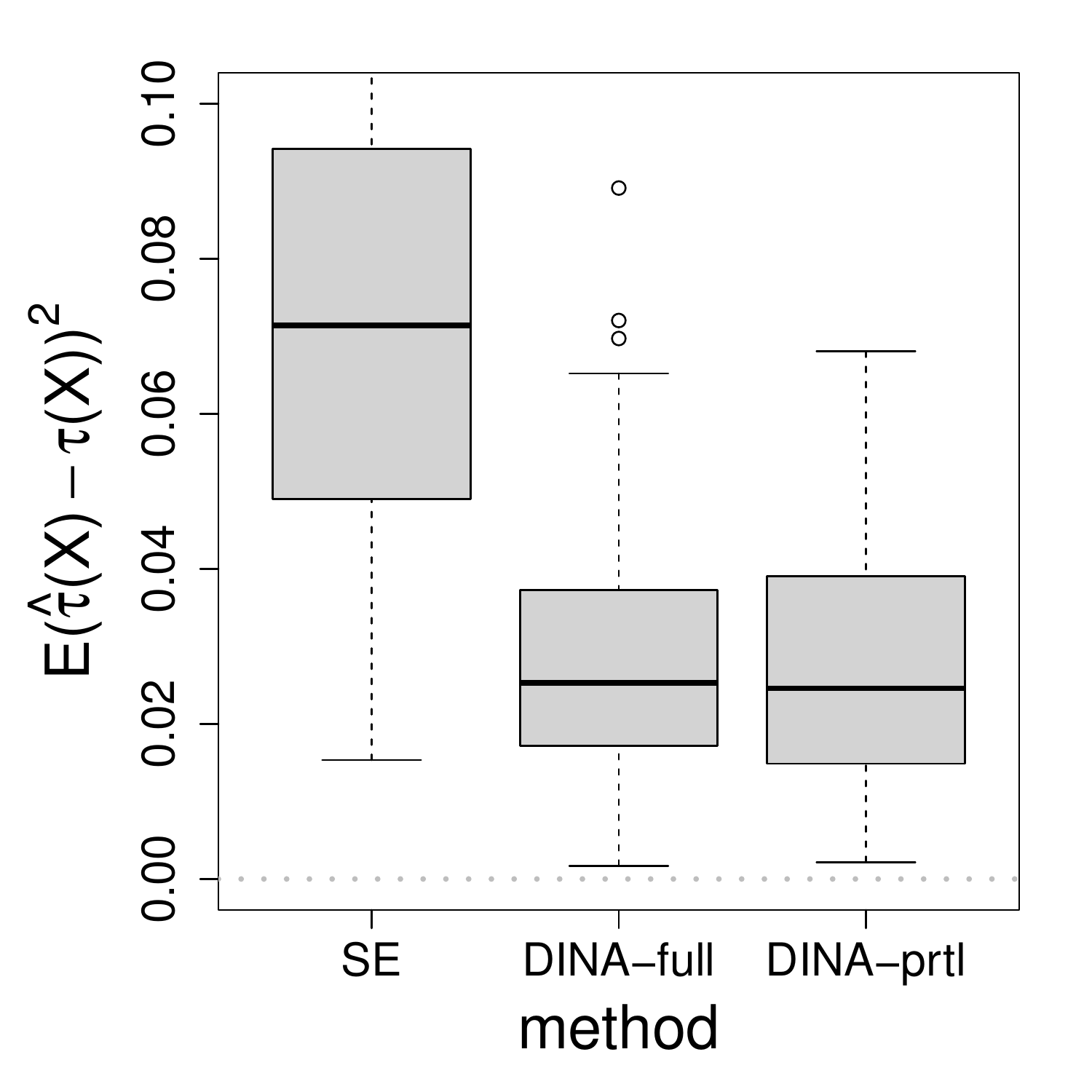}
		{{\quad (c) Weibull shape $k=5$}}
	\end{minipage}
	\caption{{Estimation error box plots of separate estimation (SE), the proposed method with full likelihood (DINA-full), and partial likelihood (DINA-prtl). From the left to right, the cumulative baseline hazard $\Lambda(y) = y$, $y^2$, and $y^5$, corresponding to Weibull distributions with shape parameters $1$, $2$, and $5$. We adopt the uniform censoring and fix the proportion of censored units at $50\%$. The baseline parameters, separate estimation method, and the proposed method with full likelihood are the same as Figure~\ref{fig:errorToy4}. 
	}}
	\label{fig:errorToy5}
\end{figure}

We investigate the theoretical properties of Algorithm~\ref{algo:coxPartial}.
Under the null hypothesis, i.e., no treatment effect, the robustness in Proposition~\ref{prop:main} is valid.

\begin{proposition}
	\label{claim:cox:null}	
	Under the regularity conditions of Proposition~\ref{prop:main}, and assume there is no treatment effect, $\|{a}_n(x) - a(x)\|_2$, $\|{\nu}_n(x) - \nu(x)\|_2 = O(c_n)$, $c_n \to 0$, then
	\begin{align*}
		\|{\beta}_n - \beta\|_2 = \tilde{O}\left(c_n^2 + n^{-1/2}\right)
	\end{align*}
	for Algorithm~\ref{algo:coxPartial}.
\end{proposition}

For non-zero treatment effects, Proposition~\ref{claim:cox:null} is in general not true. 
Despite the lack of theoretical guarantee, Algorithm~\ref{algo:coxPartial} produces promising results over simulated datasets. In Figure~\ref{fig:errorToy5}, the partial-likelihood loss in performance is slight or even negligible compared to Algorithm~\ref{algo:coxFull}, and
it requires no baseline hazards knowledge --- one of the fundamental merits of Cox model. Section~\ref{sec:simulation} contains more empirical examples of Algorithm~\ref{algo:coxPartial}. Therefore, in many applications where baseline hazards are unavailable, we recommend Algorithm~\ref{algo:coxPartial}.

\section{Simulation}\label{sec:simulation} 

In this section, we demonstrate the efficacy of the proposed method using simulated datasets.

\subsection{Exponential family}\label{subsec:exponentialFamily}
We compare the following five meta-algorithms.
\begin{enumerate}
	\item {Separate estimation (SE-learner)}. The separate estimation method estimates the control and treatment group mean functions $\hat{\eta}_0(x)$, $\hat{\eta}_1(x)$, takes the difference $\hat{\eta}_1(x) - \hat{\eta}_0(x)$, and further regresses the difference on the covariates to obtain $\hat{\beta}$. 
	The nuisance functions are the control and treatment group mean functions $\eta_0(x)$, $\eta_1(x)$ and the method does not require the propensity score.
	
	\item {X-learner (X-learner)}. 
	X-learner first estimates the control group mean function $\hat{\eta}_0(x)$, and then estimates $\hat{\beta}$ by solving a generalized linear model with $\hat{\eta}_0(x)$ as the offset.
	The nuisance functions is the control group mean function $\eta_0(x)$ and the method does not require the propensity score.
	
	\item {Propensity score adjusted X-learner (PA-X-learner)}. Motivated by a thread of works \citep{vansteelandt2014regression, dorie2019automated, hahn2020bayesian} including estimated propensity scores as a covariate, we consider an augmented X-learner where the control group mean function is learnt as a function of raw covariates, an estimated propensity score, and the interaction between. The rest of the approach is the same as X-learner above.
	The nuisance functions is the control group mean function $\eta_0(x)$ and the propensity score $e(x)$.	
	
	\item {Direct extension of R-learner (E-learner)}.  The direct extension of R-learner considers $a(x) = e(x)$ and the associated baseline $m(x) = (1-e(x))\eta_0(x) + e(x) \eta_1(x)$
	for arbitrary response types. The rest is the same as Algorithm~\ref{algo:exponentialFamily}.
	The nuisance parameters are $e(x)$ and $m(x)$.  
	
	We remark that the above extension to the natural parameter scale is different from the original R-learner which focuses on the difference in conditional means despite the type of responses.  
	\item {The proposed method (DINA-learner)}. We apply Algorithm~\ref{algo:exponentialFamily} with the $a(x)$, $\nu(x)$ in \eqref{eq:condition4}. For Gaussian responses, the proposed method and R-learner are the same; for other distributions, the two are different.
\end{enumerate}
In the simulations below, we obtain $\hat{e}(x)$ by logistic regression and $\hat{\eta}_0(x)$, $\hat{\eta}_1(x)$ by fitting generalized linear models. Results of estimating $\hat{\eta}_0(x)$, $\hat{\eta}_1(x)$ by tree boosting are available in Figure \ref{fig:errorBoosting} in the appendix. 

As for data generating mechanism, we consider $d  = 5$ covariates independently generated from uniform $[-1,1]$.  The treatment assignment follows a logistic regression model. The responses are sampled from the exponential family with natural parameter functions
\begin{align}\label{eq:simulationModel}
	\begin{split}
		\begin{cases}
			\eta_0(x) = x^\top\alpha + \delta x_1 x_2, \\
			\eta_1(x) =  x^\top(\alpha + \beta) + \delta x_1 x_2,
		\end{cases} 
	\end{split}
\end{align}
for some $\delta \neq 0$.
In both treatment and control groups, the response models are misspecified generalized linear models, while the difference of the natural parameters $\tau(x)$ is always linear. We consider continuous, binary, and discrete responses generated from Gaussian, Bernoulli, and Poisson distributions, respectively. 
We quantify the signal magnitude by the following signal-to-noise ratio (SNR)
\begin{align*}
	\frac{\var(\EE[Y\mid X,W])}{\EE[\var(Y \mid X,W)]}.
\end{align*}
SNRs of all simulation settings are approximately $0.5$.

We measure the estimation performance by the mean squared error $\EE[(\hat{\tau}(X) - \tau(X))^2]$, where the expectation is taken over the covariate population distribution. Results are summarized in Figure~\ref{fig:error}. 
Across three types of responses, our proposed method (DINA), direct extension of R-learner (E), and propensity score adjusted X-learner (PA-X) performs relatively better than X-learner (X) and separate estimation (SE). 
Among the three well-performed methods, our proposed method approximately achieves the parametric convergence rate $O(n^{-1/2})$ and is more favorable for count data. 
As for X-learner and separate estimation, the errors stop decreasing as the sample size increases due to the non-vanishing bias.

\begin{figure}
	\centering
	\begin{minipage}{7cm}
		\centering  
		{(a) {continuous}}
		\includegraphics[scale=0.45]{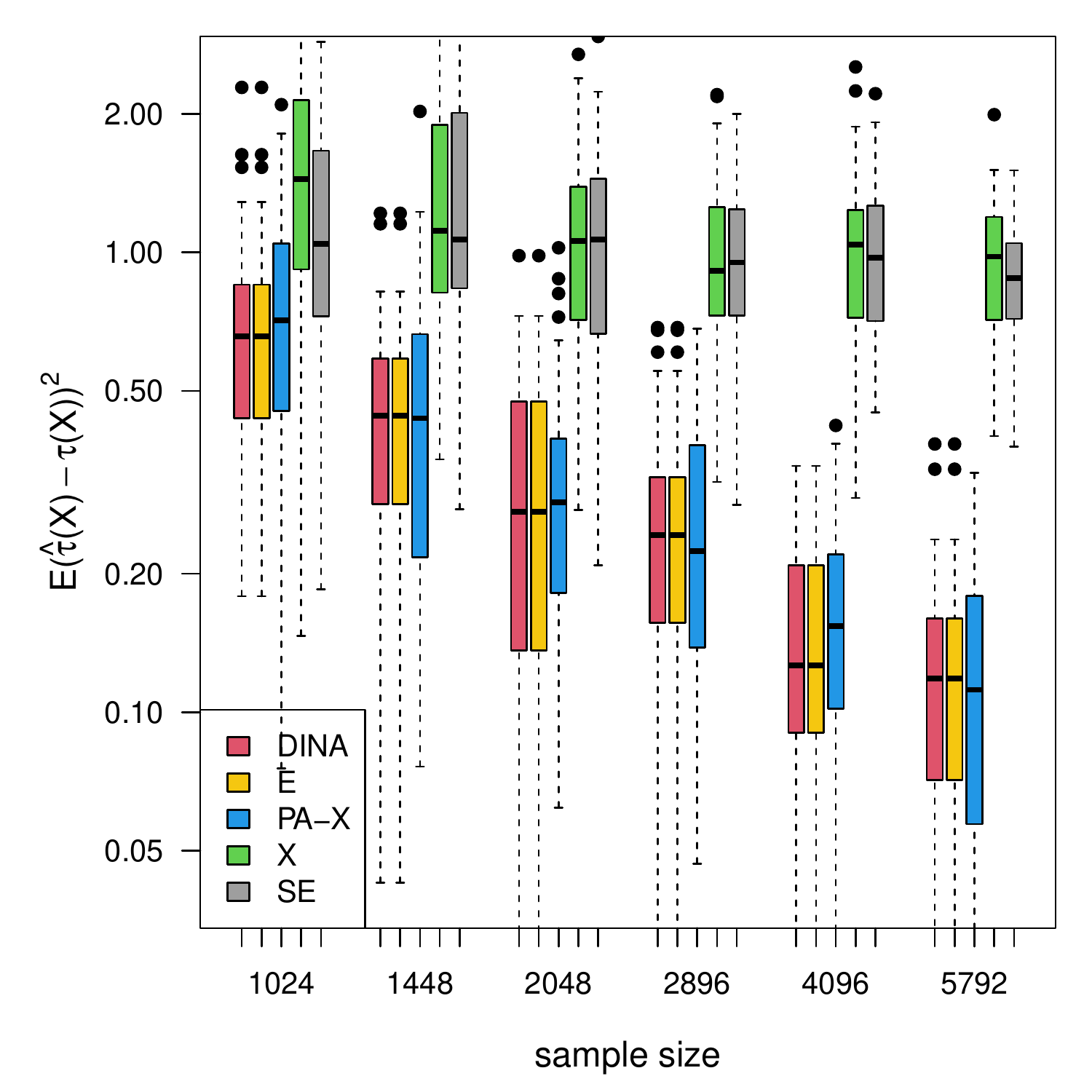}  
	\end{minipage}
	\begin{minipage}{7cm}
		\centering  
		{(b) {binary}}
		\includegraphics[scale=0.45]{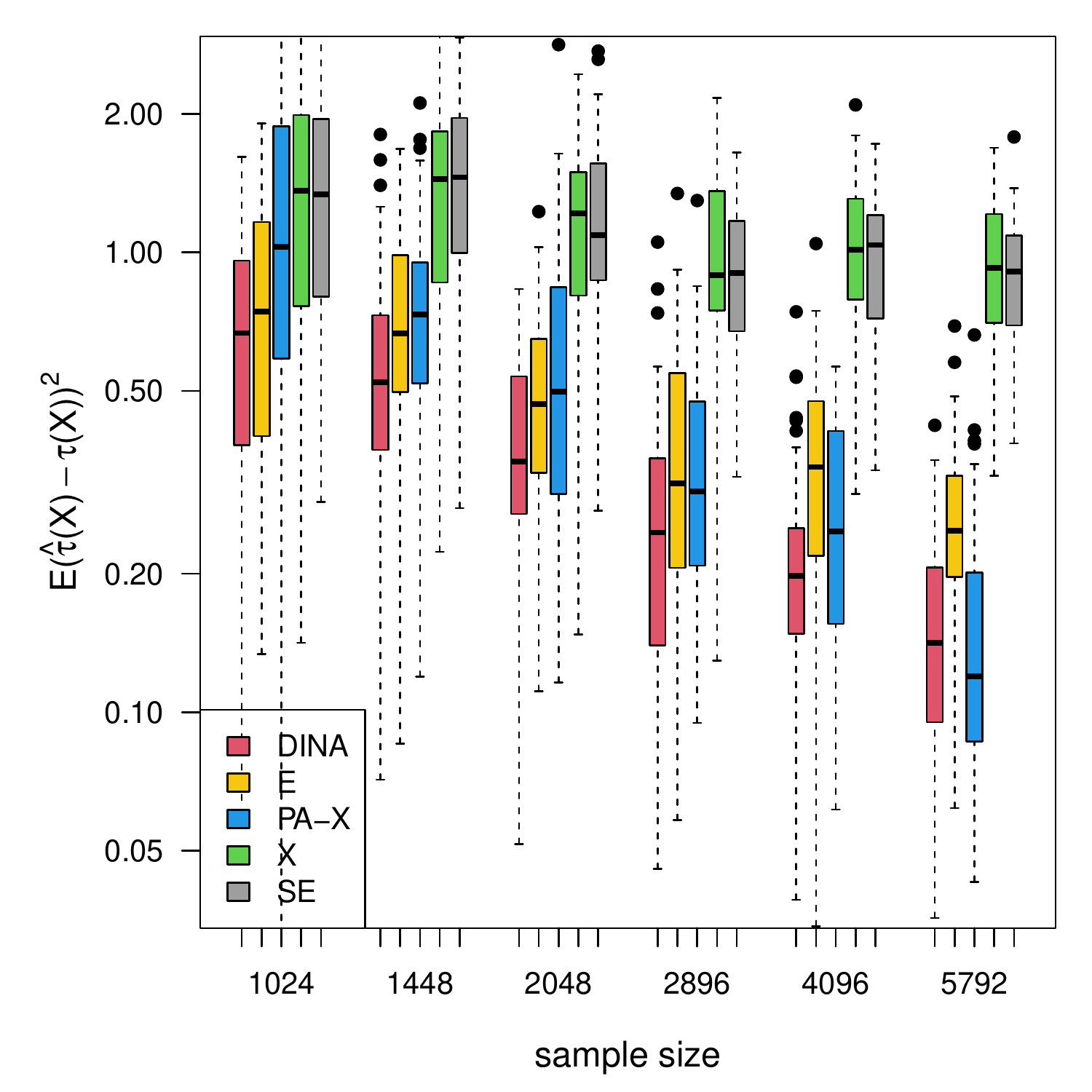}  
	\end{minipage}\\
	\begin{minipage}{7cm}
		\centering  
		{(c) {count data}}
		\includegraphics[scale=0.45]{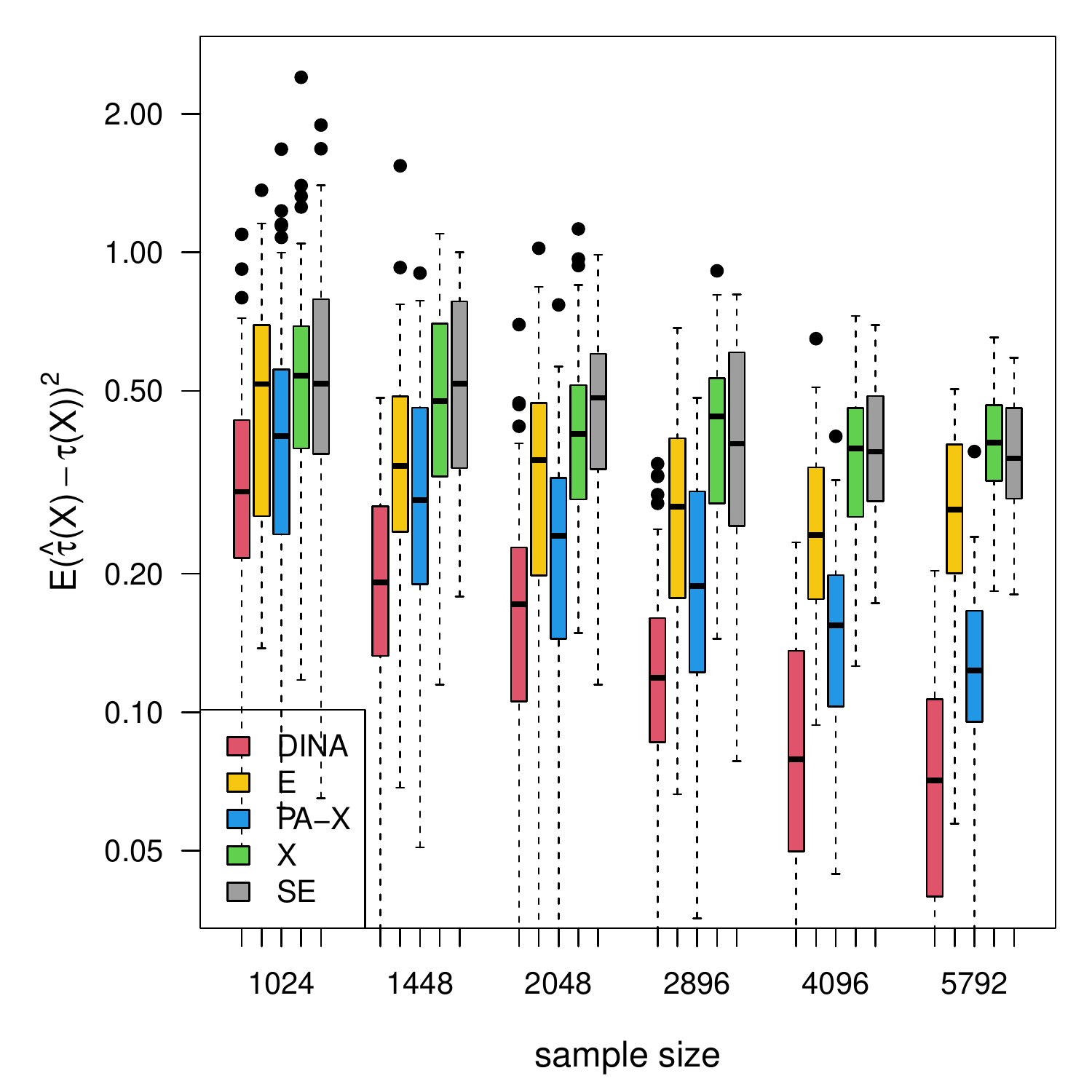}  
	\end{minipage}
	\begin{minipage}{7cm}
		\centering  
		{(d) {survival data}}
		\includegraphics[scale=0.45]{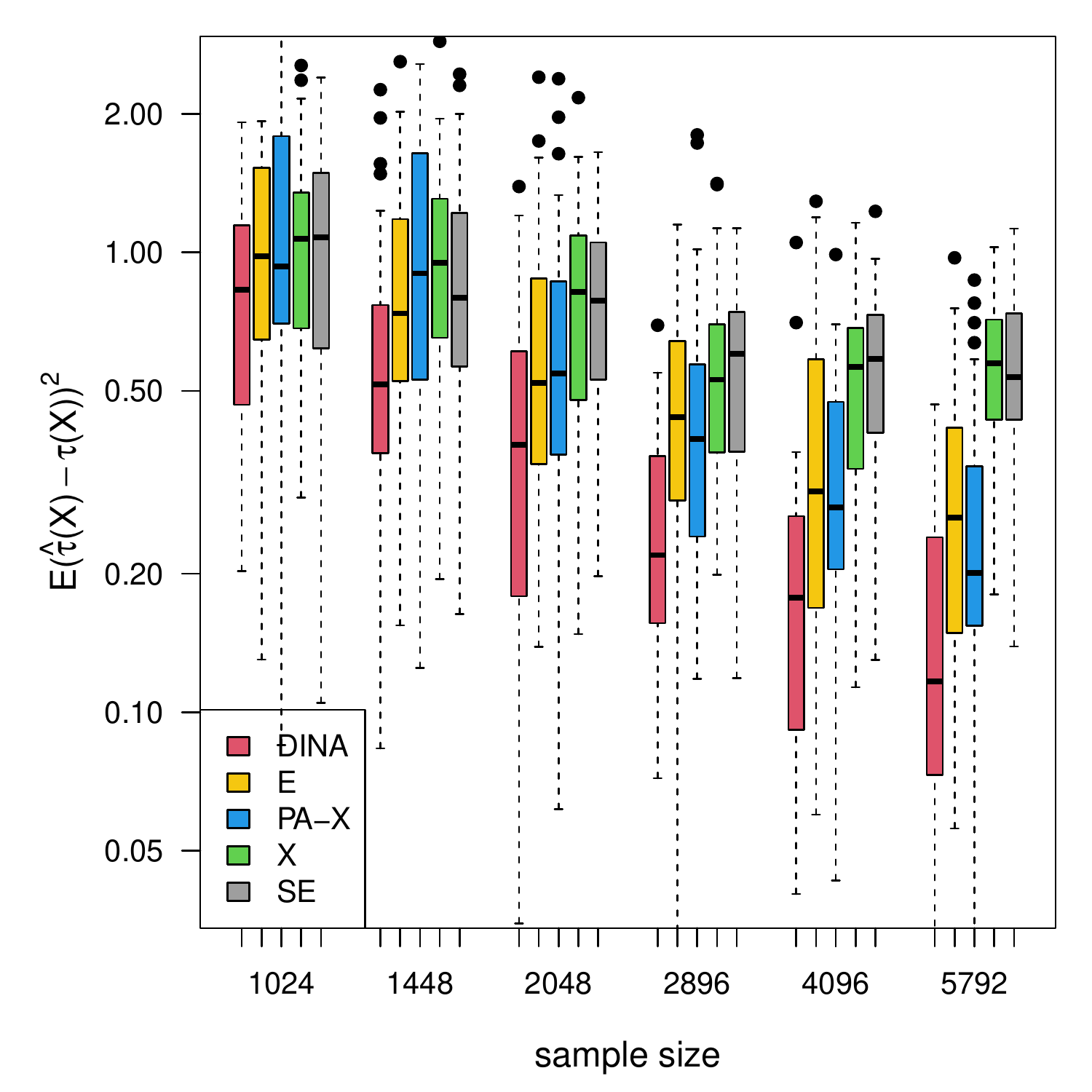}  
	\end{minipage}
	\caption{{Estimation error log-log boxplots. We display
			the estimation errors $\EE[(\hat{\tau}(X) - \tau(X))^2]$
			over sample sizes in $[1024, 1448, 2048, 2896, 4096, 5792]$. We compare five
			methods: separate estimation (SE), X-learner (X),  propensity score adjusted X-learner (PA-X), direct extension of R-learner (E),  and Algorithm~\ref{algo:exponentialFamily}
			(DINA). We adopt the response model~\eqref{eq:simulationModel} and consider four types of responses: (a) continuous (Gaussian); (b) binary (Bernoulli); (c) count data (Poisson); (d) survival data (Cox model with uniform censoring, $75\%$ units censored). We estimate the propensity score by logistic regression, and estimate the natural parameters by fitting generalized linear models or maximizing Cox partial likelihoods. We repeat all experiments $100$ times.
	}}
	\label{fig:error}
\end{figure}

We use bootstrap ($100$ bootstrap samples) to construct confidence intervals for $\hat{\beta}$. 
We remark that because of sampling with replacement, different folds of data from a bootstrap sample may overlap and the associated $\hat{\beta}^b$ may be seriously biased.
However, $\hat{\beta}^b$ from bootstrap samples can still be used to estimate the standard deviation of $\hat{\beta}$.
Table \ref{tab:poissonGLMCI} demonstrates the coverage and width of the $95\%$ confidence intervals with Poisson responses. We estimate the propensity score by logistic regression and baseline natural parameter functions by Poisson regression. The proposed method (DINA) produces the best coverages for all $\hat{\beta}$ and sample sizes.
The propensity score adjusted X-learner (PA-X) produces the second-best coverages but requires wider confidence intervals. 
For the direct extension of R-learner (E),  X-learner (X), and separate estimation (SE), the coverages of confidence intervals ($\hat{\beta}_1$ in particular) decrease as the sample size grows. The phenomenon is due to the non-vanishing bias in those estimators.
Results of confidence intervals for other types of data and nuisance-function learners are similar and can be found in the appendix.

\begin{table}
	\caption{\label{tab:poissonGLMCI}{Coverage (cvrg) and width of $95\%$ confidence intervals for count data. 
			 We consider the five meta-algorithms and vary the sample size from $1024$ to $5792$. 
			 We estimate the propensity score by logistic regression and baseline natural parameter functions by Poisson regression. Confidence intervals are constructed using $100$ bootstrap samples. Results are averaged over $100$ trials.
		}}
	\centering
	\scriptsize
	\hspace{-0cm}
	\fbox{%
	\begin{tabular}{c|c|cc|cc|cc|cc|cc|cc}
		Sample                    & Meth- & \multicolumn{2}{c|}{$\beta_1$} & \multicolumn{2}{c|}{$\beta_2$} & \multicolumn{2}{c|}{$\beta_3$} & \multicolumn{2}{c|}{$\beta_4$} & \multicolumn{2}{c|}{$\beta_5$} & \multicolumn{2}{c}{$\beta_0$} \\ 
		size &   od     & cvrg      & width     & cvrg      & width             & cvrg      & width           & cvrg      & width      & cvrg      & width     & cvrg      & width          \\\hline		\multirow{5}{*}{1024} & DINA   & 0.92  & 0.679   & 0.92  & 0.716   & 0.93  & 0.646   & 0.93  & 0.631   & 0.96  & 0.642   & 0.94  & 0.379   \\
		& E      & 0.57  & 0.706   & 0.48  & 0.732   & 0.97  & 0.661   & 0.93  & 0.639   & 0.94  & 0.640    & 0.93  & 0.452   \\
		& PA-X   & 0.94  & 0.880    & 0.89  & 0.91    & 0.97  & 0.669   & 0.95  & 0.666   & 0.94  & 0.666   & 0.93  & 0.419   \\
		& X      & 0.08  & 0.636   & 0.94  & 0.691   & 0.95  & 0.597   & 0.92  & 0.603   & 0.91  & 0.584   & 0.91  & 0.417   \\
		& SE     & 0.03  & 0.628   & 0.94  & 0.698   & 0.99  & 0.599   & 0.93  & 0.591   & 0.93  & 0.591   & 0.94  & 0.421   \\\hline	
		\multirow{5}{*}{1448} & DINA   & 0.96  & 0.553   & 0.92  & 0.580    & 0.91  & 0.546   & 0.95  & 0.534   & 0.91  & 0.526   & 0.94  & 0.314   \\
		& E      & 0.44  & 0.569   & 0.45  & 0.605   & 0.91  & 0.540    & 0.98  & 0.536   & 0.94  & 0.524   & 0.96  & 0.372   \\
		& PA-X   & 0.87  & 0.734   & 0.86  & 0.746   & 0.95  & 0.554   & 0.93  & 0.552   & 0.93  & 0.544   & 0.95  & 0.339   \\
		& X      & 0.00     & 0.536   & 0.93  & 0.573   & 0.93  & 0.495   & 0.93  & 0.500     & 0.96  & 0.496   & 0.87  & 0.349   \\
		& SE     & 0.04  & 0.532   & 0.95  & 0.587   & 0.97  & 0.500     & 0.94  & 0.492   & 0.94  & 0.499   & 0.93  & 0.349   \\\hline	
		\multirow{5}{*}{2048} & DINA   & 0.97  & 0.478   & 0.91  & 0.491   & 0.91  & 0.450    & 0.90   & 0.450    & 0.96  & 0.449   & 0.95  & 0.271   \\
		& E      & 0.37  & 0.485   & 0.30   & 0.519   & 0.91  & 0.462   & 0.93  & 0.445   & 0.96  & 0.446   & 0.96  & 0.314   \\
		& PA-X   & 0.87  & 0.589   & 0.73  & 0.636   & 0.96  & 0.466   & 0.95  & 0.458   & 0.91  & 0.449   & 0.93  & 0.282   \\
		& X      & 0.00     & 0.447   & 0.93  & 0.486   & 0.95  & 0.422   & 0.95  & 0.424   & 0.93  & 0.422   & 0.91  & 0.297   \\
		& SE     & 0.01  & 0.435   & 0.92  & 0.483   & 0.98  & 0.425   & 0.99  & 0.418   & 0.95  & 0.413   & 0.96  & 0.292   \\\hline	
		\multirow{5}{*}{2896} & DINA   & 0.89  & 0.393   & 0.91  & 0.413   & 0.98  & 0.378   & 0.92  & 0.370    & 0.93  & 0.370    & 0.97  & 0.225   \\
		& E      & 0.14  & 0.402   & 0.09  & 0.429   & 0.93  & 0.382   & 0.95  & 0.374   & 0.98  & 0.380    & 0.92  & 0.263   \\
		& PA-X   & 0.72  & 0.493   & 0.73  & 0.513   & 0.95  & 0.386   & 0.97  & 0.379   & 0.96  & 0.378   & 0.95  & 0.237   \\
		& X      & 0.00     & 0.378   & 0.97  & 0.410    & 0.93  & 0.347   & 0.95  & 0.348   & 0.96  & 0.349   & 0.91  & 0.245   \\
		& SE     & 0.00     & 0.383   & 0.95  & 0.404   & 0.94  & 0.358   & 0.93  & 0.349   & 0.93  & 0.347   & 0.97  & 0.243   \\\hline	
		\multirow{5}{*}{4096} & DINA   & 0.97  & 0.339   & 0.86  & 0.343   & 0.88  & 0.317   & 0.96  & 0.314   & 0.95  & 0.315   & 0.96  & 0.186   \\
		& E      & 0.02  & 0.332   & 0.06  & 0.364   & 0.91  & 0.317   & 0.97  & 0.318   & 0.94  & 0.311   & 0.90   & 0.213   \\
		& PA-X   & 0.65  & 0.411   & 0.64  & 0.430    & 0.98  & 0.323   & 0.93  & 0.312   & 0.96  & 0.315   & 0.91  & 0.200     \\
		& X      & 0.00     & 0.312   & 0.91  & 0.339   & 0.99  & 0.292   & 0.95  & 0.293   & 0.97  & 0.290    & 0.86  & 0.205   \\
		& SE     & 0.00     & 0.315   & 0.90   & 0.344   & 0.98  & 0.293   & 0.96  & 0.292   & 0.96  & 0.296   & 0.86  & 0.206   \\\hline	
		\multirow{5}{*}{5792} & DINA   & 0.95  & 0.276   & 0.94  & 0.291   & 0.92  & 0.269   & 0.94  & 0.260    & 0.95  & 0.261   & 0.93  & 0.158   \\
		& E      & 0.00     & 0.283   & 0.01  & 0.306   & 0.88  & 0.266   & 0.94  & 0.267   & 0.94  & 0.262   & 0.86  & 0.182   \\
		& PA-X   & 0.52  & 0.348   & 0.51  & 0.367   & 0.93  & 0.267   & 0.93  & 0.264   & 0.95  & 0.260    & 0.87  & 0.166   \\
		& X      & 0.00     & 0.266   & 0.91  & 0.287   & 0.91  & 0.245   & 0.94  & 0.247   & 0.95  & 0.247   & 0.84  & 0.172   \\
		& SE     & 0.00     & 0.263   & 0.86  & 0.288   & 0.97  & 0.254   & 0.96  & 0.245   & 0.92  & 0.246   & 0.84  & 0.173  \\
	\end{tabular}
}
\end{table}

\subsection{Cox model}\label{subsec:Cox}
We consider the five methods in Section~\ref{subsec:exponentialFamily} under the framework of Algorithm~\ref{algo:coxPartial} (partial-likelihood). We follow model~\eqref{eq:simulationModel} and use the baseline hazard function $\lambda(y) = y$, uniform censoring ($75\%$ units censored). 
In step one, we obtain $\hat{e}(x)$ by logistic regression and $\hat{\eta}_0(x)$, $\hat{\eta}_1(x)$ by fitting Cox proportional hazards regression models. In step two, all methods estimate $\beta$ by maximizing the partial likelihoods.
From the panel (d) of Figure~\ref{fig:error}, we observe that the proposed method with partial-likelihood behaves the most favorably. 
Results of confidence intervals can be found in the appendix.

\section{Real data analysis}\label{sec:realData}
In this section, we apply the estimators to the SPRINT dataset analyzed by \citet{powers2018some}. The data is collected from a randomized trial aiming to study whether a new treatment program targeting reducing systolic blood pressure (SBP) will reduce cardiovascular disease (CVD) risk. The response is whether any of the major CVD events\footnote{Major CVD events: myocardial infarction (MI), non-MI acute coronary syndrome (non-MI ACS), stroke, heart failure (HF), or death attributable to cardiovascular disease.} occur to a participant.
We start with $9$ continuous lab measurements,
follow the data preprocessing procedures of \citet{powers2018some}, and end up with $7517$ samples. The descriptive statistics of the data after preprocessing are given in Table \ref{tab:stats}.  
Since SBP and DBP, EGFR and SCREAT are seriously correlated, respectively, we further remove DBP and SCREAT and are left with $7$ covariates.
We scale the covariates to have zero mean and unit variance.
Our goal is to estimate the heterogeneous treatment effect in the scale of odds ratio. 

We consider the five methods in Section~\ref{sec:simulation}. 
Since the data is collected from a randomized trial, the propensity score is $e(x) = 0.5$. 
We estimate two natural parameter functions using random forests.
We obtain confidence intervals by bootstrap ($100$ bootstrap samples).

\begin{table}
		\caption{\label{tab:stats}{Descriptive statistics of the SPRINT data after preprocessing. Standard deviations are in parentheses.}}	
	\centering
	\hspace{-0cm}
	\fbox{%
	\begin{tabular}{c|c|c|c|c}
		Type & Abbrevation &         Feature    & Treatment Group & Control Group \\ \hline
		Group size & --- & --- &$3708$& $3809$\\
		&&&& \\\hline
		&SBP& Seated systolic blood       &         $139 ~(15.6)$         &  $140 ~(15.5)$              \\ 
		&& pressure (mm Hg)       &                 &               \\ \cline{2-5} 
		&DMP& Seated diastolic blood  &         $78.3 ~(11.7)$         &    $78.2 ~(12.1)$            \\ 
		&& pressure (mm Hg)      &                 &               \\ \cline{2-5} 
		&EGFR & eGFR MDRD  &     $72.0 ~(20.5)$             &           $72.1 ~(20.2)$     \\
		&&(mL/min/1.73m$^2$) &                 &               \\ \cline{2-5} 
		&SCREAT& Serum creatinine             &       $1.07 ~(0.34)$             &       $1.07 ~(0.32)$                     \\ 
		Lab meas-  && (mg/dL)                  &               &   \\ \cline{2-5} 
		urements  & CHR & Cholesterol (mg/dL)                          &        $191 ~(41.7)$                      &       $190 ~(40.4)$                     \\ 
		&&&& \\ \cline{2-5} 
		&GLUR& Glucose (mg/dL)                             &         $99.0 ~(13.7)$                     &         $98.9 ~(13.3)$                    \\ 
		&&&&\\ \cline{2-5} 
		&HDL& HDL-cholesterol direct            &         $52.7 ~(14.3)$                     &       $52.7 ~(14.5)$                     \\
		&& (mg/dL)               &                 &               \\ \cline{2-5} 
		&TRR& Triglycerides (mg/dL)                        &      $126 ~(22.2)$                        &         $126 ~(19.6)$                    \\ 
		&&&&\\ \cline{2-5} 
		&UMALCR& Urine Albumin/Creatinine ratio &            $39.8 ~(15.8)$                  &        $36.6 ~(12.5)$                    \\
		&  & (mg/g)     &                 &               \\ 
	\end{tabular}
}
\end{table}

\subsection{Results from randomized experiment}\label{subsec:HTE}
Table \ref{tab:coefficientsRF} displays the estimated coefficients of DINA using random forests as the nuisance-function learner.  
The proposed method, propensity score adjusted X-learner, and X-learner produce negative intercepts significant at the $95\%$ level. The estimated intercept of DINA is negatively significant, which implies that for a patient with average covariate values, the treatment decreases the odds of experiencing any CVD events. The result largely agrees with the original clinical results where researchers observed the treatment group were performing better\footnote{The original study focuses on average treatment effects on the scale of probability per year. The treatment group sees $1.65\%$ incidence per year and the control group sees $2.19\%$ incidence per year.}.
As for heterogeneity in the treatment effect, our proposal finds EGFR as a significant effect modifier. In particular, the estimated coefficient of EGFR is negative, indicating that the treatment is more beneficial to patients with high EGFR.

\begin{table}
	\caption{\label{tab:coefficientsRF}{Estimated coefficients and $95\%$ confidence intervals (CI) from the SPRINT dataset. We compare the five methods in Section~\ref{sec:simulation}. We input the true propensity score $e(x ) = 0.5$ and use random forests as the nuisance-function learner. We use bootstrap ($100$ bootstrap samples) to construct the confidence intervals.}}	
	\centering
	\scriptsize
	\fbox{%
		\begin{tabular}{cc|ccccc}
			Abbreviation                     &  & DINA                     & E                  & PA-X                   & X                    & SE                 \\\hline
			\multirow{2}{*}{SBP}       & $\hat{\beta}_{\text{SBP}}$ & 0.109                  & 0.101                 & 0.313               & 0.255                  & 0.135                \\
			& 95\% CI    & {[}-0.169, 0.387{]}    & {[}-0.166, 0.368{]}   & {[}-0.030, 0.656{]}  & {[}-0.072, 0.582{]}   & {[}-0.128, 0.398{]}  \\
			\multirow{2}{*}{EGFR}      & $\hat{\beta}_{\text{EGFR}}$  & -0.331                 & -0.353                & -0.107              & -0.086                & -0.087               \\
			& 95\% CI    & {[}-0.662, 0.000{]} & {[}-0.725, 0.019{]}  & {[}-0.548, 0.334{]} & {[}-0.412, 0.239{]}    & {[}-0.432, 0.258{]}  \\
			\multirow{2}{*}{CHR}       & $\hat{\beta}_{\text{CHR}}$  & 0.111                  & 0.129                 & -0.176              & -0.141                 & -0.121               \\
			& 95\% CI    & {[}-0.273, 0.495{]}    & {[}-0.238, 0.496{]}   & {[}-0.552, 0.200{]}   & {[}-0.480, 0.198{]}     & {[}-0.491, 0.249{]}  \\
			\multirow{2}{*}{GLUR}      & $\hat{\beta}_{\text{GLUR}}$  & -0.211                 & -0.205                & -0.116              & -0.0535                & -0.156               \\
			& 95\% CI    & {[}-0.501, 0.079{]}   & {[}-0.468, 0.058{]}  & {[}-0.455, 0.223{]} & {[}-0.306, 0.199{]}    & {[}-0.454, 0.142{]}  \\
			\multirow{2}{*}{HDL}       & $\hat{\beta}_{\text{HDL}}$  & -0.000              & 0.038                & 0.011              & -0.030                & -0.136               \\
			& 95\% CI    & {[}-0.402, 0.401{]}    & {[}-0.366, 0.442{]}   & {[}-0.448, 0.470{]}  & {[}-0.426, 0.366{]}    & {[}-0.565, 0.293{]}  \\
			\multirow{2}{*}{TRR}       & $\hat{\beta}_{\text{TRR}}$  & 0.103                  & 0.098                & 0.218               & 0.082                 & 0.026               \\
			& 95\% CI    & {[}-0.267, 0.473{]}    & {[}-0.262, 0.459{]}   & {[}-0.219, 0.655{]} & {[}-0.230, 0.393{]}     & {[}-0.305, 0.357{]}  \\
			\multirow{2}{*}{UMALCR}    & $\hat{\beta}_{\text{UMALCR}}$  & 0.001                & -0.045               & -0.039             & 0.094                 & -0.057              \\
			& 95\% CI    & {[}-0.228, 0.230{]}     & {[}-0.284, 0.194{]}   & {[}-0.257, 0.178{]} & {[}-0.124, 0.311{]}    & {[}-0.253, 0.139{]}  \\
			\multirow{2}{*}{Intercept} & $\hat{\beta}_{0}$  & -0.433                 & -0.412                & -0.350               & -0.336                 & -0.246               \\
			& 95\% CI    & {[}-0.735, -0.131{]}   & {[}-0.739, -0.085{]} & {[}-0.730, 0.030{]} & {[}-0.669, -0.003{]} & {[}-0.583, 0.091{]} \\
		\end{tabular}
	}
\end{table}

\subsection{Results from artificial observational studies}\label{subsec:robustness}

We compare the sensitivity to treatment assignment mechanisms.
We generate subsamples from the original dataset with non-constant propensity score and compare the estimated treatment effects from the original randomized trial and those from the artificial observational data.

Let $e(x)$ be an artificial propensity score following a logistic regression model. 
We subsample from the original SPRINT dataset with probabilities
\begin{align*}
	\PP(\text{select the } i \text{-th sample} \mid X_i = x) = 
	\begin{cases}
		\min\left\{\frac{e(x)}{1-e(x)}, 1 \right\}, & W_i = 1,\\
		\min\left\{\frac{1-e(x)}{e(x)}, 1 \right\}, & W_i = 0.
	\end{cases}
\end{align*}
In the artificial subsamples, the probability of a unit being treated is $e(x)$.
We run five estimation methods on the artificial subsamples, obtain estimated treatment effects, and repeat from subsampling $100$ times. For nuisance-function estimation, we use logistic regression to learn the propensity score, and random forests or logistic regression to learn the treatment/control natural parameter functions.

For any method, we regard the estimates obtained from the original dataset as the oracle, 
denoted by ${\tau}_{\text{method}}$, and compare with the estimates obtained from artificial subsamples, denoted by $\hat{\tau}_{\text{method}}$. We quantify the similarity by one minus the normalized mean squared difference,
\begin{align}\label{eq:S}
	R^2(\hat{\tau}_{\text{method}}) :=1 - \frac{\sum_{i=1}^{n} (\hat{\tau}_{\text{method}}(X_i) - \tau_{\text{method}}(X_i))^2}{\sum_{i=1}^{n} \tau_{\text{method}}^2(X_i)}.
\end{align}
The larger $R^2(\hat{\tau}_{\text{method}})$ is, the closer $\hat{\tau}_{\text{method}}$ is to its oracle ${\tau}_{\text{method}}$. In other words, a method with a larger $R^2(\hat{\tau}_{\text{method}})$ is more robust to the experimental design.

From Table \ref{tab:R2} we observe the proposed method (DINA) produces the largest or the second-largest $R^2$. The advantage is more evident when we use generalized linear regression as the nuisance-function learner.

\begin{table}
	\caption{\label{tab:R2}Comparison of the sensitivity ($R^2$ in \eqref{eq:S}) to treatment assignment mechanisms. We consider the five methods in Section~\ref{sec:simulation}. We generate $100$ subsamples from the original SPRINT dataset with a non-constant propensity score. 
			We consider two nuisance-function learners: logistic regression and random forests.
			For each method, we obtain estimators from the original SPRINT dataset and the artificial subsamples, and compute $R^2$.}
	\scriptsize
	\centering
	\fbox{
	\begin{tabular}{c|ccccc}
		Nuisance-function learner& DINA    & E       & PA-X    & X       & SE      \\\hline
		\multirow{2}{*}{Logistic regression} & 0.501   & 0.387   & 0.107   & 0.0982  & -0.468  \\
		& (0.186) & (0.252) & (0.522) & (0.404) & (0.772) \\
		\multirow{2}{*}{Random forests}  & 0.436   & 0.242   & 0.549   & 0.426   & -1.44   \\
		& (0.400)   & (0.458) & (0.305) & (0.328) & (2.48)  \\
	\end{tabular}
}
\end{table}

\section{Multi-valued treatment}\label{sec:multiValued} 
In this section, we extend our method to address multi-valued treatments, a.k.a. treatments with multiple levels.  Let $T \in \{0,1,\ldots, K\}$ be the treatment variable of $K+1$ levels, and let level $0$ denote the control group. Let $W^{(t)}$ be the indicator of receiving treatment $t$, i.e., $W^{(t) } = \1_{\{T = t\}}$. We adopt the generalized propensity score of \citet{imbens2000role} 
\begin{align}\label{eq:propMultiple}
	\begin{split}
		e_t(x) = \PP\left(T = t \mid X=x\right) = \EE[W^{(t)} \mid X=x],  \quad 0 \le t \le K.
	\end{split}
\end{align} 
We denote the natural parameter under treatment $t$ by $\eta_t(x)$, and we aim to estimate the DINAs
\begin{align}\label{eq:DINA:exponentialFamilyMultiple}
	\tau_t(x) := \eta_t(x) - \eta_0(x), \quad 1 \le t \le K.
\end{align} 
We consider flexible control baseline $\eta_0(x)$ and assume parametric forms of DINA as in \eqref{eq:DINA:parametric},
\begin{align}\label{eq:DINA:parametricMultiple}
	\tau_t(x) = x^\top \beta_t, \quad 1 \le t \le K.
\end{align}

The analysis with multi-valued treatments is reminiscent of that with single treatments in Section~\ref{sec:method} and \ref{sec:cox}. We rewrite the natural parameters as
\begin{align}\label{eq:model:exponentialFamilyMultiple}
	\eta_t(x)
	= \nu(x) + 
	\begin{bmatrix}
		- a_1(x), \ldots, 1- a_t(x), \ldots,  - a_K(x)
	\end{bmatrix}
	\begin{bmatrix}
		x^\top \beta_1 \\
		\ldots \\
		x^\top \beta_t \\
		\ldots \\
		x^\top \beta_K \\
	\end{bmatrix}, \quad 1 \le t \le K,
\end{align}
or more compactly
\begin{align*}
	\boldsymbol{\eta}(x)
	= \nu(x) \boldsymbol{1}_K + \left(\boldsymbol{I}_K -  \boldsymbol{1}_K\boldsymbol{a}(x)^\top\right) \boldsymbol{\tau}(x),
\end{align*}
where $\boldsymbol{\eta}(x) = [\eta_1(x),\ldots,\eta_K(x)]^\top$, $\boldsymbol{a}(x) = [a_1(x),\ldots,a_K(x)]^\top$, $\boldsymbol{\tau}(x) = [\tau_1(x),\ldots, \tau_K(x)]^\top$, and $\boldsymbol{1}_K = [1,\ldots, 1]^\top$.
We look for $\boldsymbol{a}(x)$ and $\nu(x)$ such that the DINA estimator is not heavily disturbed by inaccurate nuisance functions. Analogous to \eqref{eq:condition4}, we arrive at 
\begin{align}\label{eq:conditionMultiple3}
	\begin{split}
		a_t(x) = \frac{e_t(x) V_t(x)}{\sum_{s=0}^K e_s(x) V_s(x)}, &\quad 1 \le t \le K, \quad a_0(x) = 1 - \sum_{t=1}^K a_t(x),\\
		\nu(x) 
		&= \sum_{t=0}^K a_t(x) \eta_t (x).
	\end{split}
\end{align}

The interpretation of single-valued treatment in Subsection \ref{subsec:interpretation} holds.
In \eqref{eq:conditionMultiple3}, the weight $a_t(x)$ is proportional
to the treatment probability $e_t(x)$ multiplied by the response
conditional variance.
The multiplier $(w^{(t)}-a_t(x))$ associated with the HTE in
\eqref{eq:model:exponentialFamily2} is small under treatment $t$ if
the unit is inclined to receive such treatment or its treated response
has high conditional variance at $x$. 
Consequently, the proposed method equalizes the influences of treatments from all levels and attenuates the effects of excessively influential samples.

The algorithms in Section~\ref{sec:method} and \ref{sec:cox} can be
directly generalized to handle multi-valued treatments. As for
nuisance functions, we estimate the generalized propensity score
$e_t(x)$ by any multi-class classifier that provides probability
estimates. We estimate $V_t(x)$
by obtaining $\hat{\eta}_t(x)$ first and plugging in to the
appropriate exponential family formula for the variance.

Theoretically, Proposition~\ref{prop:main}, \ref{prop:mainCoxFull},
\ref{claim:cox:null} still hold, meaning that the proposed estimator
will improve over separate estimation based on the nuisance-function estimators
$\hat{\eta}_t(x)$. 
Compared to single-level treatments, generalized propensity scores of
multi-valued treatments may take more extreme values, especially
values close to zero. Extreme propensities will cause IPW-type methods \citep{dudik2011doubly,weisberg2015post}
to be highly variant, while the variances of R-learner as well
as our extension will not explode. We remark that when $\hat{\tau}(x)$
is estimated with a large number of parameters, problematic extrapolations may happen in the regions absent from certain treatment groups.

\section{Discussion}\label{sec:discussion}
In this paper, we propose to quantify the treatment effect by DINA for
responses from the exponential family and Cox model, in contrast to
the conditional mean difference. For non-continuous responses, e.g.,
binary and survival data, DINA is of more practical interest,
e.g., relative risk and hazard ratio, and is convenient for modeling
the influence of covariates to the treatment effects. Similar to
R-learner, we introduce a
DINA estimator that is insensitive to confounding and
non-collapsibility issues. The method is flexible in the sense that it
allows practitioners to use powerful off-the-shelf machine learning
tools for nuisance-function estimation.

There are several potential extensions of the proposed method.
\begin{enumerate}
	\item {Non-parametric DINA estimation}. The current method relies on the semi-parametric assumption \eqref{eq:DINA:parametric}. 
	According to \citet{nie2017quasi}, the non-parametric extension of R-learner on estimating the difference in means is discussed. We can similarly extend the proposed method to allow non-parametric modeling. The only difference is that we maximize the log-likelihood over non-parametric learners instead of parametric families in the second step. For example, in Algorithm~\ref{algo:exponentialFamily}, we can instead solve the optimization problem
	\begin{align}\label{eq:loglikeNonParametric}
		\max_{\tau'(x) \in \calF}\left\{\frac{1}{n}\sum_{i=1}^n Y_i(W-\hat{a}(X_i)) \tau'(X) - \psi\left(\hat{\nu}(X_i) + (W-\hat{a}(X_i)) \tau'(X)\right)\right\},
	\end{align}
	where $\calF$ denotes some non-parametric function class. 
	
	One way to solve \eqref{eq:loglikeNonParametric} is casting it as a varying coefficient problem \citep{hastie1993varying} and applying local regressions. 
	An alternative approach models $\tau(x)$ as a neural network, uses \eqref{eq:loglikeNonParametric} as the loss function,  and employs powerful optimization tools developed for neural networks.
	\item {Cox model baseline hazard robustness}. As
	mentioned in Section~\ref{subsec:full}, the current method
	with full likelihood maximization is sensitive to a misspecified baseline hazard function. We argue in Section~\ref{sec:cox} that the misspecification of baseline hazard $\lambda(y)$ and that of natural parameters $\eta_1(x)$, $\eta_0(x)$ are different in nature. We believe developing approaches targeting  baseline-hazard robustness is an exciting direction. 
\end{enumerate}

\section{Acknowledgement}
This research was partially supported by grants DMS 2013736 and IIS 1837931 from the National Science Foundation, and grant 5R01 EB 001988-21 from the National Institutes of Health.

\clearpage
\appendix{
%

\section{Proofs}\label{sec:appendix:proof}
\subsection{Exponential family}\label{subsec:appendix:exponentialFamily}

\begin{claim}\label{prop:consistency}
	Suppose the potential outcomes follow the exponential family model~\eqref{eq:model:exponentialFamily} with natural parameter functions~\eqref{eq:model:additive5}. Assume the treatment assignment is randomized. Consider the separate estimation method that estimates the nuisance functions $\eta_0(x, z)$, $\eta_1(x, z)$ by fitting generalized linear regressions to observed covariates $x$. Then $\hat{\tau}_{\text{SE}}(x)$ is consistent if and only if the canonical link function is linear.
\end{claim}

\begin{proof}[Proof of Claim \ref{prop:consistency}]
	Under randomized experiment, the separate estimation method obeys the population moment equations
	\begin{align}\label{eq:proof:beta}
		\begin{split}
			\EE[X\mu(\eta_0(X, Z))]
			&= \EE[X\mu(X^\top \alpha^*)],\\
			\EE[X\mu(\eta_0(X, Z) + X^\top \beta)]
			&= \EE[X\mu(X^\top \alpha^* + X^\top \beta^*)].
		\end{split}
	\end{align}
	
	On one hand, by Taylor's expansion of \eqref{eq:proof:beta},
	\begin{align}\label{eq:proof:beta4}
		\begin{split}
			0&=\EE[X\mu(\eta_0(X, Z) + X^\top \beta)] -\EE[X\mu(X^\top \alpha^* + X^\top \beta)]\\
			&= \EE[X\mu'(\eta_0(X, Z) + X^\top \beta)(\eta_0(X, Z) - X^\top \alpha^*)] + r_1,
		\end{split}
	\end{align}
	where the remainder term $r_1\le O(\max_{\eta} |\mu''(\eta)|\EE[X(\eta_0(X, Z) - X^\top \alpha^*)^2])$. 
	Notice that by \eqref{eq:proof:beta}, $\alpha^*$ does not depend on $\beta$. Since \eqref{eq:proof:beta4} is true for arbitrary $\beta$ and $\eta_0(x, z)$, then $\mu'(\eta)$ is constant. Therefore, the link function $\mu(\eta)$ is linear.
	
	On the other hand, if the canonical link is linear, then we take the difference of the moment equations \eqref{eq:proof:beta} in the control group and the treatment group and arrive at
	\begin{align}\label{eq:proof:beta3}
		\begin{split}
			\EE[XX^\top] \beta = \EE[XX^\top] \beta^* \implies \beta^* = \beta,
		\end{split}
	\end{align}
	as long as $\EE[XX^\top]$ is non-degenerate.
\end{proof}

\begin{proposition}\label{prop:derivationExponential}
	Let $Y$ be generated from the model \eqref{eq:model:exponentialFamily} with natural parameter functions $\eta_w(x) = \nu(x) + (w-a(x)) x^\top \beta$. 
	For arbitrary $\nu^*(x)$, $\beta'$, let $\ell(Y; \beta', a(x),  \nu^*(x))$ be the log-likelihood of  $Y$ evaluated at natural parameter functions $\nu^*(x) + (w-a(x)) x^\top \beta'$.  Then $a(x)$, $\nu(x)$ in \eqref{eq:condition4} ensure 
	\begin{align*}
		\beta = \argmax_{\beta'} \EE\left[\ell(Y; \beta', a(x),  \nu^*(x))\right]
	\end{align*}	
	is true under the log-likelihood quadratic approximation in Lemma \ref{lemm:likelihoodApproximationExponential}.
\end{proposition}

\begin{proof}[Proof of Proposition \ref{prop:derivationExponential}]\label{proof:robust:exponentialFamily:a(x)}
	By Lemma \ref{lemm:likelihoodApproximationExponential}, for arbitrary $\beta'$, the log-likelihood of $Y$ satisfies
	\begin{align}\label{eq:loglike1.75}
		\begin{split}
			&\ell(Y; \beta', a(x), \nu^*(x)) 
			\approx  \ell(Y;\beta, a(x), \nu^*(x)) \\
			&-\frac{1}{2} \psi''(\eta_W^*(X)) \left(r - (W - a(X)) X^\top (\beta' -\beta)\right)^2 + \frac{1}{2} \psi''(\eta_W^*(X)) r^2,
		\end{split}
	\end{align}
	where $\eta_w^*(x) = \nu^*(x) + (w - a(x))x^\top\beta$ and the residual $r = \left(Y - \psi'(\eta_W^*(X))\right)/\psi''(\eta_W^*(X))$.
	Maximizing the expectation of the quadratic approximation gives us
	\begin{align}\label{eq:solution}
		\beta' -\beta = \EE\left[\psi''\left(\eta_W^*(X)\right)(W-a(X))^2XX^\top\right]^{-1} \EE\left[ \psi''\left(\eta_W^*(X)\right)(W-a(X))X r\right].
	\end{align}
	Thus, our target $\beta' = \beta$ is equivalent to $\EE\left[ \psi''(\eta_W^*(X))(W-a(X))X r\right] = 0$. By the definition of $r$, the condition simplifies to
	\begin{align}\label{eq:condition}
		\EE\left[(W-a(X))X (Y - \psi'(\eta_W^*(X))\right] = 0.
	\end{align} 
	Note that in the exponential family, $\EE[Y \mid X, W] = \psi'(\eta_W(X))$ at true natural parameter functions $\eta_w(x)$, then by Taylor's expansion of $\psi'$, 
	\eqref{eq:condition} becomes
	\begin{align}\label{eq:condition2}
		\begin{split}
			0 
			\approx& \EE\left[e(X)(1-a(X))X \psi''(\eta_1(X)) (\eta_1(X) - \eta_1^*(X))\right.\\
			&\left.- (1-e(X))a(X)X \psi''(\eta_0(X)) (\eta_0(X) - \eta_0^*(X))\right]. 
		\end{split}
	\end{align} 
	Since the difference between the true parameters $\eta_w(x)$ and $\eta_w^*(x)$: $\eta_w^*(x) - \eta_w(x) = \nu^*(x) - \nu(x)$, is independent of the treatment assignment $W$, the two terms on the right hand side of \eqref{eq:condition2} can be merged together,
	\begin{align}\label{eq:condition3}
		\begin{split}
			0 
			=& \EE\left[X (\nu(X) - \nu^*(X))
			\left(e(X)(1-a(X)) \psi''(\eta_1(X)) - (1-e(X))a(X) \psi''(\eta_0(X))\right)\right].
		\end{split}
	\end{align} 
	We require condition \eqref{eq:condition3} to be true for arbitrary approximation error direction $(\nu^*(x) - \nu(x))/\|\nu^*(x) - \nu(x)\|$,
	then 
	\begin{align*}
		\begin{split}
			0 
			&= e(x)(1-a(x)) \psi''(\eta_1(x)) - (1-e(x))a(x) \psi''(\eta_0(x)),\\
			&\implies 
			\frac{a(x)}{1-a(x)} 
			= \frac{e(x)}{1-e(x)}\frac{\psi''(\eta_1(x))}{\psi''(\eta_0(x))}
			= \frac{e(x)}{1-e(x)}\frac{\frac{d\mu}{d\eta}(\eta_1(x))}{\frac{d\mu}{d\eta}(\eta_0(x))},
		\end{split}
	\end{align*} 	
	and subsequently $\nu(x)$,
	\begin{align*}
		\begin{split}
			\nu(x) 
			&= \frac{e(x)\frac{d\mu}{d\eta}(\eta_1(x)) }{e(x)\frac{d\mu}{d\eta}(\eta_1(x))  + (1-e(x))\frac{d\mu}{d\eta}(\eta_0(x))} \eta_1(x) + \frac{(1-e(x))\frac{d\mu}{d\eta}(\eta_0(x)) }{e(x)\frac{d\mu}{d\eta}(\eta_1(x))  + (1-e(x))\frac{d\mu}{d\eta}(\eta_0(x))} \eta_0(x).
		\end{split}
	\end{align*} 	
	Plug $V_1(x) = \frac{d\mu}{d\eta}(\eta_1(x)) $, $V_0(x) = \frac{d\mu}{d\eta}(\eta_0(x)) $ in and we have derived \eqref{eq:condition4} from the robustness requirement of $\nu(x)$. 
	
	In fact, we can show \eqref{eq:condition4} also protects the estimator from misspecified $a(x)$. Consider $a^*(x)$ as an approximation of $a(x)$, and let $\eta_w^*(x) = \nu(x) + (w - a^*(x)) x^\top \beta$. Then by Lemma \ref{lemm:likelihoodApproximationExponential},
	\begin{align*}
		\begin{split}
			&\ell(Y; \beta', a^*(x), \nu(x)) 
			\approx  \ell(Y;\beta, a^*(x), \nu(x)) \\
			&-\frac{1}{2} \psi''(\eta_W^*(X)) \left(r - (W - a^*(X)) X^\top (\beta' -\beta)\right)^2 + \frac{1}{2} \psi''(\eta_W^*(X)) r^2,
		\end{split}
	\end{align*}
	where the residual term is defined as $r := \left(Y - \psi'(\eta_W^*(X))\right)/\psi''(\eta_W^*(X))$. 
	The condition \eqref{eq:condition} becomes  
	\begin{align*}
		\begin{split}
			0 
			&= \EE\left[(W-a^*(X))X (Y - \psi'(\eta_W^*(X))\right] \\
			&= \underbrace{\EE\left[(W-a(X))X (Y - \psi'(\eta_W^*(X))\right]}_{:=\text{(I)}} + \underbrace{\EE\left[(a(X)-a^*(X))X (Y - \psi'(\eta_W^*(X))\right]}_{:=\text{(II)}}.
		\end{split}
	\end{align*} 
	Part (I) is similar to \eqref{eq:condition2} and is approximately zero under \eqref{eq:condition4}. For part (II), 
	\begin{align}\label{eq:conditiona2}
		\begin{split}
			\text{(II)} 
			&= \EE\left[(a(X)-a^*(X))X (Y - \psi'(\eta_W(X))\right] \\
			&\quad + \EE\left[(a(X)-a^*(X))X (\psi'(\eta_W(X)) - \psi'(\eta_W^*(X))\right]\\
			&= 0 + \EE\left[(a(X)-a^*(X))X \psi''(\eta_W(X))(a^*(X) - a(X))X^\top \beta)\right] + o\left(\EE[(a^*(X) - a(X))^2]\right)
		\end{split}
	\end{align} 
	by Taylor's expansion. Therefore, the error in $\beta$ is at most the quadratic of that in $a(x)$, and the method is insensitive to the errors in $a(x)$.
\end{proof}

\begin{lemma}\label{lemm:likelihoodApproximationExponential}
	Let $Y$ be generated from the model \eqref{eq:model:exponentialFamily} with the natural parameter $\eta$, then for arbitrary $\eta'$, the likelihood of $Y$ satisfies
	\begin{align*}
		\begin{split}
			\ell(Y; \eta') 
			= \ell(Y; \eta) - \frac{1}{2} \psi''(\eta) \left(r + \eta - \eta'\right)^2 + \frac{1}{2} \psi''(\eta) r^2 + O\left(\|\eta - \eta'\|_2^3\right),
		\end{split}
	\end{align*}	
	where $r := \left(Y - \psi'(\eta)\right)/\psi''(\eta)$.
\end{lemma}

\begin{proof}[Proof of Lemma \ref{lemm:likelihoodApproximationExponential}]\label{proof:exponential:likelihoodApproximation}
	Let $Y$ be generated from \eqref{eq:model:exponentialFamily} with the natural parameter $\eta$, then the log-likelihood of $Y$ at the natural parameter $\eta'$ takes the form
	\begin{align}\label{eq:loglike} 
		\begin{split}
			\ell(Y; \eta') 
			=  Y\eta' - \psi(\eta') + \log(\kappa(Y)) 
			= \ell(Y; \eta)  + Y(\eta' -\eta) - (\psi(\eta') - \psi(\eta)).
		\end{split}
	\end{align}
	We apply Taylor's expansion to $\psi$ at $\eta$ 
	\begin{align}\label{eq:Taylor}
		\psi(\eta') - \psi(\eta)
		= \psi'(\eta)(\eta' - \eta) + \frac{1}{2} \psi''(\eta)(\eta' - \eta)^2 + O\left(\|\eta' - \eta\|_2^3\right).
	\end{align}
	Plugging \eqref{eq:Taylor} into \eqref{eq:loglike} gives us
	\begin{align}\label{eq:loglike1.5}
		\begin{split}
			\ell(Y; \eta') 
			&= \ell(Y; \eta) + Y(\eta' -\eta) - \psi'(\eta)(\eta' - \eta) - \frac{1}{2} \psi''(\eta)(\eta' - \eta)^2 + O\left(\|\eta' - \eta\|_2^3\right)\\
			&= \ell(Y; \eta) - \frac{1}{2} \psi''(\eta) \left(r + \eta - \eta'\right)^2 + \frac{1}{2} \psi''(\eta) r^2 + O\left(\|\eta' - \eta\|_2^3\right),
		\end{split}
	\end{align}
	where the residual term is defined as $r := \left(Y - \psi'(\eta)\right)/\psi''(\eta)$.
\end{proof}


\begin{proof}[Proof of Claim \ref{claim:conditionalMean}]
	Under the model \eqref{eq:model:exponentialFamily},
	\begin{align*}
		\EE[Y \mid X = x]
		&= \EE[Y \mid W = 1, X = x] ~\PP(W = 1 \mid X=x) \\
		&\quad~+ \EE[Y\mid W = 0, X = x] ~\PP(W = 0 \mid X=x)\\
		&=  e(x) \mu(\eta_1(x)) + (1- e(x))\mu(\eta_0(x)). 
	\end{align*}
	Therefore, 
	\begin{align}\label{eq:proof:conditionalMean1}
		\begin{split}
			&\quad~\EE[Y \mid X = x] - \mu(\nu(x))\\
			&=  e(x) \mu(\eta_1(x))+ (1-e(x))  \mu(\eta_0(x)) - \mu(\nu(x))\\
			&= e(x) \left(\mu(\eta_1(x)) - \mu(\nu(x))\right) + (1-e(x)) \left(\mu(\eta_0(x)) - \mu(\nu(x))\right) \\
			&= e(x) \mu'(\eta_1(x))(\eta_1(x) - \nu(x)) - \frac{1}{2} e(x)\mu''(\eta_{\varepsilon_1}(x))(\eta_1(x) - \nu(x))^2\\
			&\quad~+ (1-e(x)) \mu'(\eta_0(x))(\eta_0(x) - \nu(x)) - \frac{1}{2} (1-e(x))\mu''(\eta_{\varepsilon_0}(x))(\eta_0(x) - \nu(x))^2,
		\end{split}
	\end{align}
	where we use Taylor's expansion in the last equality, $\eta_{\varepsilon_0}(x)$ is between $\eta_0(x)$, $\nu(x)$, and $\eta_{\varepsilon_1}(x)$ is between $\eta_1(x)$, $\nu(x)$.
	Recall that $\nu(x) = a(x) \eta_1(x) + (1-a(x))\eta_0(x)$ in \eqref{eq:condition4}, then
	\begin{align}\label{eq:proof:conditionalMeanDiff}
		\eta_0(x) - \nu(x) = -a(x) (\eta_1(x) - \eta_0(x)), \quad
		\eta_1(x) - \nu(x) = (1-a(x)) (\eta_1(x) - \eta_0(x)).
	\end{align}
	Plug \eqref{eq:proof:conditionalMeanDiff} into \eqref{eq:proof:conditionalMean1}, 
	\begin{align}\label{eq:proof:conditionalMean2}
		\begin{split}
			&\quad~\EE[Y \mid X = x] - \mu(\nu(x))\\
			&= e(x) \mu'(\eta_1(x)) (1-a(x)) (\eta_1(x) - \eta_0(x)) -
			(1-e(x)) \mu'(\eta_0(x))a(x) (\eta_1(x) - \eta_0(x)) + r(x) \\
			&= \left(e(x) \mu'(\eta_1(x)) (1-a(x)) - (1-e(x)) \mu'(\eta_0(x))a(x)\right) (\eta_1(x) - \eta_0(x)) + r(x),
		\end{split}
	\end{align}
	where the residual is defined as
	\begin{align}\label{eq:proof:conditionalMeanResidual}
		\begin{split}
			r(x)&:= -\frac{1}{2}e(x)\mu''(\eta_{\varepsilon_1}(x))(\eta_1(x) - \nu(x))^2  -\frac{1}{2} (1-e(x))\mu''(\eta_{\varepsilon_0}(x))(\eta_0(x) - \nu(x))^2.
		\end{split}
	\end{align}
	Next, by \eqref{eq:condition4},
	\begin{align}\label{eq:proof:conditionalMeanRatio}
		\frac{e(x) \mu'(\eta_1(x)) (1-a(x)) }{(1-e(x)) \mu'(\eta_0(x)) a(x)}
		= \frac{e(x) V_1(x) }{(1-e(x)) V_0(x)} \frac{1-a(x)}{a(x)}
		= 1.
	\end{align}
	Therefore, plug \eqref{eq:proof:conditionalMeanRatio} into \eqref{eq:proof:conditionalMean2} and we have $\EE[Y \mid X = x] - \mu(\nu(x)) = r(x) $. Finally, by \eqref{eq:proof:conditionalMeanDiff}, \eqref{eq:proof:conditionalMeanResidual}, we have $r(x) = O\left((\eta_1(x) - \eta_0(x))^2\right)$.
\end{proof}

\begin{claim}\label{prop:Neyman}
	The score derived from \eqref{eq:MLE} satisfies the Neyman's orthogonality.	
\end{claim}

\begin{proof}[Proof of Claim \ref{prop:Neyman}]
	Let $a_{\alpha}(x) = a(x) + \alpha \xi(x)$, $\nu_{\rho}(x) = \nu(x) + \rho \zeta(x)$ with $\EE[\xi^2(X)] = \EE[\zeta^2(X)] = 1$, then the score function implied by \eqref{eq:MLE} with nuisance parameters $a_{\alpha}(x)$, $\nu_{\rho}(x)$ is
	\begin{align}\label{eq:proof:score}
		\begin{split}
			s(\alpha, \rho, \beta)
			= \EE\left[\left(Y - \psi'\left(\nu_{\rho}(X) + (W-a_{\alpha}(X))X^\top \beta\right)\right)(W-a_{\alpha}(X))X\right].
		\end{split}
	\end{align}
	We take derivative of \eqref{eq:proof:score} with regard to $\alpha$ and evaluate at $\alpha = \rho = 0$,
	\begin{align}\label{eq:proof:scoreAlpha}
		\begin{split}
			\nabla_{\alpha} s(\alpha, \rho, \beta) \mid_{\alpha = \rho = 0}
			=& \underbrace{\EE\left[\psi''\left(\nu(X) + (W-a(X))X^\top \beta\right)\xi(X)X^\top\beta(W-a(X))X \right]}_{:=\text{(I)}}\\
			&- \underbrace{\EE \left[ \left(Y - \psi'\left(\nu(X) + (W-a(X))X^\top \beta\right)\right)\xi(X)X\right]}_{:=\text{(II)}}.
		\end{split}
	\end{align}
	Part (I) is zero by \eqref{eq:condition4}, and part (II) is zero because $\EE[Y \mid X,W] = \psi'(\nu(X) + (W-a(x))X^\top \beta)$. Therefore, \eqref{eq:proof:scoreAlpha} is zero.
	Next, we take derivative of \eqref{eq:proof:score} with regard to $\rho$ and evaluate at $\alpha = \rho = 0$,
	\begin{align}\label{eq:proof:scoreRho}
		\begin{split}
			\nabla_{\rho} s(\alpha, \rho, \beta) \mid_{\alpha = \rho = 0}
			=& -\EE\left[\psi''\left(\nu(X) + (W-a(X))X^\top \beta\right)\zeta(X)(W-a(X))X \right].
		\end{split}
	\end{align}
	By \eqref{eq:condition4}, \eqref{eq:proof:scoreRho} is zero. Therefore, the score derived from \eqref{eq:MLE} satisfies the Neyman's orthogonality.	
\end{proof}


\begin{proof}[Proof of Proposition~\ref{prop:main}]
	
	
	Rewrite ${a}_n(x) - a(x) = \alpha_n \xi_n(x)$, ${\nu}_n(x) - \nu(x) = \rho_n \zeta_n(x)$, where $\EE[\xi_n^2(X)]$ and $\EE[\zeta_n^2(X)]$ are one. By the assumption that ${a}_n(x)$, ${\nu}_n(x)$ converge at rate $O(c_n)$, we have $\alpha_n$, $\rho_n  = O(c_n)$. Let $s(a(X_i), \nu(X_i),\beta)$ denote the score function of the $i$-th sample, and define the empirical score $s_n(a(x), \nu(x),\beta) = \frac{1}{n}\sum_{i=1}^n s(a(X_i), \nu(X_i),\beta)$. 
	For simplicity, we write $s_n(a_n(x), \nu_n(x), \beta_n) $ as $s_n(\alpha_n, \rho_n, \beta_n)$, then the derivatives of the score take the form
	\begin{align*}
		\nabla_{\alpha} s_n(\alpha_n, \rho_n, \beta_n) &= 
		\frac{1}{n}\sum_{i=1}^n\nabla_{\alpha}s(a(X_i) + \alpha \xi_n(X_i), \nu(X_i) + \rho_n \zeta_n(X_i), \beta_n) |_{\alpha = \alpha_n}, \\
		\nabla_{\rho} s_n(\alpha_n, \rho_n, \beta_n) &= 
		\frac{1}{n}\sum_{i=1}^n\nabla_{\rho}s(a(X_i) + \alpha_n \xi_n(X_i), \nu(X_i) + \rho \zeta_n(X_i), \beta_n) |_{\rho = \rho_n}, \\
		\nabla_{\beta} s_n(\alpha_n, \rho_n, \beta_n) &= 
		\frac{1}{n}\sum_{i=1}^n\nabla_{\beta}s(a(X_i) + \alpha_n \xi_n(X_i), \nu(X_i) + \rho_n \zeta_n(X_i), \beta) |_{\beta = \beta_n}.
	\end{align*}
	Similarly we have the second order derivatives.
	
	We first show $\beta_n$ is consistent. Taylor's expansion of $s_n(\alpha_n, \rho_n, \beta_n) $ at $\alpha_n = \rho_n = 0$ is 
	\begin{align}\label{eq:Taylor1}
		\begin{split}
			s_n(\alpha_n, \rho_n, \beta_n) 
			&=  s_n(0, 0, \beta_n) + \nabla_{\alpha} s_n(\alpha_{\varepsilon}, \rho_{\varepsilon}, \beta_n)\alpha_n + \nabla_{\rho} s_n(\alpha_{\varepsilon}, \rho_{\varepsilon}, \beta_n)\rho_n, 
		\end{split}
	\end{align}
	where $\alpha_{\varepsilon} \in [0,\alpha_n]$, $\rho_{\varepsilon} \in [0,\rho_n]$. Furthermore, Taylor's expansion of $s_n(0,0,\beta_n)$ at the true $\beta$ is
	\begin{align}\label{eq:Taylor1.5}
		\begin{split}
			s_n(0, 0, \beta_n) = s_n(0, 0, \beta) + \nabla_{\beta} s_n(0, 0, \beta_{\varepsilon})(\beta_n - \beta),
		\end{split}
	\end{align}
	where $\beta_{\varepsilon} \in [\beta_n, \beta]$.
	By central limit theorem (CLT), \begin{align}\label{eq:CLT}
		s_n(0,0,\beta) = s(0,0,\beta) + O_p(n^{-1/2}) = O_p(n^{-1/2}). 
	\end{align}
	Under the regularity condition, $s_n(0, 0, \beta)$ is Lipschitz in $\beta$ and $\calB$ is bounded,  then by the uniform law of large numbers (ULLN),
	\begin{align}\label{eq:LLN}
		\sup_{\beta' \in \calB}|\nabla_{\beta} s_n(0, 0, \beta')
		- \nabla_{\beta} s(0, 0, \beta')| = o_p(1).
	\end{align}
	Thus for $n$ large enough, the minimal eigenvalue of $\nabla_{\beta} s_n(0, 0, \beta_{\varepsilon})$ is at least half of the minimal eigenvalue of $\nabla_{\beta} s(0, 0, \beta_{\varepsilon})$ and is lower bounded by $C/2$.
	Notice that $\nabla_{\alpha} s_n(\alpha_n, \rho_n, \beta_n)$, $\nabla_{\rho} s_n(\alpha_n, \rho_n, \beta_n)$ are bounded,
	we plug Eq.~\eqref{eq:Taylor1.5}, \eqref{eq:CLT}, \eqref{eq:LLN} in Eq.~\eqref{eq:Taylor1},
	\begin{align*}
		0 
		&= s_n(\alpha_n, \rho_n, {\beta}_n) 
		= \nabla_{\beta} s_n(0, 0, \beta_{\varepsilon})(\beta_n - \beta)  + O_p(n^{-1/2} + \alpha_n + \rho_n)\\
		\implies &{\beta}_n - \beta
		= \left(\nabla_{\beta} s_n(0, 0, \beta_{\varepsilon})\right)^{-1} O_p(n^{-1/2} + \alpha_n + \rho_n) = o_p(1).
	\end{align*}
	Therefore, ${\beta}_n$ is consistent.
	
	We next prove the convergence rate of ${\beta}_n$. The second order Taylor's expansion of $s_n(\alpha_n, \rho_n, \beta_n) $ at $\beta_n = \beta$ and $\alpha = \rho = 0$ is 
	\begin{align}\label{eq:Taylor2}
		\begin{split}
			s_n(\alpha_n, \rho_n, \beta_n) 
			&=  s_n(0, 0, \beta_n) + \nabla_{\alpha} s_n(0, 0, \beta_n)\alpha_n + \nabla_{\rho} s_n(0,0, \beta_n)\rho_n \\
			&\quad ~+ \frac{1}{2} \nabla^2_{\alpha} s_n(\alpha_{\varepsilon}, \rho_{\varepsilon}, \beta_n)\alpha_n^2 + \frac{1}{2} \nabla^2_{\rho} s_n(\alpha_{\varepsilon}, \rho_{\varepsilon}, \beta_n)\rho_n^2 + \nabla_{\alpha\rho} s_n(\alpha_{\varepsilon}, \rho_{\varepsilon}, \beta_n)\alpha_n \rho_n\\
			&= s_n(0, 0, \beta) + \nabla_{\beta} s_n(0, 0, \beta_{\varepsilon})(\beta_n - \beta) + \nabla_{\alpha} s_n(0, 0, \beta_n)\alpha_n + \nabla_{\rho} s_n(0,0, \beta_n)\rho_n \\
			&\quad ~+ \frac{1}{2} \nabla^2_{\alpha} s_n(\alpha_{\varepsilon}, \rho_{\varepsilon}, \beta_n)\alpha_n^2 + \frac{1}{2} \nabla^2_{\rho} s_n(\alpha_{\varepsilon}, \rho_{\varepsilon}, \beta_n)\rho_n^2 + \nabla_{\alpha\rho} s_n(\alpha_{\varepsilon}, \rho_{\varepsilon}, \beta_n)\alpha_n \rho_n,
		\end{split}
	\end{align}
	where $\beta_{\varepsilon}\in [\beta_n,\beta]$, $\alpha_{\varepsilon} \in [0,\alpha]$, $\rho_{\varepsilon} \in [0,\rho]$.
	The first order Taylor's expansion of $\nabla_{\alpha} s_n(0, 0, \beta_n)$ at $\beta$ is
	\begin{align}\label{eq:Taylor3}
		\begin{split}
			\nabla_{\alpha} s_n(0, 0, \beta_n)
			&= \nabla_{\alpha} s_n(0, 0,\beta) + \nabla_{\alpha\beta} s_n(0, 0, \beta_{\varepsilon})(\beta_n - \beta)\\
			&= \nabla_{\alpha} s(0, 0,\beta) + O_p(n^{-1/2}) + \nabla_{\alpha\beta} s_n(0, 0, \beta_{\varepsilon})(\beta_n - \beta)\\
			&= O_p(n^{-1/2}) + \nabla_{\alpha\beta} s_n(0, 0, \beta_{\varepsilon})(\beta_n - \beta),
		\end{split}
	\end{align}
	where we apply CLT again in the second last equation, and use $\nabla_{\alpha} s(0, 0,\beta) = 0$ in the last equation.
	Similar analysis can be done for $\nabla_{\rho} s_n(0, 0, \beta_n)$.
	Plug Eq.~\eqref{eq:CLT}, \eqref{eq:Taylor3} into Eq.~\eqref{eq:Taylor2},
	\begin{align}\label{eq:Taylor4}
		\begin{split}
			s_n(\alpha_n, \rho_n, \beta)
			&= O_p(n^{-1/2}) +  \nabla_{\beta} s_n(0, 0, \beta_{\varepsilon})(\beta_n - \beta) + O_p(n^{-1/2} (\alpha_n + \rho_n)) \\
			&\quad~+ \nabla_{\alpha\beta} s_n(0, 0, \beta_{\varepsilon})(\beta_n - \beta)\alpha_n + \nabla_{\rho\beta} s_n(0, 0, \beta_{\varepsilon})(\beta_n - \beta)\rho_n \\
			&\quad~+ \frac{1}{2} \nabla^2_{\alpha} s_n(\alpha_{\varepsilon}, \rho_{\varepsilon}, \beta_n)\alpha_n^2 + \frac{1}{2} \nabla^2_{\rho} s_n(\alpha_{\varepsilon}, \rho_{\varepsilon}, \beta_n)\rho_n^2 + \nabla_{\alpha\rho} s_n(\alpha_{\varepsilon}, \rho_{\varepsilon}, \beta_n)\alpha_n \rho_n\\
			&= \left(\nabla_{\beta} s_n(0, 0, \beta_{\varepsilon}) + O_p(\alpha_n + \rho_n) \right)(\beta_n - \beta) + O_p(\alpha_n^2 + \rho_n^2 + \alpha_n \rho_n + n^{-1/2}),
		\end{split}
	\end{align}
	where we use the condition that the second derivatives are bounded. 
	As shown above, the minimal eigenvalue of $\nabla_{\beta} s_n(0, 0, \beta_{\varepsilon})$ is uniformly lower bounded by $C/2$.
	Thus Eq.~\eqref{eq:Taylor4} implies
	\begin{align*}
		\beta_n - \beta
		= \left(\nabla_{\beta} s_n(0, 0, \beta_{\varepsilon}) + O_p(\alpha_n + \rho_n)\right)^{-1} O_p(n^{-1/2} + \alpha_n^2 + \rho_n^2 + \alpha_n \rho_n)
		= O_p(n^{-1/2} + c_n^2),
	\end{align*}
	and we finish the proof.
\end{proof}

\subsection{Cox model}\label{subsec:appendix:cox}

\begin{proof}[Proof of Claim \ref{claim:coxExp}]
	Since $\lambda(y) > 0$, then $\Lambda(y)$ is strictly increasing and there exists an inverse function $\Lambda^{-1}(z)$. We compute the cumulative distribution function of $Z = \Lambda(Y)$,
	\begin{align*}
		&\quad ~\PP(Z > z \mid X = x, W = w)
		= \PP(\Lambda(Y) > z \mid X = x, W = w)\\
		&= \PP(Y > \Lambda^{-1}(z) \mid X = x, W = w)
		= e^{-\Lambda(\Lambda^{-1}(z)) e^{\eta_w(x)} }
		= e^{-z e^{\eta_w(x)} }.
	\end{align*}
	Therefore $\Lambda(Y) \mid X = x, W= w$ follows the exponential distribution with rate $e^{\eta_w(x)}$.
\end{proof}

\begin{proposition}\label{prop:derivationCox}
	Let $Y$ be generated from the model \eqref{eq:model:cox} with baseline $\Lambda(y)$, parameter functions $\eta_w(x) = \nu(x) + (w-a(x)) x^\top \beta$, and random censorship.
	For arbitrary $\nu^*(x)$, $\beta'$, let $\ell(Y; \beta', a(x),  \nu^*(x), \Lambda(y))$ be the full log-likelihood of  $Y$ evaluated at parameter functions $\nu^*(x) + (w-a(x)) x^\top \beta'$.  Then $a(x)$, $\nu(x)$ in \eqref{eq:conditionCox3} ensure 
	\begin{align*}
		\beta = \argmax_{\beta'} \EE\left[\ell(Y; \beta', a(x),  \nu^*(x), \Lambda(y))\right]
	\end{align*}	
	is true under the log-likelihood quadratic approximation in Lemma \ref{lemm:likelihoodApproximationExponential}.
\end{proposition}

\begin{proof}[Proof of Proposition \eqref{prop:derivationCox}]\label{proof:robust:cox:full}
	We write down the full log-likelihood
	\begin{align}\label{eq:cox:fulllike}
		\ell(Y^c; \eta_w(x), \Lambda(y))
		= \Delta (\lambda(Y^c) + \eta_W(X))  - \Lambda(Y^c) e^{\eta_W(X)}.
	\end{align}
	We stick to the reparametrization \eqref{eq:model:exponentialFamily2} and derive the $a(x)$ robust to baseline model misspecifications. Let $\nu^*(x)$ be an approximation of the baseline function, and $\eta^*_w(x) := \nu^*(x) + (w-a(x))x^\top \beta$.
	A similar quadratic approximation as Lemma \ref{lemm:likelihoodApproximationExponential} gives us\footnote{We omit the part independent of $\beta'$ in the right hand side of \eqref{eq:cox:fulllike2}.}
	\begin{align}\label{eq:cox:fulllike2}
		\begin{split}
			\ell(Y^c; \beta', &a(x), \nu^*(x), \Lambda(y)) 
			\approx  \ell(Y^c;\beta, a(x), \nu^*(x),\Lambda(y))\\ 
			&-\frac{1}{2} \Lambda(Y^c) e^{\eta_W^*(X)} \left(\frac{\Delta}{\Lambda(Y^c) e^{\eta_W^*(X)}} - 1 - (W - a(X)) X^\top (\beta' -\beta)\right)^2.
		\end{split}
	\end{align}
	Maximizing the expectation of the quadratic term yields
	\begin{align}\label{eq:solutionCox}
		\beta' -\beta = \EE\left[\Lambda(Y^c) e^{\eta_W^*(X)} (W-a(X))^2XX^\top\right]^{-1} \EE\left[ (W-a(X))X\left(\Delta -  \Lambda(Y^c) e^{\eta_W^*(X)}\right) \right].
	\end{align}
	Our target $\beta' = \beta$ is equivalent to
	\begin{align}\label{eq:conditionCox}
		0 
		= \EE\left[(W-a(X))X\left(\Delta - \Lambda(Y^c) e^{\eta_W^*(X)} \right)\right].
	\end{align} 
	Plug Lemma \ref{lemm:equationCox} in \eqref{eq:conditionCox}, we have
	\begin{align}\label{eq:conditionCox2}
		0 
		= \EE\left[(W-a(X))X\left(1 - e^{\eta_W^*(X)-\eta_W(X)} \right)\int \left(1-e^{-\Lambda(c)e^{\eta_W(X)}}\right) f_C(c \mid W, X) ~dc\right],
	\end{align}
	where $f_C(c \mid W, X)$ denotes the conditional density of the censoring time.
	Notice that $\eta^*_w(x) - \eta_w(x) = \nu^*(x) - \nu(x)$ is independent of the treatment assignment $W$ and can be in arbitrary direction, then the condition \eqref{eq:conditionCox2} entails
	\begin{align*}
		\frac{a(x)}{1-a(x)} = \frac{e(x)}{1-e(x)}\frac{\int \left(1-e^{-\Lambda(c)e^{\eta_1(X)}}\right) f_C(c\mid W=1, X) ~dc}{\int \left(1-e^{-\Lambda(c)e^{\eta_0(X)}}\right) f_C(c\mid W=0, X) ~dc}
		= \frac{e(x)}{1-e(x)} \frac{\PP(C \ge Y \mid W=1, X=x )}{\PP(C \ge Y \mid W=0, X = x)}.
	\end{align*}
	
\end{proof}

\begin{lemma}\label{lemm:equationCox} 
	Let $Y$ be generated from the Cox model \eqref{eq:model:cox} with the natural parameter function $\eta_w(x)$. Consider random censorship and let $f_C(c \mid W, X)$ be the conditional density of the censoring time. Then 	
	\begin{align*}
		\begin{split}
			&\EE[\Delta \mid W, X]
			=\PP(C \ge Y \mid W, X )
			= \int \left(1 - e^{-\Lambda(c) e^{\eta_W(X)}}\right) f_C(c \mid W, X) ~dc, \\
			&\EE[\Lambda(Y^c) \mid W, X]
			= e^{-\eta_W(X)}\int \left(1-e^{-\Lambda(c)e^{\eta_W(X)}}\right) f_C(c \mid W, X) ~dc.
		\end{split}
	\end{align*}
\end{lemma}

\begin{proof}[Proof of Lemma \ref{lemm:equationCox}]
	By definition, the first equality in Lemma \ref{lemm:equationCox} is true. For  the second equality, by definition and random censorship, 
	\label{proof:cox:censoring}
	\begin{align*}
		\begin{split}
			\EE\left[\Lambda(Y^c) \mid W, X\right]
			&= \EE\left[\Lambda(Y) \1_{\{C \ge Y\}} \mid W, X\right] + \EE\left[\Lambda(C) \1_{\{C < Y\}} \mid W, X\right]\\
			&=\int_0^\infty \int_{0}^c \Lambda(y) e^{-\Lambda(y)e^{\eta_W(X)}} d\Lambda(y)e^{\eta_W(X)}  ~f_C(c \mid W, X) dc \\
			&\quad~+ \int_0^\infty \int_{c}^{\infty} \Lambda(c) e^{-\Lambda(y)e^{\eta_W(X)}} d\Lambda(y)e^{\eta_W(X)}  ~f_C(c \mid W, X) dc\\
			&= e^{-\eta_W(X)} \int_0^\infty \int_{0}^{\Lambda(c)e^{\eta_W(X)}} z e^{-z} dz  ~f_C(c \mid W, X) dc \\
			&\quad ~ + \int_0^\infty \Lambda(c) e^{-\Lambda(c)e^{\eta_W(X)}} ~f_C(c \mid W, X) dc \\
			&= e^{-\eta_W(X)} \int_0^\infty \left(1-e^{-\Lambda(c)e^{\eta_W(X)}}\left(\Lambda(c)e^{\eta_W(X)}+1\right)\right) f_C(c \mid W, X) dc \\
			&\quad ~+ \int_0^\infty \Lambda(c) e^{-\Lambda(c)e^{\eta_W(X)}} ~f_C(c \mid W, X) dc \\
			&= e^{-\eta_W(X)} \int_0^\infty \left(1 - e^{-\Lambda(c)e^{\eta_W(X)}}\right) f_C(c \mid W, X) dc,
		\end{split}
	\end{align*}
	where we use $\int z e^{-z} dz = - e^{-z}(z+1)$.
\end{proof}


\begin{proof}[Proof of Proposition~\ref{prop:mainCoxFull}]
	To show Proposition~\ref{prop:mainCoxFull} is true, we only need to verify the score function is Neyman's orthogonal, and the rest of the proof is the same as that of Proposition~\ref{prop:main}. 
	
	The score equation of the full likelihood function at $a_n(x) = a(x) + \alpha_n \xi_n(x)$, $\nu_{n}(x) = \nu(x) + \rho_n \zeta_n(x)$ is
	\begin{align}\label{eq:proof:cox:scoreFull}
		\begin{split}
			s(\alpha_n, \rho_n, \beta_n) = \EE\left[\Delta (W-a_{n}(X))X - \Lambda(Y^c) (W-a_{n}(X))X e^{\nu_{n}(X) + (W-a_{n}(X))X^\top \beta_n}\right]. 
		\end{split}
	\end{align}
	The derivative of \eqref{eq:proof:cox:scoreFull} with regard to $\rho$ evaluated at $\alpha = \rho = 0$ and $\beta_n = \beta$ is 
	\begin{align}\label{eq:proof:cox:scoreRhoFull}
		\begin{split}
			\nabla_{\rho} s(0,0, \beta)
			&= -\EE\left[\Lambda(Y^c) (W-a(X))X e^{\nu(X) + (W-a(X))X^\top \beta}\zeta(X)\right]. 
		\end{split}
	\end{align}
	By Lemma \ref{lemm:equationCox} and \eqref{eq:conditionCox3}, \eqref{eq:proof:cox:scoreRhoFull} becomes
	\begin{align*}
		\begin{split}
			\nabla_{\rho} s(0,0, \beta) 
			&= -\EE\left[\PP(C \ge Y \mid W, X) e^{-\eta_W(X)} (W-a(X))X e^{\nu(X) + (W-a(X))X^\top \beta }\zeta(X)\right]\\
			&= -\EE\left[\PP(C \ge Y \mid W, X) (W-a(X))X \zeta(X)\right] = 0.
		\end{split}
	\end{align*}
	Next, the derivative of \eqref{eq:proof:cox:scoreFull} with regard to $\alpha$ evaluated at $\alpha = \rho = 0$ and $\beta_n = \beta$ is 
	\begin{align}\label{eq:proof:cox:scoreAlphaFull}
		\begin{split}
			\nabla_{\alpha} s(0,0, \beta) 
			&= -\EE\left[\Delta \xi(X)X\right] + \EE\left[\Lambda(Y^c) \xi(X)X e^{\nu(X) + (W-a(X))X^\top \beta }\right] \\
			&\quad~ + \EE\left[\Lambda(Y^c) (W-a(X))X e^{\nu(X) + (W-a(X))X^\top \beta }\xi(X) X^\top \beta\right]. 
		\end{split}
	\end{align}
	The third term in \eqref{eq:proof:cox:scoreAlphaFull} is zero by the same argument of \eqref{eq:proof:cox:scoreRhoFull}. For the second term, by Lemma \ref{lemm:equationCox}, 
	\begin{align*}
		\begin{split}
			\EE\left[\Lambda(Y^c) \xi(X)X e^{\nu(X) + (W-a(X))X^\top \beta }\right]
			= \EE\left[\PP(C \ge Y \mid W, X) \xi(X) X \right] 
			= \EE\left[\Delta \xi(X)X\right].
		\end{split}
	\end{align*}
	Therefore, combine the three terms and we prove \eqref{eq:proof:cox:scoreAlphaFull} is zero. In this way, we have verified the Neyman's orthogonality of the full likelihood score.
\end{proof}


\begin{proof}[Analysis of  Algorithm~\ref{algo:coxPartial}]\label{proof:cox:partial:robustness}
	The partial likelihood proposed by \citep{cox1972regression} is
	\begin{align}\label{eq:cox:partiallike}
		\text{pl}_n(\eta_w(x))
		= \frac{1}{n}\sum_{\Delta_i = 1} \left(\eta_{W_i}(X_i) - \log\left(\sum_{\calR_i} e^{\eta_{W_i}(X_i)}\right) \right), 
	\end{align}
	where the risk set of subject $i$ is defined as $\calR_i = \{j: Y_j^c \ge Y_i^c\}$. We follow the notations in \citep{tsiatis1981large} and analyse the asymptotic value of the partial likelihood \eqref{eq:cox:partiallike}
	\begin{align}\label{eq:cox:partiallike2}
		\text{pl}(\eta_w(x))
		=  \EE\left[\Delta \eta_W(X) \right] - \int_{0}^\infty \log\left(\EE\left[e^{\eta_W(X)} 1_{\{Y^c \ge y\}}\right]\right) d\PP(y)
	\end{align}
	where $\PP(y) = \EE[Y_i^c \le y]$ is the expected probability that the censored survival time of a randomly selected unit is no longer than $y$. We continue to use the notations in the full likelihood analysis and obtain by Lemma \ref{lemm:likelihoodApproximationCox}
	\begin{align}\label{eq:cox:partiallike3}
		\begin{split}
			\text{pl}(\eta'_w(x))
			&\approx  \text{pl}(\eta^*_w(x)) + (\beta' - \beta)^\top \left(\underbrace{\EE\left[ (\nu(X) - \nu^*(X))\PP(C \ge Y \mid W, X) (W-a(X))X\right]}_{:=\text{(a)}}  \right. \\
			&\left.\quad\underbrace{\int_{0}^\infty \frac{\EE\left[(\nu(X) - \nu^*(X)) e^{\eta_W^*(X)}1_{\{Y^c \ge y\}}\right]}{\EE\left[e^{\eta_W^*(X)} 1_{\{Y^c \ge y\}}\right]}\EE\left[e^{\eta_W^*(X)} 1_{\{Y^c \ge y\}}(W-a(X))X\right] d\PP(y)}_{:=\text{(b)}} \right) \\
			&\quad - (\beta' - \beta)^\top \Sigma (\beta' - \beta),
		\end{split}
	\end{align}
	where $\Sigma$ is a positively definite weight matrix. To guarantee the robustness to baseline parameters, we look for $a(x)$ such that $(a) + (b)$ in \eqref{eq:cox:partiallike3} is zero. Nevertheless, the $a(x)$ in \eqref{eq:conditionCox3} ensures part $(a)$ to be zero, but not part $(b)$, and thus Proposition~\ref{prop:main} is not valid. We comment that part $(b)$ is difficult to handle due to the three expectations and the integration, which inhibits the construction of the desired $a(x)$.
\end{proof}

\begin{lemma}\label{lemm:likelihoodApproximationCox}
	Let $Y$ be generated from the Cox model \eqref{eq:model:cox} with the natural parameter function $\eta_w(x) = \nu(x) + (w-a(x)) x^\top \beta$. Then for $\eta^*_w(x) = \nu^*(x) + (w-a(x)) x^\top \beta$ and $\eta'_w(x) = \nu^*(x) + (w-a(x)) x^\top \beta'$, the asymptotic partial likelihood satisfies 
	\begin{align*}
		\begin{split}
			\text{pl}(\eta'_w(x))
			&=  \text{pl}(\eta^*_w(x)) + (\beta' - \beta)^\top \left(\underbrace{\EE\left[ (\nu(X) - \nu^*(X))\PP(C \ge Y \mid W, X) (W-a(X))X\right]}_{:=\text{(a)}}  \right. \\
			&\left.\quad\underbrace{\int_{0}^\infty \frac{\EE\left[(\nu(X) - \nu^*(X)) e^{\eta_W^*(X)}1_{\{Y^c \ge y\}}\right]}{\EE\left[e^{\eta_W^*(X)} 1_{\{Y^c \ge y\}}\right]}\EE\left[e^{\eta_W^*(X)} 1_{\{Y^c \ge y\}}(W-a(X))X\right] d\PP(y)}_{:=\text{(b)}} \right) \\
			&\quad - (\beta' - \beta)^\top \Sigma (\beta' - \beta) + O\left(\|\beta' - \beta\|_2^3\right),
		\end{split}
	\end{align*}
	where $\Sigma$ is a positively definite weight matrix.
\end{lemma}

\begin{proof}[Proof of Lemma \ref{lemm:likelihoodApproximationCox}]
	By \eqref{eq:cox:partiallike2}, the asymptotic partial likelihood obeys
	\begin{align*}
		\text{pl}(\eta_w'(x))
		&= \text{pl}(\eta_w^*(x)) + \underbrace{\EE\left[\Delta (W-a(X))X^\top (\beta' -\beta) \right]}_{:=\text{(I)}} - \underbrace{\int_{0}^\infty \log\left(\frac{\EE\left[e^{\eta_W'(X)} 1_{\{Y^c \ge y\}}\right]}{\EE\left[e^{\eta_W^*(X)} 1_{\{Y^c \ge y\}}\right]}\right) d\PP(y)}_{:=\text{(II)}}.
	\end{align*}
	For part (I), by the random censoring mechanism,
	\begin{align}\label{eq:proof:cox:(I)}
		\text{(I)} = \EE\left[\PP(C \ge Y \mid W, X) (W-a(X))X^\top (\beta' -\beta) \right].
	\end{align}
	For part (II), we apply the approximation $\log(1+x) = x - x^2/2 + O(x^3)$ and get 
	\begin{align*}
		\begin{split}
			\text{(II)} &= \underbrace{\int_{0}^\infty \frac{\EE\left[\left(e^{\eta_W'(X)} - e^{\eta_W^*(X)} \right)1_{\{Y^c \ge y\}}\right]}{\EE\left[e^{\eta_W^*(X)} 1_{\{Y^c \ge y\}}\right]} d\PP(y)}_{:=\text{(II.a)}} \\
			&\quad ~- \underbrace{\int_{0}^\infty \frac{1}{2} \left(\frac{\EE\left[\left(e^{\eta_W'(X)} - e^{\eta_W^*(X)} \right)1_{\{Y^c \ge y\}}\right]}{\EE\left[e^{\eta_W^*(X)} 1_{\{Y^c \ge y\}}\right]}\right)^2 d\PP(y)}_{:=\text{(II.b)}} + r,
		\end{split}
	\end{align*}
	where the residual $r = O\left( \int_{0}^\infty \left(\frac{\EE\left[\left(e^{\eta_W'(X)} - e^{\eta_W^*(X)} \right)1_{\{Y^c \ge y\}}\right]}{\EE\left[e^{\eta_W^*(X)} 1_{\{Y^c \ge y\}}\right]}\right)^3 d\PP(y) \right) = O(\|\eta_W'(X) - \eta_W^*(X)\|_2^3)$, and thus is $O(\|\beta - \beta'\|_2^3)$.
	Notice that
	\begin{align*}
		\begin{split}
			&\quad ~\EE\left[\left(e^{\eta_W'(X)} - e^{\eta_W^*(X)} \right)1_{\{Y^c \ge y\}}\right]\\
			&=\EE\left[\left(\eta_W'(X) -  \eta_W^*(X)\right)e^{\eta_W^*(X)} 1_{\{Y^c \ge y\}}\right] + O\left(\EE\left[\left(\eta_W'(X) -  \eta_W^*(X)\right)^2\right]\right)\\
			&= \EE\left[(W-a(X))X^\top \left(\beta' - \beta\right)e^{\eta_W^*(X)} 1_{\{Y^c \ge y\}}\right] + O\left(\EE\left[\left(\eta_W'(X) -  \eta_W^*(X)\right)^2\right]\right)\\
			&= \left(\beta' - \beta\right)^\top \EE\left[(W-a(X)) X e^{\eta_W^*(X)} 1_{\{Y^c \ge y\}}\right] + O\left(\EE\left[\left(\eta_W'(X) -  \eta_W^*(X)\right)^2\right]\right),
		\end{split}
	\end{align*}
	thus (II.b) is approximately a quadratic term of $(\beta - \beta')$.
	
	For part (II.a), notice that $d\PP(y) = \lambda(y) \EE\left[e^{\eta_W(X)}1_{\{Y^c \ge y\}}\right]$,
	\begin{align*}
		\begin{split}
			\text{(II.a)} &= \underbrace{\int_{0}^\infty \EE\left[\left(e^{\eta_W'(X)} - e^{\eta_W^*(X)} \right)1_{\{Y^c \ge y\}}\right] \lambda(y) dy}_{:=\text{(II.a.1)}} \\
			&\quad ~+ \underbrace{\int_{0}^\infty \frac{\EE\left[\left(e^{\eta_W'(X)} - e^{\eta_W^*(X)} \right)1_{\{Y^c \ge y\}}\right] \EE\left[\left(e^{\eta_W(X)} - e^{\eta_W^*(X)} \right)1_{\{Y^c \ge y\}}\right]}{\EE\left[e^{\eta_W^*(X)} 1_{\{Y^c \ge y\}}\right]} d\PP(y)}_{:=\text{(II.a.2)}}.
		\end{split}
	\end{align*}
	We first compute the first term (II.a.1). By Assumption~\eqref{assu:censor} and $\PP(Y \ge y \mid X = x, W = w) = e^{-\Lambda(y) e^{-\eta_w(x)}}$, we exchange the order of expectation and integral,
	\begin{align}
		\label{eq:proof:cox:(II.a.1)}
		\begin{split}
			\text{(II.a.1)} &= \EE\left[ \int_{0}^\infty \left(e^{\eta_W'(X)} - e^{\eta_W^*(X)} \right)\1_{\{Y \ge y\}} 1_{\{C \ge y\}} \lambda(y) dy \right] \\
			&= \EE\left[ \left(e^{\eta_W'(X)} - e^{\eta_W^*(X)} \right) \int_0^\infty \int_{0}^c e^{-\Lambda(y)e^{\eta_W(X)}} \lambda(y) dy ~f_C(c \mid W, X) dc\right] \\
			&= \EE\left[ \left(e^{\eta_W'(X) - \eta_W^*(X)} - 1 \right) e^{\eta_W^*(X) - \eta_W(X)} \int_0^\infty \int_{0}^c e^{-\Lambda(y)e^{\eta_W(X)}} d\Lambda(y)e^{\eta_W(X)}  ~f_C(c \mid W, X) dc\right] \\
			&= \EE\left[ \left(e^{\eta_W'(X) - \eta_W^*(X)} - 1 \right) e^{\eta_W^*(X) - \eta_W(X)} \int_0^\infty  \left(1 - e^{-\Lambda(c) e^{\eta_W(X)}}\right) ~f_C(c \mid W, X) dc\right] \\
			&= \EE\left[ \left(e^{\eta_W'(X) - \eta_W^*(X)} - 1 \right) e^{\eta_W^*(X) - \eta_W(X)} \PP(C \ge Y \mid W, X) \right], 
		\end{split}
	\end{align}
	where we use Lemma \ref{lemm:equationCox} in the last equality. Since $\eta_W'(X) - \eta_W^*(X) = (W-a(X))X^\top(\beta' - \beta)$,
	\begin{align}\label{eq:proof:cox:(II.b.1)2}
		\begin{split}
			\text{(II.a.1)} 
			&= \EE\left[ \left(\eta_W'(X) - \eta_W^*(X)\right) e^{\eta_W^*(X) - \eta_W(X)} \PP(C \ge Y \mid W, X) \right] \\ 
			&\quad ~ + (\beta' - \beta)^\top \Sigma (\beta' - \beta) + O\left(\|\beta' - \beta\|_2^3\right)\\
			&= \EE\left[ (W-a(X))X^\top(\beta' - \beta)e^{\eta_W^*(X) - \eta_W(X)} \PP(C \ge Y \mid W, X)  \right] \\
			&\quad~+ (\beta' - \beta)^\top \Sigma (\beta' - \beta) + O\left(\|\beta' - \beta\|_2^3\right).
		\end{split}
	\end{align}
	Combine (II.a.1) in \eqref{eq:proof:cox:(II.b.1)2} and (I) in \eqref{eq:proof:cox:(I)}, 
	\begin{align}\label{eq:proof:cox:(I)(II.b.1)}
		\begin{split}
			\text{(I)} - \text{(II.a.1)} 
			&= \EE\left[ \PP(C \ge Y \mid W, X)(W-a(X))X^\top(\beta' - \beta)(1 - e^{\eta_W^*(X) - \eta_W(X))}\right]\\
			&\quad ~ - (\beta' - \beta)^\top \Sigma (\beta' - \beta) + O\left(\|\beta' - \beta\|_2^3\right)\\
			&\approx \EE\left[ \PP(C \ge Y \mid W, X) (W-a(X))X^\top(\beta' - \beta)(\nu(X) - \nu^*(X))\right] \\
			&\quad ~ - (\beta' - \beta)^\top \Sigma (\beta' - \beta) + O\left(\|\beta' - \beta\|_2^3\right).	
		\end{split}
	\end{align}
	In this way, $\text{(I)} - \text{(II.a.1)}$ gives part (a) in Lemma \ref{lemm:likelihoodApproximationCox}. Furthermore, (II.a.2) gives part (b) in Lemma \ref{lemm:likelihoodApproximationCox}. Combine all parts and we finish the proof.
\end{proof}

\begin{proof}[Proof of Proposition~\ref{claim:cox:null}]
	
	To show Proposition~\ref{claim:cox:null} is true, we only need to verify the asymptotic score is Neyman's orthogonal if $\tau(x) = 0$, and the rest of the proof is the same as Proposition~\ref{prop:main}. 
	
	According to \citep{tsiatis1981large}, the score equation of the partial likelihood function at $a_{n}(x) = a(x) + \alpha_n \xi_n(x)$, $\nu_{n}(x) = \nu(x) + \rho_n \zeta_n(x)$, and $\beta_n$ is asymptotically 
	\begin{align}\label{eq:proof:cox:score}
		\begin{split}
			s(\alpha_n, \rho_n, \beta_n) = \EE[\Delta (W-a_{n}(X))X) ] - \int_{0}^\infty \frac{\EE\left[1_{\{Y^c \ge y\}}e^{\nu_n(X) + (W-a_n(X)) X^\top \beta_n} (W-a_{n}(X)) X \right]}{\EE\left[1_{\{Y^c \ge y\}}e^{\nu_n(X) + (W-a_{n}(X)) X^\top \beta_n} \right]} d\PP(y),
		\end{split}
	\end{align}
	where $d\PP(y) = \lambda(y) \EE\left[e^{\eta_W(X)}1_{\{Y^c \ge y\}}\right]$ with the true parameter $\eta_w(x)$.
	The derivative of the score \eqref{eq:proof:cox:score} with regard to $\rho$ evaluated at $\alpha = \rho = 0$ and $\beta_n = \beta$ is  
	\begin{align}\label{eq:proof:cox:scoreRho}
		\begin{split}
			&\nabla_{\rho} s(0,0, \beta) 
			= \underbrace{\int_{0}^\infty \EE\left[1_{\{Y^c \ge y\}}e^{\nu(X) + (W-a(X)) X^\top \beta} \zeta(X)(W-a(X)) X \right] \lambda(y) dy}_{:=\text{(I)}} \\
			&- \underbrace{\int_{0}^\infty \frac{\EE\left[1_{\{Y^c \ge y\}}e^{\nu(X) + (W-a(X)) X^\top \beta} (W-a(X)) X \right] \EE\left[1_{\{Y^c \ge y\}}e^{\nu(X) + (W-a(X)) X^\top \beta} \zeta(X) \right]}{\EE\left[1_{\{Y^c \ge y\}}e^{\nu(X) + (W-a(X)) X^\top \beta} \right]} \lambda(y) dy}_{:=\text{(II)}}. 
		\end{split}
	\end{align}
	Part (I) is zero for arbitrary $\beta$ by the analysis of (II.a) in \eqref{eq:proof:cox:(II.a.1)} and \eqref{eq:conditionCox3}. As for part (II), recall the definition of $a(x)$ and  that $\beta = 0$,
	\begin{align*}
		\begin{split}
			&\quad~ \EE\left[1_{\{Y^c \ge y\}}e^{\nu(X) + (W-a(X)) X^\top \beta} (W-a(X)) X \right]
			= \EE\left[1_{\{Y^c \ge y\}}e^{\nu(X)} (W-a(X)) X\right]  = 0. 
		\end{split}
	\end{align*}
	Next, the derivative of the score \eqref{eq:proof:cox:score} with regard to $\alpha$ evaluated at $\alpha = \rho = 0$, $\beta_n = \beta$ is 	
	\begin{align}\label{eq:proof:cox:scoreAlpha}
		\begin{split}
			&\nabla_{\alpha} s(\alpha, \rho, \beta) 
			= -\EE[\Delta \xi(X) X]\\
			&+\underbrace{\int_{0}^\infty \EE\left[1_{\{Y^c \ge y\}}e^{\nu(X) + (W-a(X)) X^\top \beta} \xi(X) X^\top \beta(W-a(X)) X \right] \lambda(y) dy}_{:=\text{(I'.a)}} \\
			&+ \underbrace{\int_{0}^\infty \EE\left[1_{\{Y^c \ge y\}}e^{\nu(X) + (W-a(X)) X^\top \beta} \xi(X) X\right] \lambda(y) dy}_{:=\text{(I'.b)}} \\
			&- \underbrace{\int_{0}^\infty \frac{\EE\left[1_{\{Y^c \ge y\}}e^{\nu(X) + (W-a(X)) X^\top \beta} (W-a(X)) X \right] \EE\left[1_{\{Y^c \ge y\}}e^{\nu(X) + (W-a(X)) X^\top \beta} \xi(X) X^\top \beta \right]}{\EE\left[1_{\{Y^c \ge y\}}e^{\nu(X) + (W-a(X)) X^\top \beta} \right]} \lambda(y) dy}_{:=\text{(II')}}. \\
		\end{split}
	\end{align}
	The analysis of part (I'a) and (II') are the same as that of part (I) and (II) in \eqref{eq:proof:cox:scoreRho}. For part (I'b), for arbitrary $\beta$, the analysis of \eqref{eq:proof:cox:(II.a.1)} shows
	\begin{align*}
		\begin{split}
			\text{(I'b)} = \EE[\PP(C \ge Y \mid W, X) \xi(X) X],
		\end{split}
	\end{align*}
	which cancels the term $-\EE[\Delta \xi(X) X]$ in \eqref{eq:proof:cox:scoreAlpha}.
	Then combine all parts together and we have shown \eqref{eq:proof:cox:scoreAlpha} is zero.
\end{proof}

\section{Additional figures}\label{subsec:appendix:figure}

We display the simulation results with gradient boosting as the nuisance-function learners in Figure~\ref{fig:errorBoosting}. In order to demonstrate the robustness of the proposed method to inaccurate nuisance-function estimators, in gradient boosting, we consider a small number ($50$) of  trees, learning rate $0.1$, and tree-depth $3$. We limit the number of trees and the nuisance-function estimators will underfit the data. The results show that the proposed estimator is the most robust to inaccurate nuisance-function estimators. 

\begin{figure}
	\centering
	\begin{minipage}{7cm}
		\centering  
		{(a) {continuous}}
		\includegraphics[scale=0.45]{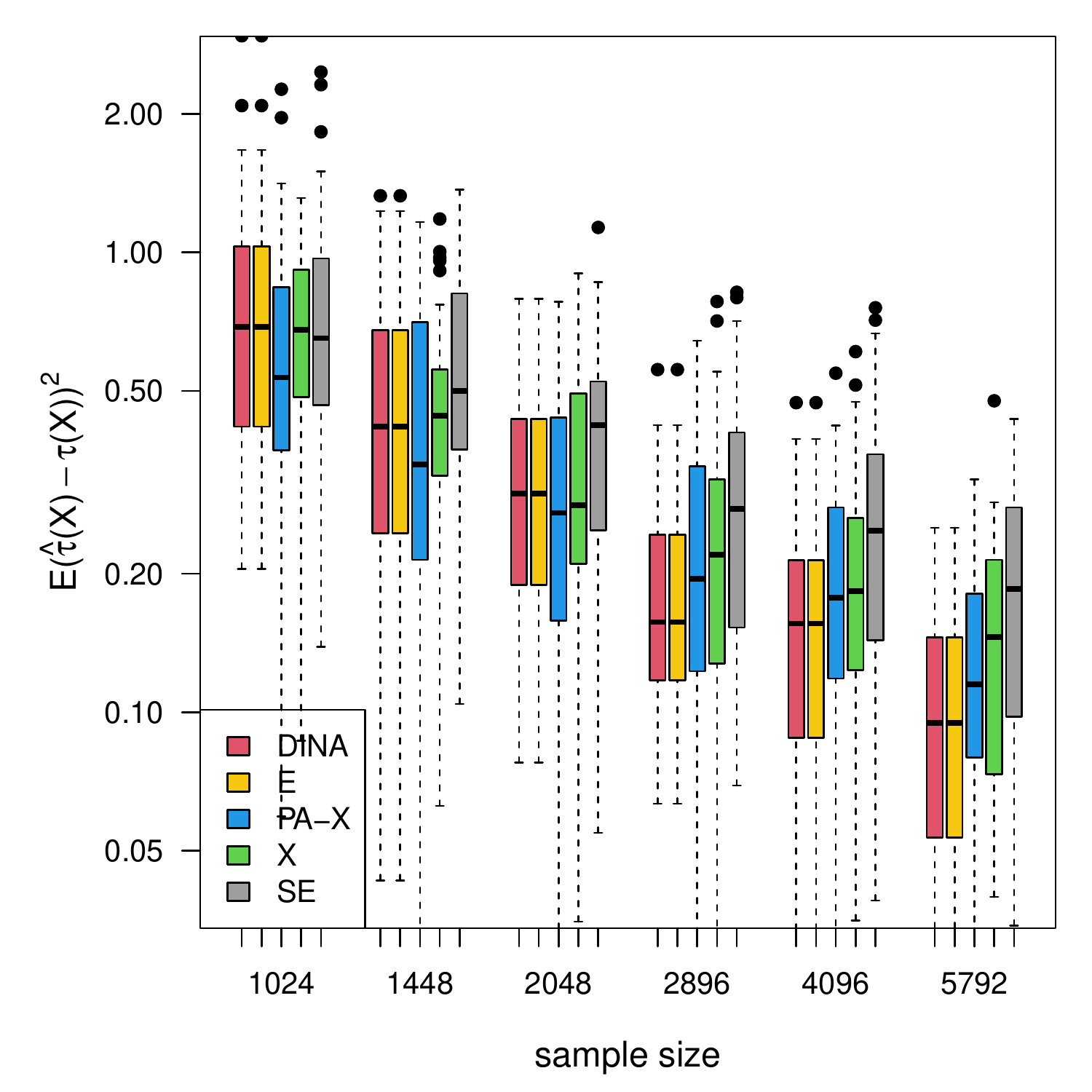}  
	\end{minipage}
	\begin{minipage}{7cm}
		\centering  
		{(b) {binary}}
		\includegraphics[scale=0.45]{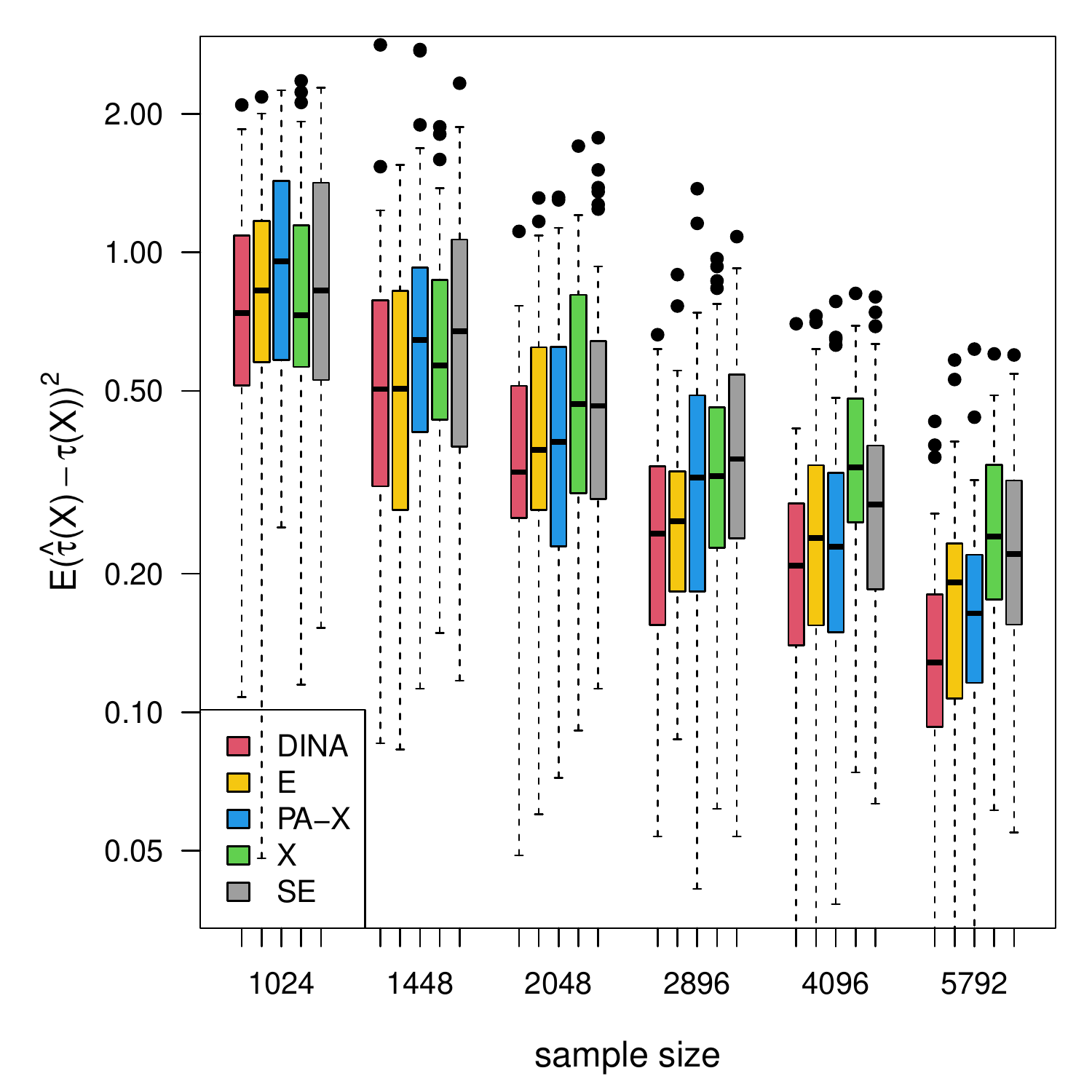}  
	\end{minipage}\\
	\begin{minipage}{7cm}
		\centering  
		{(c) {count data}}
		\includegraphics[scale=0.45]{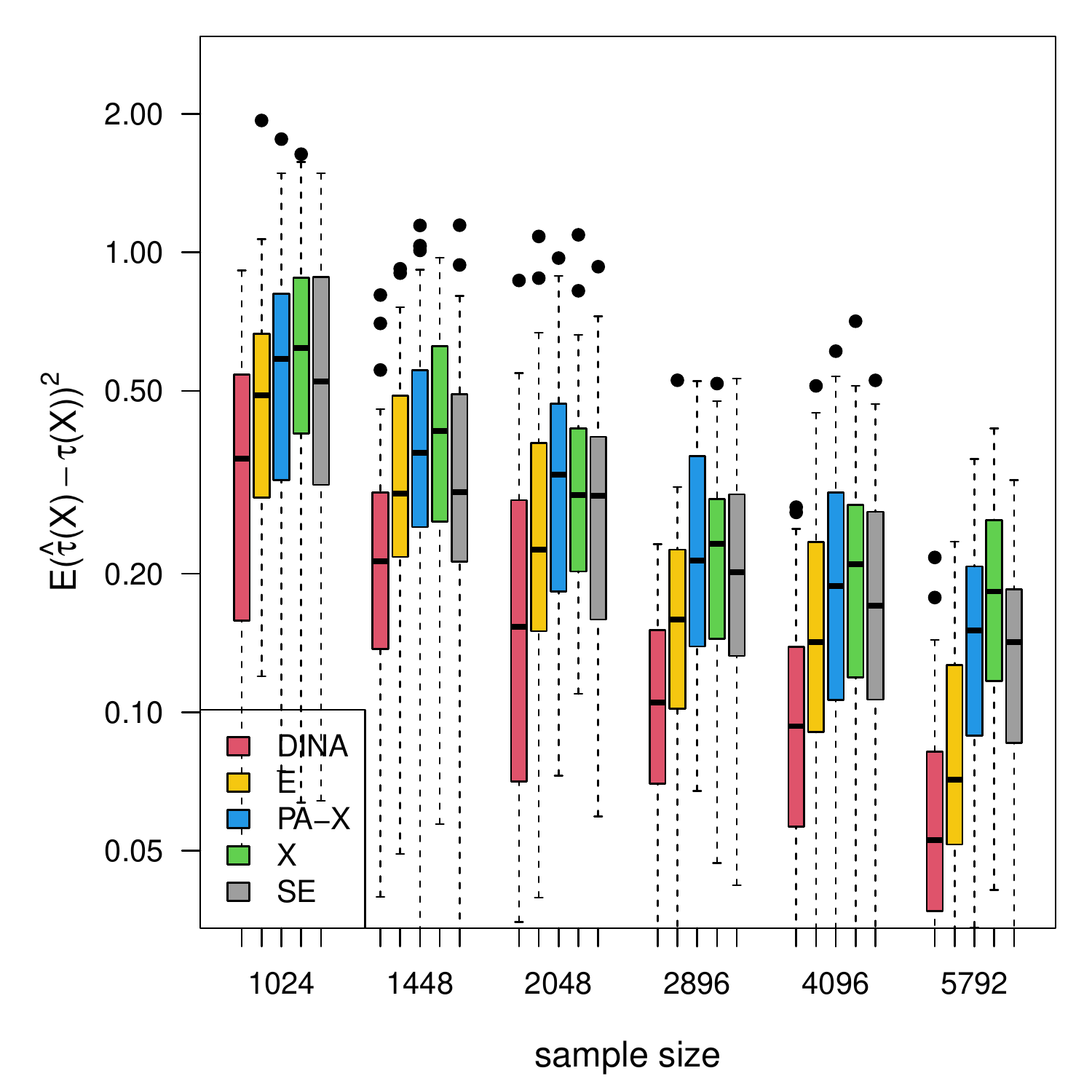}  
	\end{minipage}
	\begin{minipage}{7cm}
		\centering  
		{(d) {survival data}}
		\includegraphics[scale=0.45]{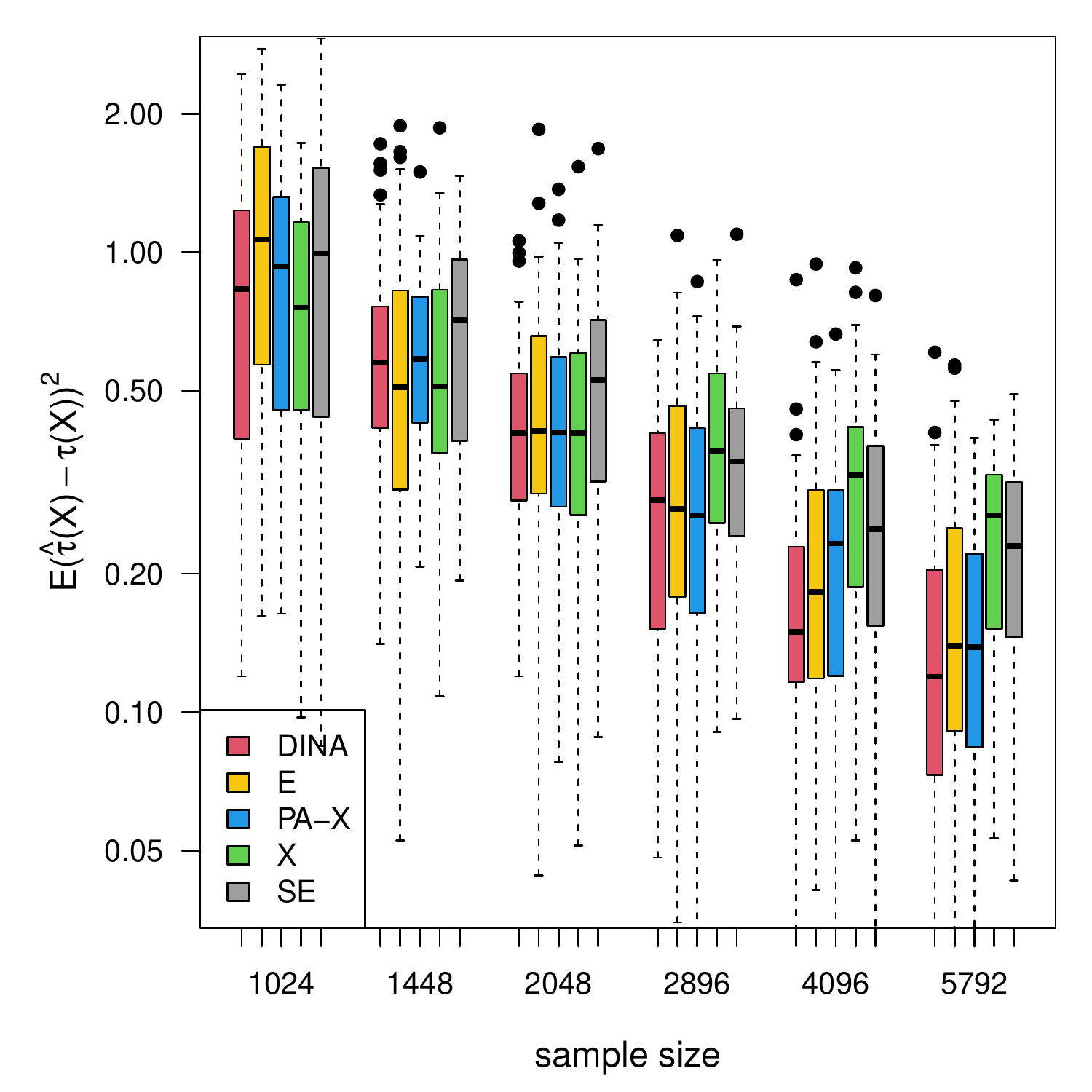}  
	\end{minipage}
	\caption{{Estimation error log-log boxplots. We display
			the estimation errors $\EE[(\hat{\tau}(X) - \tau(X))^2]$
			over sample sizes in $[1024, 1448, 2048, 2896, 4096, 5792]$. We compare five
			methods: separate estimation (SE), X-learner (X),  propensity score adjusted X-learner (PA-X), direct extension of R-learner (E),  and Algorithm~\ref{algo:exponentialFamily}
			(DINA). We adopt the response model~\eqref{eq:simulationModel} and consider four types of responses: (a) continuous (Gaussian); (b) binary (Bernoulli); (c) count data (Poisson); (d) survival data (Cox model with uniform censoring, $75\%$ units censored). We estimate the propensity score by logistic regression, and estimate the natural parameters by gradient boosting. We repeat all experiments $100$ times.
	}}
	\label{fig:errorBoosting}
\end{figure}

\section{Additional tables}\label{subsec:appendix:table}
In this section, we first display the coverage and width of the $95\%$ confidence intervals for various types of responses and nuisance-function estimators. 
We use bootstrap ($100$ bootstrap samples) to construct confidence intervals for $\hat{\beta}$. Table \ref{tab:gaussianGLMCI} considers Gaussian responses and uses linear regression as the nuisance-function estimator; Table \ref{tab:binomialGLMCI} considers binary responses and uses logistic regression as the nuisance-function estimator; Table \ref{tab:coxGLMCI} considers survival responses and uses Cox regression as the nuisance-function estimator; Table \ref{tab:gaussianBoostingCI}, \ref{tab:binomialBoostingCI}, \ref{tab:poissonBoostingCI}, and \ref{tab:coxBoostingCI} use gradient boosting as the nuisance-function estimator and deal with Gaussian, binary, count, and survival responses, respectively.

We observe the proposed method (DINA) produces the best coverages for all $\hat{\beta}$ and sample sizes.
The propensity score adjusted X-learner (PA-X) produces the second-best coverages but requires wider confidence intervals. 
For the direct extension of R-learner (E) (except for Gaussian responses),  X-learner (X), and separate estimation (SE), the coverages of confidence intervals (for $\hat{\beta}_1$ in particular) decrease as the sample size grows. The phenomenon is due to the non-vanishing bias in those estimators.

\begin{table}
		\caption{\label{tab:gaussianGLMCI}Coverage (cvrg) and width of $95\%$ confidence intervals for Gaussian responses. We consider the five meta-algorithms and vary the sample size from $1024$ to $5792$. For Gaussian responses, the proposed method (DINA) and R-learner (E) are the same. We estimate the propensity score by logistic regression and baseline natural parameter functions by linear regression. Confidence intervals are constructed using $100$ bootstrap samples. Results are averaged over $100$ trials.}
	\centering
	\scriptsize
	\fbox{%
	\begin{tabular}{c|c|cc|cc|cc|cc|cc|cc}
		Sample                    & Meth- & \multicolumn{2}{c|}{$\beta_1$} & \multicolumn{2}{c|}{$\beta_2$} & \multicolumn{2}{c|}{$\beta_3$} & \multicolumn{2}{c|}{$\beta_4$} & \multicolumn{2}{c|}{$\beta_5$} & \multicolumn{2}{c}{$\beta_0$} \\ 
		size &    od    & cvrg      & width     & cvrg      & width             & cvrg      & width           & cvrg      & width      & cvrg      & width     & cvrg      & width          \\\hline			
		\multirow{5}{*}{1024} & DINA   & 0.97  & 1.48    & 0.96  & 1.49    & 0.93  & 1.38    & 0.95  & 1.35    & 0.93  & 1.34    & 0.98  & 0.789   \\
		& E      & 0.97  & 1.48    & 0.96  & 1.49    & 0.93  & 1.38    & 0.95  & 1.35    & 0.93  & 1.34    & 0.98  & 0.789   \\
		& PA-X   & 0.97  & 1.56    & 0.93  & 1.56    & 0.91  & 1.46    & 0.94  & 1.44    & 0.96  & 1.47    & 0.97  & 0.787   \\
		& X      & 0.50   & 1.35    & 0.48  & 1.36    & 0.91  & 1.29    & 0.96  & 1.27    & 0.95  & 1.25    & 0.99  & 0.811   \\
		& SE     & 0.46  & 1.36    & 0.53  & 1.34    & 0.87  & 1.26    & 0.91  & 1.27    & 0.91  & 1.28    & 0.92  & 0.795   \\\hline
		\multirow{5}{*}{1448} & DINA   & 0.96  & 1.26    & 0.96  & 1.26    & 0.96  & 1.10     & 0.95  & 1.13    & 0.92  & 1.13    & 0.91  & 0.660    \\
		& E      & 0.96  & 1.26    & 0.96  & 1.26    & 0.96  & 1.10     & 0.95  & 1.13    & 0.92  & 1.13    & 0.91  & 0.660   \\
		& PA-X   & 0.97  & 1.30    & 0.97  & 1.29    & 0.93  & 1.22    & 0.97  & 1.20     & 0.98  & 1.22    & 0.95  & 0.653   \\
		& X      & 0.31  & 1.14    & 0.37  & 1.15    & 0.96  & 1.08    & 0.95  & 1.06    & 0.94  & 1.08    & 0.94  & 0.667   \\
		& SE     & 0.35  & 1.14    & 0.40  & 1.17    & 0.97  & 1.08    & 0.92  & 1.05    & 0.94  & 1.08    & 0.97  & 0.669   \\\hline
		\multirow{5}{*}{2048} & DINA   & 0.93  & 1.04    & 0.94  & 1.04    & 0.94  & 0.957   & 0.95  & 0.955   & 0.94  & 0.947   & 0.93  & 0.549   \\
		& E       & 0.93  & 1.04    & 0.94  & 1.04    & 0.94  & 0.957   & 0.95  & 0.955   & 0.94  & 0.947   & 0.93  & 0.549   \\
		& PA-X   & 0.93  & 1.07    & 0.95  & 1.05    & 0.95  & 1.00       & 0.96  & 1.01    & 0.99  & 0.991   & 0.95  & 0.544   \\
		& X      & 0.25  & 0.975   & 0.24  & 0.946   & 0.95  & 0.891   & 0.97  & 0.886   & 0.95  & 0.898   & 0.94  & 0.576   \\
		& SE     & 0.24  & 0.950    & 0.32  & 0.949   & 0.97  & 0.883   & 0.94  & 0.897   & 0.94  & 0.871   & 0.94  & 0.562   \\\hline
		\multirow{5}{*}{2896} & DINA   & 0.93  & 0.869   & 0.94  & 0.884   & 0.94  & 0.796   & 0.96  & 0.794   & 0.98  & 0.789   & 0.93  & 0.454   \\
		& E     & 0.93  & 0.869   & 0.94  & 0.884   & 0.94  & 0.796   & 0.96  & 0.794   & 0.98  & 0.789   & 0.93  & 0.454    \\
		& PA-X   & 0.96  & 0.906   & 0.98  & 0.901   & 0.93  & 0.847   & 0.94  & 0.840    & 0.95  & 0.846   & 0.97  & 0.462   \\
		& X      & 0.10   & 0.810    & 0.09  & 0.803   & 0.96  & 0.750    & 0.93  & 0.743   & 0.95  & 0.729   & 0.96  & 0.472   \\
		& SE     & 0.13  & 0.833   & 0.12  & 0.822   & 0.95  & 0.748   & 0.95  & 0.751   & 0.94  & 0.761   & 0.93  & 0.466   \\\hline
		\multirow{5}{*}{4096} & DINA   & 0.99  & 0.733   & 0.94  & 0.754   & 0.91  & 0.669   & 0.96  & 0.659   & 0.92  & 0.662   & 0.93  & 0.389   \\
		& E     & 0.99  & 0.733   & 0.94  & 0.754   & 0.91  & 0.669   & 0.96  & 0.659   & 0.92  & 0.662   & 0.93  & 0.389    \\
		& PA-X   & 0.96  & 0.749   & 0.96  & 0.739   & 0.94  & 0.705   & 0.99  & 0.685   & 0.95  & 0.707   & 0.98  & 0.380    \\
		& X      & 0.01  & 0.685   & 0.04  & 0.682   & 0.92  & 0.629   & 0.97  & 0.642   & 0.97  & 0.627   & 0.97  & 0.397   \\
		& SE     & 0.05  & 0.675   & 0.06  & 0.669   & 0.96  & 0.630    & 0.96  & 0.625   & 0.93  & 0.632   & 0.90   & 0.395   \\\hline
		\multirow{5}{*}{5792} & DINA   & 0.95  & 0.618   & 0.95  & 0.621   & 0.93  & 0.559   & 0.96  & 0.556   & 0.96  & 0.561   & 0.97  & 0.325   \\
		& E       & 0.95  & 0.618   & 0.95  & 0.621   & 0.93  & 0.559   & 0.96  & 0.556   & 0.96  & 0.561   & 0.97  & 0.325   \\
		& PA-X   & 0.97  & 0.627   & 0.98  & 0.627   & 0.97  & 0.595   & 0.90   & 0.575   & 0.96  & 0.597   & 0.94  & 0.321   \\
		& X      & 0.01  & 0.563   & 0.00     & 0.566   & 0.91  & 0.524   & 0.93  & 0.524   & 0.98  & 0.535   & 0.93  & 0.334   \\
		& SE     & 0.01  & 0.573   & 0.00     & 0.557   & 0.88  & 0.519   & 0.93  & 0.528   & 0.95  & 0.521   & 0.95  & 0.332   \\
	\end{tabular}
}
\end{table}

\begin{table}
		\caption{\label{tab:binomialGLMCI}  {Coverage (cvrg) and width of $95\%$ confidence intervals for binary responses. We consider the five meta-algorithms and vary the sample size from $1024$ to $5792$. We estimate the propensity score by logistic regression and baseline natural parameter functions by logistic regression. Confidence intervals are constructed using $100$ bootstrap samples. Results are averaged over $100$ trials.}}
	\centering
	\scriptsize
\fbox{%
	\begin{tabular}{c|c|cc|cc|cc|cc|cc|cc}
		\hline \hline
		Sample                    & Meth- & \multicolumn{2}{c|}{$\beta_1$} & \multicolumn{2}{c|}{$\beta_2$} & \multicolumn{2}{c|}{$\beta_3$} & \multicolumn{2}{c|}{$\beta_4$} & \multicolumn{2}{c|}{$\beta_5$} & \multicolumn{2}{c}{$\beta_0$} \\ 
		size &    od    & cvrg      & width     & cvrg      & width             & cvrg      & width           & cvrg      & width      & cvrg      & width     & cvrg      & width          \\\hline			
		\multirow{5}{*}{1024} & DINA   & 0.93  & 1.67    & 0.94  & 1.65    & 0.93  & 1.59    & 0.97  & 1.52    & 0.98  & 1.51    & 0.97  & 0.875   \\
		& E      & 0.91  & 1.62    & 0.95  & 1.65    & 0.92  & 1.55    & 0.97  & 1.48    & 0.96  & 1.48    & 0.94  & 0.891   \\
		& PA-X   & 0.93  & 2.24    & 0.97  & 2.07    & 0.97  & 1.90     & 0.93  & 1.86    & 0.96  & 1.82    & 0.99  & 1.08    \\
		& X      & 0.45  & 1.49    & 0.85  & 1.57    & 0.96  & 1.46    & 0.93  & 1.42    & 0.99  & 1.41    & 0.95  & 0.930    \\
		& SE     & 0.47  & 1.47    & 0.76  & 1.53    & 0.97  & 1.42    & 0.96  & 1.40     & 0.97  & 1.39    & 0.96  & 0.931   \\\hline
		\multirow{5}{*}{1448} & DINA   & 0.95  & 1.37    & 0.96  & 1.38    & 0.97  & 1.30     & 0.96  & 1.25    & 0.96  & 1.26    & 0.89  & 0.727   \\
		& E      & 0.88  & 1.31    & 0.92  & 1.36    & 0.98  & 1.25    & 0.97  & 1.23    & 0.97  & 1.24    & 0.89  & 0.715   \\
		& PA-X   & 1.00     & 1.75    & 0.96  & 1.59    & 0.96  & 1.54    & 0.96  & 1.46    & 0.98  & 1.49    & 0.98  & 0.887   \\
		& X      & 0.26  & 1.23    & 0.80   & 1.27    & 0.92  & 1.19    & 0.96  & 1.15    & 0.97  & 1.16    & 0.93  & 0.767   \\
		& SE     & 0.30   & 1.22    & 0.79  & 1.30     & 0.92  & 1.20     & 0.96  & 1.16    & 0.92  & 1.16    & 0.90   & 0.760    \\\hline
		\multirow{5}{*}{2048} & DINA   & 0.95  & 1.14    & 0.95  & 1.14    & 0.95  & 1.11    & 0.95  & 1.04    & 0.95  & 1.04    & 0.86  & 0.619   \\
		& E      & 0.86  & 1.08    & 0.85  & 1.12    & 0.98  & 1.05    & 0.95  & 1.03    & 0.94  & 1.03    & 0.82  & 0.599   \\
		& PA-X   & 0.95  & 1.41    & 0.92  & 1.27    & 0.96  & 1.21    & 0.93  & 1.21    & 0.97  & 1.20     & 0.95  & 0.693   \\
		& X      & 0.14  & 1.02    & 0.78  & 1.06    & 0.95  & 1.00       & 0.96  & 0.950    & 0.97  & 0.943   & 0.80   & 0.616   \\
		& SE     & 0.18  & 1.03    & 0.71  & 1.06    & 0.95  & 0.993   & 0.97  & 0.953   & 0.98  & 0.961   & 0.85  & 0.621   \\\hline
		\multirow{5}{*}{2896} & DINA   & 0.95  & 0.941   & 0.87  & 0.921   & 0.93  & 0.902   & 0.97  & 0.872   & 0.95  & 0.867   & 0.95  & 0.503   \\
		& E      & 0.82  & 0.896   & 0.79  & 0.931   & 0.90   & 0.866   & 0.96  & 0.832   & 0.95  & 0.857   & 0.84  & 0.494   \\
		& PA-X   & 0.92  & 1.10     & 0.93  & 1.03    & 0.93  & 0.990    & 0.94  & 0.973   & 0.94  & 0.980    & 0.94  & 0.565   \\
		& X      & 0.06  & 0.845   & 0.54  & 0.881   & 0.92  & 0.814   & 0.94  & 0.795   & 0.95  & 0.797   & 0.84  & 0.524   \\
		& SE     & 0.05  & 0.847   & 0.62  & 0.901   & 0.92  & 0.836   & 0.95  & 0.801   & 0.97  & 0.803   & 0.89  & 0.523   \\\hline
		\multirow{5}{*}{4096} & DINA   & 0.92  & 0.77    & 0.93  & 0.777   & 0.96  & 0.728   & 0.95  & 0.708   & 0.97  & 0.715   & 0.89  & 0.424   \\
		& E      & 0.77  & 0.762   & 0.8   & 0.764   & 0.96  & 0.727   & 0.96  & 0.705   & 0.95  & 0.703   & 0.66  & 0.410    \\
		& PA-X   & 0.95  & 0.921   & 0.84  & 0.849   & 0.92  & 0.827   & 0.94  & 0.807   & 0.94  & 0.813   & 0.95  & 0.457   \\
		& X      & 0.00     & 0.705   & 0.54  & 0.734   & 0.92  & 0.692   & 0.97  & 0.674   & 0.97  & 0.668   & 0.80   & 0.437   \\
		& SE     & 0.01  & 0.694   & 0.52  & 0.732   & 0.92  & 0.691   & 0.94  & 0.664   & 0.97  & 0.671   & 0.79  & 0.428   \\\hline
		\multirow{5}{*}{5792} & DINA   & 0.92  & 0.647   & 0.85  & 0.658   & 0.96  & 0.638   & 0.97  & 0.608   & 0.97  & 0.600     & 0.82  & 0.358   \\
		& E      & 0.73  & 0.633   & 0.66  & 0.648   & 0.96  & 0.611   & 0.97  & 0.582   & 0.94  & 0.589   & 0.69  & 0.351   \\
		& PA-X   & 0.91  & 0.768   & 0.87  & 0.713   & 0.98  & 0.701   & 0.99  & 0.674   & 0.91  & 0.659   & 0.89  & 0.381   \\
		& X      & 0.00     & 0.606   & 0.41  & 0.624   & 0.97  & 0.565   & 0.93  & 0.564   & 0.97  & 0.568   & 0.65  & 0.368   \\
		& SE     & 0.00     & 0.594   & 0.39  & 0.624   & 0.92  & 0.575   & 0.97  & 0.558   & 0.95  & 0.552   & 0.68  & 0.360   \\\hline
	\end{tabular}
}
\end{table}

\begin{table}
		\caption{\label{tab:coxGLMCI}  {Coverage (cvrg) and width of $95\%$ confidence intervals for survival responses. We consider the five meta-algorithms and vary the sample size from $1024$ to $5792$. We estimate the propensity score by logistic regression and baseline natural parameter functions by Cox regression. Confidence intervals are constructed using $100$ bootstrap samples. Results are averaged over $100$ trials.}}
	\centering
	\scriptsize
\fbox{%
	\begin{tabular}{c|c|cc|cc|cc|cc|cc|cc}
		Sample                    & Meth- & \multicolumn{2}{c|}{$\beta_1$} & \multicolumn{2}{c|}{$\beta_2$} & \multicolumn{2}{c|}{$\beta_3$} & \multicolumn{2}{c|}{$\beta_4$} & \multicolumn{2}{c|}{$\beta_5$} & \multicolumn{2}{c}{$\beta_0$} \\ 
		size &    od    & cvrg      & width     & cvrg      & width             & cvrg      & width           & cvrg      & width      & cvrg      & width     & cvrg      & width          \\\hline			
		\multirow{5}{*}{1024} & DINA   & 0.92  & 1.84    & 0.93  & 1.67    & 0.96  & 1.61    & 0.97  & 1.55    & 0.97  & 1.56    & 0.95  & 1.09    \\
		& E      & 0.87  & 1.91    & 0.91  & 1.66    & 0.96  & 1.58    & 0.96  & 1.54    & 0.92  & 1.51    & 0.92  & 1.11    \\
		& PA-X   & 0.95  & 3.38    & 0.98  & 3.23    & 0.97  & 2.14    & 1.00     & 2.01    & 0.97  & 2.04    & 0.97  & 1.20     \\
		& X      & 0.97  & 1.62    & 0.75  & 1.53    & 0.93  & 1.40     & 0.96  & 1.41    & 0.97  & 1.38    & 0.95  & 0.971   \\
		& SE     & 0.95  & 1.55    & 0.67  & 1.49    & 0.91  & 1.40     & 0.96  & 1.39    & 0.97  & 1.39    & 0.86  & 0.868   \\\hline		
		\multirow{5}{*}{1448} & DINA   & 0.92  & 1.52    & 0.88  & 1.33    & 0.93  & 1.29    & 0.97  & 1.29    & 0.94  & 1.27    & 0.90   & 0.885   \\
		& E      & 0.77  & 1.55    & 0.94  & 1.36    & 0.96  & 1.27    & 0.96  & 1.25    & 0.93  & 1.27    & 0.88  & 0.879   \\
		& PA-X   & 0.94  & 2.35    & 0.97  & 2.25    & 0.99  & 1.55    & 0.95  & 1.51    & 0.98  & 1.53    & 0.93  & 0.865   \\
		& X      & 0.88  & 1.30     & 0.43  & 1.25    & 0.94  & 1.15    & 0.94  & 1.15    & 0.91  & 1.11    & 0.94  & 0.770    \\
		& SE     & 0.89  & 1.29    & 0.51  & 1.26    & 0.95  & 1.15    & 0.93  & 1.15    & 0.95  & 1.15    & 0.83  & 0.718   \\\hline		
		\multirow{5}{*}{2048} & DINA   & 0.96  & 1.24    & 0.95  & 1.11    & 0.93  & 1.06    & 0.93  & 1.05    & 0.91  & 1.04    & 0.94  & 0.720    \\
		& E      & 0.73  & 1.28    & 0.94  & 1.11    & 0.89  & 1.04    & 0.95  & 1.03    & 0.91  & 1.02    & 0.96  & 0.734   \\
		& PA-X   & 0.92  & 1.77    & 0.94  & 1.68    & 0.97  & 1.20     & 0.95  & 1.19    & 0.98  & 1.18    & 0.94  & 0.672   \\
		& X      & 0.92  & 1.08    & 0.25  & 1.02    & 0.91  & 0.954   & 0.95  & 0.943   & 0.94  & 0.963   & 0.94  & 0.639   \\
		& SE     & 0.86  & 1.07    & 0.27  & 1.01    & 0.97  & 0.947   & 0.93  & 0.944   & 0.95  & 0.941   & 0.74  & 0.589   \\\hline		
		\multirow{5}{*}{2896} & DINA   & 0.93  & 1.04    & 0.88  & 0.926   & 0.93  & 0.868   & 0.95  & 0.866   & 0.94  & 0.873   & 0.90   & 0.586   \\
		& E      & 0.72  & 1.05    & 0.91  & 0.939   & 0.94  & 0.869   & 0.97  & 0.855   & 0.93  & 0.855   & 0.96  & 0.595   \\
		& PA-X   & 0.91  & 1.43    & 0.96  & 1.34    & 0.93  & 0.964   & 0.97  & 0.946   & 0.93  & 0.954   & 0.96  & 0.543   \\
		& X      & 0.93  & 0.903   & 0.19  & 0.838   & 0.94  & 0.799   & 0.95  & 0.776   & 0.94  & 0.778   & 0.96  & 0.530    \\
		& SE     & 0.88  & 0.881   & 0.19  & 0.844   & 0.95  & 0.785   & 0.90   & 0.772   & 0.95  & 0.788   & 0.56  & 0.468   \\\hline		
		\multirow{5}{*}{4096} & DINA   & 0.98  & 0.869   & 0.93  & 0.770    & 0.99  & 0.737   & 0.97  & 0.714   & 0.96  & 0.703   & 0.93  & 0.501   \\
		& E      & 0.48  & 0.872   & 0.91  & 0.777   & 0.99  & 0.724   & 0.96  & 0.715   & 0.97  & 0.717   & 0.93  & 0.486   \\
		& PA-X   & 0.92  & 1.14    & 0.94  & 1.06    & 0.98  & 0.805   & 0.97  & 0.786   & 0.98  & 0.772   & 0.92  & 0.453   \\
		& X      & 0.90   & 0.758   & 0.05  & 0.711   & 1.00     & 0.667   & 0.97  & 0.677   & 0.98  & 0.663   & 0.93  & 0.439   \\
		& SE     & 0.84  & 0.742   & 0.04  & 0.706   & 0.90   & 0.661   & 0.97  & 0.654   & 0.96  & 0.652   & 0.51  & 0.400     \\\hline		
		\multirow{5}{*}{5792} & DINA   & 0.93  & 0.715   & 0.93  & 0.633   & 0.93  & 0.605   & 0.94  & 0.594   & 0.92  & 0.595   & 0.93  & 0.410    \\
		& E      & 0.39  & 0.720    & 0.96  & 0.646   & 0.93  & 0.600     & 0.95  & 0.585   & 0.91  & 0.584   & 0.91  & 0.408   \\
		& PA-X   & 0.91  & 0.904   & 0.93  & 0.868   & 0.95  & 0.638   & 0.97  & 0.629   & 0.96  & 0.650    & 0.97  & 0.372   \\
		& X      & 0.85  & 0.623   & 0.02  & 0.590    & 0.94  & 0.546   & 0.95  & 0.541   & 0.98  & 0.546   & 0.99  & 0.369   \\
		& SE     & 0.82  & 0.619   & 0.00     & 0.590    & 0.97  & 0.549   & 0.94  & 0.542   & 0.94  & 0.541   & 0.30   & 0.327  \\
	\end{tabular}
}
\end{table}

\begin{table}
		\caption{\label{tab:gaussianBoostingCI}  {Coverage (cvrg) and width of $95\%$ confidence intervals for Gaussian responses.  We consider the five meta-algorithms and vary the sample size from $1024$ to $5792$. For Gaussian responses, the proposed method (DINA) and R-learner (E) are the same. We estimate the propensity score by logistic regression and baseline natural parameter functions by gradient boosting. Confidence intervals are constructed using $100$ bootstrap samples. Results are averaged over $100$ trials.}}
	\centering
	\scriptsize
	\fbox{%
	\begin{tabular}{c|c|cc|cc|cc|cc|cc|cc}
		Sample                    & Meth- & \multicolumn{2}{c|}{$\beta_1$} & \multicolumn{2}{c|}{$\beta_2$} & \multicolumn{2}{c|}{$\beta_3$} & \multicolumn{2}{c|}{$\beta_4$} & \multicolumn{2}{c|}{$\beta_5$} & \multicolumn{2}{c}{$\beta_0$} \\ 
		size &    od    & cvrg      & width     & cvrg      & width             & cvrg      & width           & cvrg      & width      & cvrg      & width     & cvrg      & width          \\\hline			
		\multirow{5}{*}{1024} & DINA   & 0.96  & 1.39    & 0.97  & 1.40     & 0.94  & 1.29    & 0.97  & 1.30     & 0.89  & 1.30     & 0.92  & 0.746   \\
		& E     & 0.96  & 1.39    & 0.97  & 1.40     & 0.94  & 1.29    & 0.97  & 1.30     & 0.89  & 1.30     & 0.92  & 0.746   \\
		& PA-X   & 0.93  & 1.43    & 0.86  & 1.35    & 0.92  & 1.26    & 0.99  & 1.21    & 0.95  & 1.20     & 0.95  & 0.787   \\
		& X      & 0.87  & 1.34    & 0.75  & 1.40     & 0.95  & 1.27    & 0.96  & 1.18    & 0.94  & 1.22    & 0.94  & 0.784   \\
		& SE     & 0.80   & 1.28    & 0.92  & 1.27    & 0.92  & 1.20     & 0.93  & 1.07    & 1.00     & 1.06    & 0.95  & 0.778   \\\hline		
		\multirow{5}{*}{1448} & DINA   & 0.99  & 1.16    & 0.91  & 1.14    & 0.93  & 1.09    & 0.94  & 1.11    & 0.88  & 1.11    & 0.94  & 0.617   \\
		& E     & 0.99  & 1.16    & 0.91  & 1.14    & 0.93  & 1.09    & 0.94  & 1.11    & 0.88  & 1.11    & 0.94  & 0.617   \\
		& PA-X   & 0.97  & 1.18    & 0.92  & 1.14    & 0.87  & 1.06    & 0.93  & 0.998   & 0.91  & 0.996   & 0.97  & 0.647   \\
		& X      & 0.89  & 1.12    & 0.81  & 1.19    & 0.93  & 1.03    & 0.96  & 1.01    & 0.91  & 1.02    & 0.93  & 0.650    \\
		& SE     & 0.85  & 1.04    & 0.89  & 1.05    & 0.91  & 0.968   & 0.96  & 0.88    & 0.91  & 0.887   & 0.92  & 0.650    \\\hline		
		\multirow{5}{*}{2048} & DINA   & 0.93  & 0.956   & 0.95  & 0.954   & 0.94  & 0.933   & 0.91  & 0.917   & 0.93  & 0.919   & 0.97  & 0.530    \\
		& E      & 0.93  & 0.956   & 0.95  & 0.954   & 0.94  & 0.933   & 0.91  & 0.917   & 0.93  & 0.919   & 0.97  & 0.530    \\
		& PA-X   & 0.96  & 0.992   & 0.89  & 0.965   & 0.91  & 0.900     & 0.90   & 0.837   & 0.94  & 0.821   & 0.95  & 0.546   \\
		& X      & 0.82  & 0.927   & 0.74  & 0.970    & 0.90   & 0.888   & 0.93  & 0.84    & 0.97  & 0.821   & 0.96  & 0.547   \\
		& SE     & 0.77  & 0.889   & 0.89  & 0.889   & 0.88  & 0.858   & 0.94  & 0.716   & 0.94  & 0.727   & 0.89  & 0.531   \\\hline		
		\multirow{5}{*}{2896} & DINA   & 0.92  & 0.797   & 0.87  & 0.819   & 0.94  & 0.765   & 0.89  & 0.758   & 0.95  & 0.765   & 0.94  & 0.446   \\
		& E     & 0.92  & 0.797   & 0.87  & 0.819   & 0.94  & 0.765   & 0.89  & 0.758   & 0.95  & 0.765   & 0.94  & 0.446   \\
		& PA-X   & 0.91  & 0.831   & 0.85  & 0.808   & 0.88  & 0.755   & 0.90   & 0.687   & 0.93  & 0.688   & 0.93  & 0.450    \\
		& X      & 0.79  & 0.764   & 0.65  & 0.833   & 0.89  & 0.736   & 0.91  & 0.692   & 0.92  & 0.695   & 0.94  & 0.454   \\
		& SE     & 0.83  & 0.745   & 0.90   & 0.740    & 0.86  & 0.697   & 0.96  & 0.590    & 0.92  & 0.602   & 0.94  & 0.437   \\\hline		
		\multirow{5}{*}{4096} & DINA   & 0.95  & 0.682   & 0.86  & 0.686   & 0.96  & 0.638   & 0.95  & 0.648   & 0.98  & 0.639   & 0.94  & 0.367   \\
		& E       & 0.95  & 0.682   & 0.86  & 0.686   & 0.96  & 0.638   & 0.95  & 0.648   & 0.98  & 0.639   & 0.94  & 0.367   \\
		& PA-X   & 0.93  & 0.706   & 0.88  & 0.688   & 0.84  & 0.637   & 0.94  & 0.584   & 0.96  & 0.573   & 0.93  & 0.375   \\
		& X      & 0.83  & 0.649   & 0.72  & 0.711   & 0.88  & 0.624   & 0.95  & 0.581   & 0.96  & 0.576   & 0.92  & 0.382   \\
		& SE     & 0.67  & 0.629   & 0.90   & 0.628   & 0.91  & 0.597   & 0.96  & 0.496   & 0.95  & 0.495   & 0.95  & 0.373   \\\hline		
		\multirow{5}{*}{5792} & DINA   & 0.95  & 0.562   & 0.93  & 0.572   & 0.94  & 0.550    & 0.96  & 0.550    & 0.89  & 0.541   & 0.91  & 0.307   \\
		& E   & 0.95  & 0.562   & 0.93  & 0.572   & 0.94  & 0.550    & 0.96  & 0.550    & 0.89  & 0.541   & 0.91  & 0.307   \\
		& PA-X   & 0.89  & 0.589   & 0.88  & 0.576   & 0.87  & 0.538   & 0.96  & 0.476   & 0.92  & 0.476   & 0.95  & 0.311   \\
		& X      & 0.79  & 0.55    & 0.62  & 0.582   & 0.91  & 0.516   & 0.97  & 0.477   & 0.91  & 0.479   & 0.92  & 0.316   \\
		& SE     & 0.61  & 0.521   & 0.92  & 0.530    & 0.91  & 0.516   & 0.97  & 0.401   & 0.96  & 0.410    & 0.89  & 0.312  \\
	\end{tabular}
}
\end{table}

\begin{table}
		\caption{\label{tab:binomialBoostingCI}  {Coverage (cvrg) and width of $95\%$ confidence intervals for binary responses.  We consider the five meta-algorithms and vary the sample size from $1024$ to $5792$. We estimate the propensity score by logistic regression and baseline natural parameter functions by gradient boosting. Confidence intervals are constructed using $100$ bootstrap samples. Results are averaged over $100$ trials.}}
	\centering
	\scriptsize
\fbox{%
	\begin{tabular}{c|c|cc|cc|cc|cc|cc|cc}
		Sample                    & Meth- & \multicolumn{2}{c|}{$\beta_1$} & \multicolumn{2}{c|}{$\beta_2$} & \multicolumn{2}{c|}{$\beta_3$} & \multicolumn{2}{c|}{$\beta_4$} & \multicolumn{2}{c|}{$\beta_5$} & \multicolumn{2}{c}{$\beta_0$} \\ 
		size &    od    & cvrg      & width     & cvrg      & width             & cvrg      & width           & cvrg      & width      & cvrg      & width     & cvrg      & width          \\\hline			
		\multirow{5}{*}{1024} & DINA   & 0.93  & 1.59    & 0.92  & 1.63    & 0.96  & 1.53    & 0.91  & 1.53    & 0.94  & 1.52    & 0.97  & 0.897   \\
		& E      & 0.91  & 1.65    & 0.96  & 1.70     & 0.96  & 1.61    & 0.92  & 1.58    & 0.91  & 1.55    & 0.95  & 0.913   \\
		& PA-X   & 0.91  & 1.80     & 0.91  & 1.86    & 0.93  & 1.70     & 0.95  & 1.58    & 0.97  & 1.59    & 0.92  & 1.10     \\
		& X      & 0.92  & 1.63    & 0.88  & 1.79    & 0.96  & 1.63    & 0.96  & 1.53    & 0.97  & 1.52    & 0.94  & 1.04    \\
		& SE     & 0.82  & 1.50     & 0.93  & 1.52    & 0.94  & 1.45    & 0.96  & 1.28    & 0.97  & 1.27    & 0.94  & 0.934   \\\hline		
		\multirow{5}{*}{1448} & DINA   & 0.89  & 1.30     & 0.95  & 1.37    & 0.90   & 1.24    & 0.93  & 1.22    & 0.93  & 1.22    & 0.92  & 0.728   \\
		& E      & 0.90   & 1.33    & 0.91  & 1.38    & 0.89  & 1.31    & 0.92  & 1.26    & 0.96  & 1.28    & 0.91  & 0.752   \\
		& PA-X   & 0.95  & 1.50     & 0.97  & 1.50     & 0.89  & 1.36    & 0.94  & 1.27    & 0.98  & 1.27    & 0.94  & 0.875   \\
		& X      & 0.90   & 1.33    & 0.96  & 1.48    & 0.88  & 1.33    & 0.95  & 1.22    & 0.96  & 1.23    & 0.95  & 0.825   \\
		& SE     & 0.74  & 1.21    & 0.97  & 1.27    & 0.91  & 1.20     & 0.96  & 1.04    & 0.96  & 1.04    & 0.89  & 0.776   \\\hline		
		\multirow{5}{*}{2048} & DINA   & 0.91  & 1.09    & 0.96  & 1.11    & 0.98  & 1.03    & 0.92  & 1.03    & 0.92  & 1.03    & 0.97  & 0.607   \\
		& E      & 0.89  & 1.10     & 0.88  & 1.11    & 0.98  & 1.08    & 0.93  & 1.04    & 0.91  & 1.04    & 0.97  & 0.608   \\
		& PA-X   & 0.97  & 1.21    & 0.88  & 1.22    & 0.87  & 1.14    & 0.94  & 1.05    & 0.92  & 1.04    & 0.99  & 0.708   \\
		& X      & 0.80   & 1.07    & 0.84  & 1.20     & 0.91  & 1.07    & 0.93  & 1.00       & 0.97  & 1.00       & 0.99  & 0.676   \\
		& SE     & 0.55  & 0.986   & 0.98  & 1.04    & 0.92  & 0.990    & 0.94  & 0.848   & 0.92  & 0.847   & 0.86  & 0.622   \\\hline		
		\multirow{5}{*}{2896} & DINA   & 0.97  & 0.913   & 0.94  & 0.925   & 0.95  & 0.866   & 0.89  & 0.859   & 0.93  & 0.859   & 0.88  & 0.506   \\
		& E      & 0.94  & 0.893   & 0.88  & 0.926   & 0.95  & 0.890    & 0.89  & 0.868   & 0.93  & 0.856   & 0.84  & 0.514   \\
		& PA-X   & 0.96  & 1.01    & 0.88  & 1.00       & 0.86  & 0.933   & 0.96  & 0.846   & 0.95  & 0.844   & 0.96  & 0.583   \\
		& X      & 0.77  & 0.907   & 0.86  & 1.00       & 0.90   & 0.877   & 0.97  & 0.840    & 0.97  & 0.817   & 0.96  & 0.567   \\
		& SE     & 0.50   & 0.842   & 0.95  & 0.867   & 0.93  & 0.811   & 0.97  & 0.678   & 0.90   & 0.682   & 0.86  & 0.522   \\\hline		
		\multirow{5}{*}{4096} & DINA   & 0.91  & 0.753   & 0.95  & 0.773   & 0.92  & 0.732   & 0.95  & 0.722   & 0.91  & 0.702   & 0.93  & 0.420    \\
		& E      & 0.85  & 0.757   & 0.86  & 0.761   & 0.93  & 0.734   & 0.92  & 0.711   & 0.96  & 0.730    & 0.88  & 0.422   \\
		& PA-X   & 0.96  & 0.841   & 0.91  & 0.826   & 0.84  & 0.765   & 0.92  & 0.676   & 0.96  & 0.689   & 0.96  & 0.464   \\
		& X      & 0.75  & 0.750    & 0.85  & 0.828   & 0.87  & 0.740    & 0.93  & 0.666   & 0.96  & 0.672   & 0.95  & 0.459   \\
		& SE     & 0.43  & 0.699   & 0.93  & 0.720    & 0.90   & 0.668   & 0.95  & 0.557   & 0.97  & 0.557   & 0.82  & 0.423   \\\hline		
		\multirow{5}{*}{5792} & DINA   & 0.94  & 0.635   & 0.97  & 0.646   & 0.97  & 0.622   & 0.96  & 0.610    & 0.94  & 0.607   & 0.91  & 0.363   \\
		& E      & 0.89  & 0.633   & 0.83  & 0.636   & 0.91  & 0.618   & 0.97  & 0.607   & 0.93  & 0.597   & 0.85  & 0.357   \\
		& PA-X   & 0.95  & 0.699   & 0.86  & 0.715   & 0.92  & 0.664   & 0.92  & 0.579   & 0.97  & 0.569   & 0.95  & 0.396   \\
		& X      & 0.61  & 0.619   & 0.81  & 0.686   & 0.96  & 0.615   & 0.91  & 0.544   & 0.94  & 0.550    & 0.94  & 0.379   \\
		& SE     & 0.33  & 0.582   & 0.95  & 0.619   & 0.84  & 0.569   & 0.96  & 0.450    & 0.91  & 0.444   & 0.78  & 0.359  \\
	\end{tabular}
}
\end{table}

\begin{table}
		\caption{\label{tab:poissonBoostingCI}  {Coverage (cvrg) and width of $95\%$ confidence intervals for count data.  We consider the five meta-algorithms and vary the sample size from $1024$ to $5792$. We estimate the propensity score by logistic regression and baseline natural parameter functions by gradient boosting. Confidence intervals are constructed using $100$ bootstrap samples. Results are averaged over $100$ trials.}}
	\centering
	\scriptsize
\fbox{%
	\begin{tabular}{c|c|cc|cc|cc|cc|cc|cc}
		Sample                    & Meth- & \multicolumn{2}{c|}{$\beta_1$} & \multicolumn{2}{c|}{$\beta_2$} & \multicolumn{2}{c|}{$\beta_3$} & \multicolumn{2}{c|}{$\beta_4$} & \multicolumn{2}{c|}{$\beta_5$} & \multicolumn{2}{c}{$\beta_0$} \\ 
		size &    od    & cvrg      & width     & cvrg      & width             & cvrg      & width           & cvrg      & width      & cvrg      & width     & cvrg      & width          \\\hline			
		\multirow{5}{*}{1024} & DINA   & 0.95  & 0.712   & 0.91  & 0.769   & 0.91  & 0.694   & 0.95  & 0.670    & 0.91  & 0.676   & 0.92  & 0.392   \\
		& E      & 0.95  & 0.719   & 0.89  & 0.766   & 0.89  & 0.701   & 0.95  & 0.692   & 0.95  & 0.694   & 0.94  & 0.410    \\
		& PA-X   & 0.81  & 0.701   & 0.93  & 0.733   & 0.86  & 0.650    & 0.96  & 0.568   & 0.96  & 0.559   & 0.97  & 0.396   \\
		& X      & 0.84  & 0.666   & 0.92  & 0.722   & 0.90   & 0.648   & 0.97  & 0.569   & 0.95  & 0.568   & 0.95  & 0.384   \\
		& SE     & 0.82  & 0.594   & 0.82  & 0.637   & 0.77  & 0.576   & 0.97  & 0.487   & 0.94  & 0.493   & 0.92  & 0.387   \\\hline			
		\multirow{5}{*}{1448} & DINA   & 0.97  & 0.580    & 0.94  & 0.632   & 0.90   & 0.565   & 0.95  & 0.552   & 0.91  & 0.562   & 0.92  & 0.331   \\
		& E      & 0.98  & 0.608   & 0.96  & 0.628   & 0.95  & 0.583   & 0.95  & 0.569   & 0.92  & 0.551   & 0.93  & 0.339   \\
		& PA-X   & 0.89  & 0.581   & 0.86  & 0.601   & 0.75  & 0.547   & 0.93  & 0.462   & 0.98  & 0.462   & 0.92  & 0.332   \\
		& X      & 0.91  & 0.555   & 0.86  & 0.600     & 0.77  & 0.533   & 0.95  & 0.466   & 0.97  & 0.470    & 0.91  & 0.323   \\
		& SE     & 0.76  & 0.492   & 0.77  & 0.541   & 0.71  & 0.476   & 0.96  & 0.404   & 0.97  & 0.399   & 0.92  & 0.317   \\\hline			
		\multirow{5}{*}{2048} & DINA   & 0.92  & 0.484   & 0.95  & 0.523   & 0.94  & 0.485   & 0.93  & 0.476   & 0.95  & 0.480    & 0.91  & 0.271   \\
		& E      & 0.94  & 0.492   & 0.92  & 0.525   & 0.94  & 0.471   & 0.95  & 0.468   & 0.95  & 0.468   & 0.92  & 0.279   \\
		& PA-X   & 0.86  & 0.487   & 0.83  & 0.499   & 0.78  & 0.458   & 0.94  & 0.378   & 0.96  & 0.378   & 0.92  & 0.268   \\
		& X      & 0.89  & 0.462   & 0.84  & 0.502   & 0.87  & 0.462   & 0.93  & 0.384   & 0.98  & 0.387   & 0.91  & 0.275   \\
		& SE     & 0.69  & 0.413   & 0.77  & 0.445   & 0.73  & 0.396   & 0.95  & 0.325   & 0.92  & 0.329   & 0.96  & 0.265   \\\hline			
		\multirow{5}{*}{2896} & DINA   & 0.95  & 0.405   & 0.93  & 0.445   & 0.89  & 0.402   & 0.91  & 0.400     & 0.94  & 0.395   & 0.94  & 0.233   \\
		& E      & 0.96  & 0.407   & 0.92  & 0.437   & 0.92  & 0.400     & 0.92  & 0.390    & 0.96  & 0.387   & 0.92  & 0.239   \\
		& PA-X   & 0.79  & 0.403   & 0.80   & 0.425   & 0.81  & 0.395   & 0.94  & 0.320    & 0.96  & 0.314   & 0.95  & 0.224   \\
		& X      & 0.78  & 0.391   & 0.82  & 0.415   & 0.84  & 0.389   & 0.96  & 0.315   & 0.95  & 0.320    & 0.94  & 0.223   \\
		& SE     & 0.59  & 0.352   & 0.79  & 0.372   & 0.79  & 0.329   & 0.96  & 0.266   & 0.94  & 0.256   & 0.91  & 0.216   \\\hline			
		\multirow{5}{*}{4096} & DINA   & 0.97  & 0.344   & 0.93  & 0.37    & 0.94  & 0.342   & 0.94  & 0.338   & 0.89  & 0.329   & 0.93  & 0.195   \\
		& E      & 0.98  & 0.344   & 0.89  & 0.362   & 0.94  & 0.334   & 0.96  & 0.320    & 0.91  & 0.324   & 0.95  & 0.198   \\
		& PA-X   & 0.86  & 0.341   & 0.81  & 0.346   & 0.77  & 0.329   & 0.95  & 0.260    & 0.94  & 0.261   & 0.93  & 0.186   \\
		& X      & 0.89  & 0.331   & 0.86  & 0.357   & 0.81  & 0.324   & 0.93  & 0.263   & 0.96  & 0.261   & 0.91  & 0.187   \\
		& SE     & 0.51  & 0.301   & 0.70   & 0.304   & 0.63  & 0.281   & 0.94  & 0.213   & 0.94  & 0.213   & 0.91  & 0.182   \\\hline			
		\multirow{5}{*}{5792} & DINA   & 0.94  & 0.295   & 0.93  & 0.313   & 0.98  & 0.290    & 0.90   & 0.218   & 0.95  & 0.237   & 0.94  & 0.166   \\
		& E      & 0.96  & 0.287   & 0.91  & 0.308   & 0.95  & 0.280    & 0.92  & 0.212   & 0.94  & 0.231   & 0.89  & 0.162   \\
		& PA-X   & 0.78  & 0.291   & 0.79  & 0.299   & 0.76  & 0.278   & 0.96  & 0.216   & 0.92  & 0.214   & 0.95  & 0.157   \\
		& X      & 0.81  & 0.278   & 0.83  & 0.297   & 0.79  & 0.278   & 0.95  & 0.216   & 0.94  & 0.221   & 0.92  & 0.157   \\
		& SE     & 0.48  & 0.253   & 0.46  & 0.258   & 0.53  & 0.231   & 0.93  & 0.169   & 0.97  & 0.173   & 0.93  & 0.155  \\
	\end{tabular}
}
\end{table}

\begin{table}
	\caption{\label{tab:coxBoostingCI}	  {Coverage (cvrg) and width of $95\%$ confidence intervals for survival responses.  We consider the five meta-algorithms and vary the sample size from $1024$ to $5792$. We estimate the propensity score by logistic regression and baseline natural parameter functions by gradient boosting. Confidence intervals are constructed using $100$ bootstrap samples. Results are averaged over $100$ trials.}}
\centering
	\scriptsize
\fbox{%
		\begin{tabular}{c|c|cc|cc|cc|cc|cc|cc}
		Sample                    & Meth- & \multicolumn{2}{c|}{$\beta_1$} & \multicolumn{2}{c|}{$\beta_2$} & \multicolumn{2}{c|}{$\beta_3$} & \multicolumn{2}{c|}{$\beta_4$} & \multicolumn{2}{c|}{$\beta_5$} & \multicolumn{2}{c}{$\beta_0$} \\ 
		size &    od    & cvrg      & width     & cvrg      & width             & cvrg      & width           & cvrg      & width      & cvrg      & width     & cvrg      & width          \\\hline			
		\multirow{5}{*}{1024} & DINA   & 0.95  & 1.86    & 0.95  & 1.74    & 0.94  & 1.69    & 0.96  & 1.63    & 0.99  & 1.63    & 0.95  & 1.05    \\
		& E      & 0.90   & 1.86    & 0.95  & 1.74    & 0.97  & 1.66    & 0.93  & 1.62    & 0.94  & 1.63    & 0.96  & 1.06    \\
		& PA-X   & 0.96  & 1.89    & 0.95  & 1.91    & 0.99  & 1.67    & 0.97  & 1.58    & 0.98  & 1.58    & 0.97  & 1.16    \\
		& X      & 0.91  & 1.72    & 0.87  & 1.62    & 0.97  & 1.61    & 0.98  & 1.57    & 0.97  & 1.58    & 0.94  & 1.12    \\
		& SE     & 0.88  & 1.49    & 0.90   & 1.43    & 0.97  & 1.43    & 0.97  & 1.38    & 0.98  & 1.35    & 0.88  & 0.894   \\\hline
		\multirow{5}{*}{1448} & DINA   & 0.98  & 1.47    & 0.91  & 1.40     & 0.94  & 1.35    & 0.91  & 1.31    & 0.94  & 1.33    & 0.95  & 0.837   \\
		& E      & 0.95  & 1.46    & 0.91  & 1.39    & 0.98  & 1.30     & 0.91  & 1.31    & 0.98  & 1.31    & 0.92  & 0.837   \\
		& PA-X   & 0.91  & 1.49    & 0.93  & 1.56    & 0.93  & 1.33    & 0.96  & 1.24    & 0.94  & 1.26    & 0.98  & 0.901   \\
		& X      & 0.83  & 1.38    & 0.85  & 1.32    & 0.95  & 1.31    & 0.95  & 1.26    & 0.95  & 1.25    & 0.97  & 0.859   \\
		& SE     & 0.79  & 1.19    & 0.90   & 1.16    & 0.88  & 1.18    & 0.95  & 1.12    & 0.92  & 1.11    & 0.80   & 0.717   \\\hline
		\multirow{5}{*}{2048} & DINA   & 0.91  & 1.23    & 0.95  & 1.16    & 0.97  & 1.07    & 0.98  & 1.06    & 0.95  & 1.05    & 0.91  & 0.663   \\
		& E      & 0.85  & 1.19    & 0.93  & 1.12    & 0.94  & 1.05    & 0.95  & 1.04    & 0.95  & 1.04    & 0.94  & 0.658   \\
		& PA-X   & 0.93  & 1.21    & 0.93  & 1.25    & 0.93  & 1.05    & 0.97  & 1.02    & 0.94  & 1.01    & 0.97  & 0.699   \\
		& X      & 0.78  & 1.14    & 0.77  & 1.05    & 0.96  & 1.06    & 0.99  & 1.01    & 0.93  & 0.997   & 0.97  & 0.675   \\
		& SE     & 0.71  & 1.02    & 0.86  & 0.943   & 0.92  & 0.927   & 0.90   & 0.905   & 0.94  & 0.901   & 0.70   & 0.582   \\\hline
		\multirow{5}{*}{2896} & DINA   & 0.96  & 0.975   & 0.96  & 0.943   & 0.96  & 0.907   & 0.97  & 0.883   & 0.95  & 0.887   & 0.85  & 0.541   \\
		& E      & 0.90   & 0.973   & 0.94  & 0.917   & 0.95  & 0.879   & 0.94  & 0.838   & 0.94  & 0.856   & 0.93  & 0.534   \\
		& PA-X   & 0.86  & 0.995   & 0.94  & 1.02    & 0.88  & 0.855   & 0.99  & 0.817   & 0.91  & 0.815   & 0.93  & 0.549   \\
		& X      & 0.84  & 0.913   & 0.78  & 0.861   & 0.93  & 0.862   & 0.97  & 0.807   & 0.93  & 0.803   & 0.91  & 0.541   \\
		& SE     & 0.65  & 0.812   & 0.87  & 0.777   & 0.90   & 0.760    & 0.95  & 0.712   & 0.92  & 0.719   & 0.52  & 0.478   \\\hline
		\multirow{5}{*}{4096} & DINA   & 0.95  & 0.831   & 0.88  & 0.768   & 0.97  & 0.747   & 0.95  & 0.732   & 0.88  & 0.719   & 0.91  & 0.452   \\
		& E      & 0.79  & 0.815   & 0.90   & 0.767   & 0.96  & 0.728   & 0.95  & 0.697   & 0.93  & 0.703   & 0.95  & 0.45    \\
		& PA-X   & 0.85  & 0.830    & 0.97  & 0.835   & 0.89  & 0.717   & 0.96  & 0.652   & 0.97  & 0.645   & 0.90   & 0.444   \\
		& X      & 0.71  & 0.770    & 0.78  & 0.737   & 0.92  & 0.707   & 0.95  & 0.669   & 0.93  & 0.648   & 0.92  & 0.437   \\
		& SE     & 0.57  & 0.684   & 0.81  & 0.651   & 0.85  & 0.642   & 0.95  & 0.580    & 0.98  & 0.595   & 0.48  & 0.406   \\\hline
		\multirow{5}{*}{5792} & DINA   & 0.92  & 0.696   & 0.93  & 0.632   & 0.92  & 0.599   & 0.97  & 0.583   & 0.95  & 0.598   & 0.91  & 0.382   \\
		& E      & 0.81  & 0.675   & 0.89  & 0.620    & 0.94  & 0.594   & 0.98  & 0.577   & 0.97  & 0.575   & 0.96  & 0.375   \\
		& PA-X   & 0.76  & 0.669   & 0.94  & 0.675   & 0.85  & 0.602   & 0.97  & 0.521   & 0.97  & 0.531   & 0.90   & 0.354   \\
		& X      & 0.60   & 0.625   & 0.74  & 0.585   & 0.92  & 0.573   & 0.94  & 0.530    & 0.98  & 0.536   & 0.93  & 0.359   \\
		& SE     & 0.48  & 0.569   & 0.91  & 0.529   & 0.80   & 0.518   & 0.95  & 0.478   & 0.97  & 0.483   & 0.24  & 0.331  \\
	\end{tabular}
}
\end{table}
}


\bibliographystyle{rss}
\bibliography{DINA2Ref}

\begin{thebibliography}{51}
\expandafter\ifx\csname natexlab\endcsname\relax\def\natexlab#1{#1}\fi
\expandafter\ifx\csname url\endcsname\relax
  \def\url#1{\texttt{#1}}\fi
\expandafter\ifx\csname urlprefix\endcsname\relax\def\urlprefix{URL: }\fi

\bibitem[{Andersen and Gill(1982)}]{andersen1982cox}
Andersen, P.~K. and Gill, R.~D. (1982) Cox's regression model for counting
  processes: a large-sample study.
\newblock \textit{The annals of statistics}, 1100--1120.

\bibitem[{Athey and Imbens(2015)}]{athey2015machine}
Athey, S. and Imbens, G.~W. (2015) Machine learning methods for estimating
  heterogeneous causal effects.
\newblock \textit{arXiv preprint arXiv:1504.01132}.

\bibitem[{Bennett and Lanning(2007)}]{bennett2007netflix}
Bennett, J. and Lanning, S. (2007) The netflix prize.
\newblock In \textit{Proceedings of the KDD Cup Workshop}, vol. 2007, 3--6. New
  York, NY, USA.

\bibitem[{Chernozhukov et~al.(2018)Chernozhukov, Chetverikov, Demirer, Duflo,
  Hansen, Newey and Robins}]{chernozhukov2018double}
Chernozhukov, V., Chetverikov, D., Demirer, M., Duflo, E., Hansen, C., Newey,
  W. and Robins, J. (2018) Double/debiased machine learning for treatment and
  structural parameters.

\bibitem[{Cox(1972)}]{cox1972regression}
Cox, D.~R. (1972) Regression models and life-tables.
\newblock \textit{Journal of the Royal Statistical Society: Series B
  (Methodological)}, \textbf{34}, 187--202.

\bibitem[{Daniel et~al.(2021)Daniel, Zhang and Farewell}]{daniel2021making}
Daniel, R., Zhang, J. and Farewell, D. (2021) Making apples from oranges:
  Comparing noncollapsible effect estimators and their standard errors after
  adjustment for different covariate sets.
\newblock \textit{Biometrical Journal}, \textbf{63}, 528--557.

\bibitem[{Dorie et~al.(2019)Dorie, Hill, Shalit, Scott and
  Cervone}]{dorie2019automated}
Dorie, V., Hill, J., Shalit, U., Scott, M. and Cervone, D. (2019) Automated
  versus do-it-yourself methods for causal inference: Lessons learned from a
  data analysis competition.
\newblock \textit{Statistical Science}, \textbf{34}, 43--68.

\bibitem[{Dud{\'\i}k et~al.(2011)Dud{\'\i}k, Langford and Li}]{dudik2011doubly}
Dud{\'\i}k, M., Langford, J. and Li, L. (2011) Doubly robust policy evaluation
  and learning.
\newblock \textit{Proceedings of the 28th International Conference on Machine
  Learning}.

\bibitem[{Fienberg(2007)}]{fienberg2007analysis}
Fienberg, S.~E. (2007) \textit{The analysis of cross-classified categorical
  data}.
\newblock Springer Science \& Business Media.

\bibitem[{Foster et~al.(2011)Foster, Taylor and Ruberg}]{foster2011subgroup}
Foster, J.~C., Taylor, J.~M. and Ruberg, S.~J. (2011) Subgroup identification
  from randomized clinical trial data.
\newblock \textit{Statistics in medicine}, \textbf{30}, 2867--2880.

\bibitem[{Gail et~al.(1984)Gail, Wieand and Piantadosi}]{gail1984biased}
Gail, M.~H., Wieand, S. and Piantadosi, S. (1984) Biased estimates of treatment
  effect in randomized experiments with nonlinear regressions and omitted
  covariates.
\newblock \textit{Biometrika}, \textbf{71}, 431--444.

\bibitem[{Gao and Han(2020)}]{gao2020minimax}
Gao, Z. and Han, Y. (2020) Minimax optimal nonparametric estimation of
  heterogeneous treatment effects.
\newblock \textit{arXiv preprint arXiv:2002.06471}.

\bibitem[{Glas et~al.(2003)Glas, Lijmer, Prins, Bonsel and
  Bossuyt}]{glas2003diagnostic}
Glas, A.~S., Lijmer, J.~G., Prins, M.~H., Bonsel, G.~J. and Bossuyt, P.~M.
  (2003) The diagnostic odds ratio: a single indicator of test performance.
\newblock \textit{Journal of clinical epidemiology}, \textbf{56}, 1129--1135.

\bibitem[{Greenland and Robins(1986)}]{greenland1986identifiability}
Greenland, S. and Robins, J.~M. (1986) Identifiability, exchangeability, and
  epidemiological confounding.
\newblock \textit{International journal of epidemiology}, \textbf{15},
  413--419.

\bibitem[{Greenland et~al.(1999)Greenland, Robins and
  Pearl}]{greenland1999confounding}
Greenland, S., Robins, J.~M. and Pearl, J. (1999) Confounding and
  collapsibility in causal inference.
\newblock \textit{Statistical science}, 29--46.

\bibitem[{Hahn et~al.(2020)Hahn, Murray and Carvalho}]{hahn2020bayesian}
Hahn, P.~R., Murray, J.~S. and Carvalho, C.~M. (2020) Bayesian regression tree
  models for causal inference: regularization, confounding, and heterogeneous
  effects (with discussion).
\newblock \textit{Bayesian Analysis}, \textbf{15}, 965--1056.

\bibitem[{Hansen(2008)}]{hansen2008prognostic}
Hansen, B.~B. (2008) The prognostic analogue of the propensity score.
\newblock \textit{Biometrika}, \textbf{95}, 481--488.

\bibitem[{Hastie and Tibshirani(1993)}]{hastie1993varying}
Hastie, T. and Tibshirani, R. (1993) Varying-coefficient models.
\newblock \textit{Journal of the Royal Statistical Society: Series B
  (Methodological)}, \textbf{55}, 757--779.

\bibitem[{Hern{\'a}n and Robins(2010)}]{hernan2010causal}
Hern{\'a}n, M.~A. and Robins, J.~M. (2010) Causal inference.

\bibitem[{Hill(2011)}]{hill2011bayesian}
Hill, J.~L. (2011) Bayesian nonparametric modeling for causal inference.
\newblock \textit{Journal of Computational and Graphical Statistics},
  \textbf{20}, 217--240.

\bibitem[{Hu et~al.(2021)Hu, Ji and Li}]{hu2021estimating}
Hu, L., Ji, J. and Li, F. (2021) Estimating heterogeneous survival treatment
  effect in observational data using machine learning.
\newblock \textit{Statistics in Medicine}, \textbf{00}, 1--23.

\bibitem[{Imai et~al.(2013)Imai, Ratkovic et~al.}]{imai2013estimating}
Imai, K., Ratkovic, M. et~al. (2013) Estimating treatment effect heterogeneity
  in randomized program evaluation.
\newblock \textit{The Annals of Applied Statistics}, \textbf{7}, 443--470.

\bibitem[{Imbens(2000)}]{imbens2000role}
Imbens, G.~W. (2000) The role of the propensity score in estimating
  dose-response functions.
\newblock \textit{Biometrika}, \textbf{87}, 706--710.

\bibitem[{Imbens and Rubin(2015)}]{imbens2015causal}
Imbens, G.~W. and Rubin, D.~B. (2015) \textit{Causal inference for statistics,
  social, and biomedical sciences}.
\newblock Cambridge University Press.

\bibitem[{K{\"u}nzel et~al.(2019)K{\"u}nzel, Sekhon, Bickel and
  Yu}]{Kunzel4156}
K{\"u}nzel, S.~R., Sekhon, J.~S., Bickel, P.~J. and Yu, B. (2019) Metalearners
  for estimating heterogeneous treatment effects using machine learning.
\newblock \textit{Proceedings of the National Academy of Sciences},
  \textbf{116}, 4156--4165.
\newblock \urlprefix\url{https://www.pnas.org/content/116/10/4156}.

\bibitem[{K{\"u}nzel et~al.(2018)K{\"u}nzel, Stadie, Vemuri, Ramakrishnan,
  Sekhon and Abbeel}]{kunzel2018transfer}
K{\"u}nzel, S.~R., Stadie, B.~C., Vemuri, N., Ramakrishnan, V., Sekhon, J.~S.
  and Abbeel, P. (2018) Transfer learning for estimating causal effects using
  neural networks.
\newblock \textit{arXiv preprint arXiv:1808.07804}.

\bibitem[{Lesko(2007)}]{lesko2007personalized}
Lesko, L. (2007) Personalized medicine: elusive dream or imminent reality?
\newblock \textit{Clinical Pharmacology and Therapeutics}, \textbf{81},
  807--816.

\bibitem[{Lin et~al.(2013)Lin, Logan and Henley}]{lin2013bias}
Lin, N.~X., Logan, S. and Henley, W.~E. (2013) Bias and sensitivity analysis
  when estimating treatment effects from the cox model with omitted covariates.
\newblock \textit{Biometrics}, \textbf{69}, 850--860.

\bibitem[{Low et~al.(2016)Low, Gallego and Shah}]{low2016comparing}
Low, Y.~S., Gallego, B. and Shah, N.~H. (2016) Comparing high-dimensional
  confounder control methods for rapid cohort studies from electronic health
  records.
\newblock \textit{Journal of comparative effectiveness research}, \textbf{5},
  179--192.

\bibitem[{Lu et~al.(2018)Lu, Sadiq, Feaster and Ishwaran}]{lu2018estimating}
Lu, M., Sadiq, S., Feaster, D.~J. and Ishwaran, H. (2018) Estimating individual
  treatment effect in observational data using random forest methods.
\newblock \textit{Journal of Computational and Graphical Statistics},
  \textbf{27}, 209--219.

\bibitem[{Miettinen and Cook(1981)}]{miettinen1981confounding}
Miettinen, O.~S. and Cook, E.~F. (1981) Confounding: essence and detection.
\newblock \textit{American journal of epidemiology}, \textbf{114}, 593--603.

\bibitem[{Murphy et~al.(2016)Murphy, Redding and Twyman}]{murphy2016handbook}
Murphy, M., Redding, S. and Twyman, J. (2016) \textit{Handbook on personalized
  learning for states, districts, and schools}.
\newblock Center for Innovations in Learning, Philadelphia, PA.

\bibitem[{Newey(1994)}]{newey1994asymptotic}
Newey, W.~K. (1994) The asymptotic variance of semiparametric estimators.
\newblock \textit{Econometrica: Journal of the Econometric Society},
  1349--1382.

\bibitem[{Neyman(1959)}]{neyman1959optimal}
Neyman, J. (1959) Optimal asymptotic tests of composite statistical hypotheses.
\newblock \textit{Probability and statistics}, \textbf{57}, 213.

\bibitem[{Nie and Wager(2017)}]{nie2017quasi}
Nie, X. and Wager, S. (2017) Quasi-oracle estimation of heterogeneous treatment
  effects.
\newblock \textit{arXiv preprint arXiv:1712.04912}.

\bibitem[{Powers et~al.(2018)Powers, Qian, Jung, Schuler, Shah, Hastie and
  Tibshirani}]{powers2018some}
Powers, S., Qian, J., Jung, K., Schuler, A., Shah, N.~H., Hastie, T. and
  Tibshirani, R. (2018) Some methods for heterogeneous treatment effect
  estimation in high dimensions.
\newblock \textit{Statistics in medicine}, \textbf{37}, 1767--1787.

\bibitem[{Robinson(1988)}]{robinson1988root}
Robinson, P.~M. (1988) Root-n-consistent semiparametric regression.
\newblock \textit{Econometrica: Journal of the Econometric Society}, 931--954.

\bibitem[{Rosenbaum et~al.(2010)}]{rosenbaum2010design}
Rosenbaum, P.~R. et~al. (2010) \textit{Design of observational studies},
  vol.~10.
\newblock Springer.

\bibitem[{Rothman(2012)}]{rothman2012epidemiology}
Rothman, K.~J. (2012) \textit{Epidemiology: an introduction}.
\newblock Oxford university press.

\bibitem[{Rubin(1974)}]{rubin1974estimating}
Rubin, D.~B. (1974) Estimating causal effects of treatments in randomized and
  nonrandomized studies.
\newblock \textit{Journal of Educational Psychology}, \textbf{66}, 688--701.

\bibitem[{Samuels(1981)}]{samuels1981matching}
Samuels, M.~L. (1981) Matching and design efficiency in epidemiological
  studies.
\newblock \textit{Biometrika}, \textbf{68}, 577--588.

\bibitem[{Shalit et~al.(2017)Shalit, Johansson and
  Sontag}]{shalit2017estimating}
Shalit, U., Johansson, F.~D. and Sontag, D. (2017) Estimating individual
  treatment effect: generalization bounds and algorithms.
\newblock In \textit{International Conference on Machine Learning}, 3076--3085.
  PMLR.

\bibitem[{Splawa-Neyman et~al.(1990)Splawa-Neyman, Dabrowska and
  Speed}]{splawa1990application}
Splawa-Neyman, J., Dabrowska, D. and Speed, T. (1990) On the application of
  probability theory to agricultural experiments. essay on principles, section
  9.
\newblock \textit{Statistical Science}, 465--472.

\bibitem[{Su et~al.(2009)Su, Tsai, Wang, Nickerson and Li}]{su2009subgroup}
Su, X., Tsai, C.-L., Wang, H., Nickerson, D.~M. and Li, B. (2009) Subgroup
  analysis via recursive partitioning.
\newblock \textit{Journal of Machine Learning Research}, \textbf{10}, 141--158.

\bibitem[{Tian et~al.(2012)Tian, Alizadeh, Gentles and
  Tibshirani}]{tian2012simple}
Tian, L., Alizadeh, A., Gentles, A. and Tibshirani, R. (2012) A simple method
  for detecting interactions between a treatment and a large number of
  covariates. submitted on dec 2012.
\newblock \textit{arXiv preprint arXiv:1212.2995}.

\bibitem[{Tsiatis(1981)}]{tsiatis1981large}
Tsiatis, A.~A. (1981) A large-sample study of cox's regression model.
\newblock \textit{The Annals of Statistics}, \textbf{9}, 93--108.

\bibitem[{Van~der Vaart(1998)}]{van2000asymptotic}
Van~der Vaart, A. (1998) \textit{Asymptotic statistics}.
\newblock Cambridge university press.

\bibitem[{Vansteelandt and Daniel(2014)}]{vansteelandt2014regression}
Vansteelandt, S. and Daniel, R.~M. (2014) On regression adjustment for the
  propensity score.
\newblock \textit{Statistics in medicine}, \textbf{33}, 4053--4072.

\bibitem[{Wager and Athey(2018)}]{wager2018estimation}
Wager, S. and Athey, S. (2018) Estimation and inference of heterogeneous
  treatment effects using random forests.
\newblock \textit{Journal of the American Statistical Association},
  \textbf{113}, 1228--1242.

\bibitem[{Weisberg and Pontes(2015)}]{weisberg2015post}
Weisberg, H.~I. and Pontes, V.~P. (2015) Post hoc subgroups in clinical trials:
  Anathema or analytics?
\newblock \textit{Clinical trials}, \textbf{12}, 357--364.

\bibitem[{Wendling et~al.(2018)Wendling, Jung, Callahan, Schuler, Shah and
  Gallego}]{wendling2018comparing}
Wendling, T., Jung, K., Callahan, A., Schuler, A., Shah, N. and Gallego, B.
  (2018) Comparing methods for estimation of heterogeneous treatment effects
  using observational data from health care databases.
\newblock \textit{Statistics in medicine}, \textbf{37}, 3309--3324.

\end{thebibliography}

\end{document}